%% file: main.tex
\documentclass[11pt,letterpaper]{article}

\usepackage[bookmarks,colorlinks,breaklinks]{hyperref}
\hypersetup{urlcolor=blue, colorlinks=true, citecolor=green!50!black, linkcolor=blue}

\ifdefined\DEBUG
\usepackage{showlabels}
\def\danupon#1{\textcolor{orange}{DN: #1}}
\def\sayan#1{\textcolor{blue}{SB: #1}}
\def\monika#1{\textcolor{red}{MH: #1}}
\newcommand{\xiaowei}[1]{\textcolor{red}{#1}}

\else
\def\danupon#1{}
\def\sayan#1{}
\def\monika#1{}
\newcommand{\xiaowei}[1]{}

\fi

\usepackage[margin=1in]{geometry}
\usepackage{theorem,latexsym,graphicx,amssymb}
\usepackage{amsmath,enumerate}
\usepackage{float}
\usepackage{xspace}
\usepackage{paralist}
\usepackage{enumerate}
\usepackage{cases}
\usepackage{caption}
\usepackage{algorithm}
\usepackage{algorithmicx}
\usepackage[noend]{algpseudocode}
\usepackage{xcolor}
\usepackage{multicol}
\usepackage{graphicx}
\usepackage{subcaption}
\usepackage{varwidth}
\usepackage{complexity}
\usepackage[normalem]{ulem}
\usepackage{enumitem}

\newenvironment{proof}{{\bf Proof:  }}{\hfill\rule{2mm}{2mm}\vspace*{5pt}}

\numberwithin{figure}{section}
\numberwithin{equation}{section}
\newtheorem{definition}{Definition}[section]
\newtheorem{property}{Property}[section]

\newtheorem{corollary}{Corollary}[section]
\newtheorem{theorem}{Theorem}[section]
\newtheorem{lemma}{Lemma}[section]
\newtheorem{claim}{Claim}[section]

\newtheorem{observation}{Observation}[section]
\newtheorem{invariant}{Invariant}[section]

\newtheorem{assumption}{Assumption}[section]

\renewcommand{\poly}{\operatorname{poly}}

\newcommand{\opt}{\mathsf{OPT}}

\title{Dynamic Set Cover: Improved Amortized and Worst-Case Update Time}

\author{Sayan Bhattacharya\thanks{University of Warwick, UK. Email: {\tt S.Bhattacharya@warwick.ac.uk}} \and Monika Henzinger\thanks{University of Vienna, Austria. Email: {\tt monika.henzinger@univie.ac.at}} \and Danupon Nanongkai\thanks{KTH, Stockholm, Sweden. Email: {\tt danupon@kth.se }} \and Xiaowei Wu\thanks{IOTSC, University of Macau, China. Email: {\tt xiaoweiwu@um.edu.mo}. This work was done in part while the author was a postdoc at the University of Vienna.}}
\date{}

\begin{document}

\renewcommand{\S}{\mathcal{S}}
\newcommand{\U}{\mathcal{U}}

\begin{titlepage}
	\maketitle
	\pagenumbering{roman}
	\begin{abstract}

\input{abstract}

	\end{abstract}
	\newpage
	\setcounter{tocdepth}{2}
	\tableofcontents
\end{titlepage}

\newpage

\part{EXTENDED ABSTRACT}

\newpage
\pagenumbering{arabic}

\input{intro}

\input{overview}

\newpage

\part{FULL VERSION}


\appendix

\input{Preliminaries}

\input{Algorithm_Amortized}

\input{Potential_Amortized}

\input{Algorithm_Worst_Case}

\input{Rebuild_Worst_Case}

\paragraph{Acknowledgements:} 
The project has received funding from the Engineering and Physical Sciences Research Council, UK (EPSRC) under Grant Ref: EP/S03353X/1.

The research leading to these results has received funding from the European Research Council under the
European Union's Seventh Framework Programme (FP/2007-2013) / ERC Grant Agreement no. 340506. 

This project has received funding from the European Research Council (ERC) under the European Union's Horizon 2020 research and innovation programme under grant agreement No 715672. Nanongkai
was also supported by the Swedish Research Council (Reg. No. 2015-04659).

Funded by The Science and Technology Development Fund, Macau SAR (File no. SKL-IOTSC-2018-2020), the Start-up Research Grant of University of Macau (File no. SRG2020-00020-IOTSC).

{
	\bibliography{setcover}
	\bibliographystyle{alpha}
}

\end{document}

%% file: abstract.tex

In the dynamic minimum set cover problem, a challenge is to minimize the update time while guaranteeing close to the optimal $\min(O(\log n), f)$ approximation factor. (Throughout,  $m$, $n$, $f$, and $C$ are parameters denoting the maximum number of sets, number of elements, frequency, and the cost range.)  In the {\em high-frequency} range, when $f=\Omega(\log n)$, this was achieved by a deterministic $O(\log n)$-approximation algorithm with  $O(f \log n)$ amortized update time [Gupta~et~al. STOC'17]. In the {\em low-frequency} range,  the line of work by Gupta et al.~[STOC'17], Abboud et al.~[STOC'19], and  Bhattacharya et al.~[ICALP'15, IPCO'17, FOCS'19] led to  a deterministic $(1+\epsilon)f$-approximation algorithm with $O(f \log (Cn)/\epsilon^2)$ amortized update time. In this paper we improve the latter update time and provide the first bounds that subsume (and sometimes improve) the state-of-the-art dynamic vertex cover algorithms. We obtain:

{\bf 1. $(1+\epsilon)f$-approximation ratio in $O(f\log^2 (Cn)/\epsilon^3)$  worst-case update time:} No non-trivial worst-case update time was previously known for dynamic set cover. Our bound subsumes and improves by a logarithmic factor the $O(\log^3 n/\poly(\epsilon))$ worst-case update time for unweighted dynamic vertex cover (i.e., when $f=2$ and $C=1$) by  Bhattacharya et al.~[SODA'17].

{\bf 2. $(1+\epsilon)f$-approximation ratio in	$O\left((f^2/\epsilon^3)+(f/\epsilon^2) \log C\right)$ amortized update time:}
This result improves the previous $O(f \log (Cn)/\epsilon^2)$ update time bound for most values of $f$ in the low-frequency range, i.e. whenever $f=o(\log n)$. It is the first that is independent of $m$ and $n$. It subsumes the constant amortized update time of Bhattacharya and Kulkarni [SODA'19] for unweighted dynamic vertex cover (i.e., when $f = 2$ and $C = 1$).

These results are achieved by leveraging the {\em approximate complementary  slackness} and {\em background schedulers} techniques. These techniques were used in the {\em local} update scheme for dynamic vertex cover. Our main technical contribution is to adapt these techniques within the {\em global} update scheme of Bhattacharya~et~al. [FOCS'19] for the dynamic set cover problem.

%% file: intro.tex
\section{Introduction}
\label{main:sec:intro}
\label{sec:perspective}

In the {\em minimum set cover} problem, we get a universe of elements $\U$ and a collection of sets $\S \subseteq 2^{\U}$  as input, where $\bigcup_{s \in \S} s = \U$ and each set $s \in \S$ has a {\em cost} $c_s > 0$ associated with it. A collection of sets $\S' \subseteq \S$ forms a {\em set-cover} of $\U$ iff $\bigcup_{s \in \S'} s = \U$. The goal is to compute a set cover $\S'$ of $\U$ with minimum total cost $c(\S') = \sum_{s \in \S'} c_s$. This is one of the most fundamental problems in approximation algorithms. In recent years, this problem has also received significant attention in the {\em dynamic setting}, where the input keeps changing over time. Specifically, here we want to design a {\em dynamic algorithm} for minimum set cover that can handle the following operations:

\smallskip
\noindent {\em Preprocessing:} Initially, the algorithm receives as input a universe of elements $\U$, a collection of sets $\S \subseteq 2^{\U}$ with $\bigcup_{s \in \S} s = \U$, and a cost $c_s \geq 0$ for each set $s \in \S$.  

\smallskip
\noindent {\em Updates:} Subsequently, the input keeps changing via a sequence of updates, where each update either (1) deletes an element $e$ from the universe $\U$ and from every set $s \in \S$ that contains $e$, or (2) inserts an element $e$ into the universe $\U$ and specifies the sets in $\S$ that the element $e$ belongs to.

\smallskip
After each update, we would like to maintain an approximate cost of the optimal set cover of the updated set system. (Some algorithms also allow accessing a solution with such cost. See the remark after Theorem~\ref{th:main:intro:both results}.) 
\sayan{I removed this line: The time taken by a dynamic algorithm to handle the preprocessing step is referred to as its {\em preprocessing time}.}
A dynamic algorithm has an {\em amortized update time} of $O(t)$ iff it takes $O((\alpha + \beta) \cdot t)$ total time (including the time spent on preprocessing) to handle any sequence of $\beta \geq 1$ updates, where $\alpha$ is the number of elements being preprocessed.
We want to design a dynamic algorithm with small approximation ratio and  update time.
We get two main results:




\begin{theorem}
	\label{th:main:intro:both results}\label{th:main:intro:amortized time}
	There are deterministic dynamic algorithms for the minimum set cover problem with $(1+\epsilon)f$-approximation ratio and 
\begin{enumerate}
	\item a worst-case update time of  $O(f\log^2 (Cn)/\epsilon^3)$, and 
	\item an amortized update time of $O\left(\frac{f^2}{\epsilon^3}+\frac{f}{\epsilon^2}\log C\right)$. 
\end{enumerate}	
 Here, the symbol $f$ denotes an upper bound on the maximum frequency of any element across all the updates,\footnote{Frequency of an element $e \in \U$ is defined as the number of sets in $\S$ that contain $e$.},  $C \geq 1$ is a parameter such that $1/C \leq c_s \leq 1$ for all sets $s \in \S$, $m$ is the number of sets in $\S$, and $n$ is the maximum number of elements in the universe $\U$ across all the updates. 
\end{theorem}

%
\noindent {\bf Remark:} Both our algorithms maintain an $(1+\epsilon)f$-approximation to the {\em cost} of the minimum set cover after every update and can return this value in constant time.
In addition, the algorithm for amortized update time  maintains a \emph{solution} of such value throughout the updates (i.e. it outputs the change of the maintained solution after every update), while the algorithm for worst-case update time does not and instead outputs the whole solution in time linear to the solution size whenever the solution is asked for (similar to the dynamic matching algorithm in \cite{BernsteinFH-soda19}).



\smallskip
\noindent {\bf Perspective:}
The minimum set cover problem is very well understood in the static setting. There is a simple primal-dual algorithm that gives an $f$-approximation in $\Theta(f  n)$ time, whereas  a simple greedy algorithm gives a $\Theta(\log n)$-approximation in $\Theta(f  n)$ time. Furthermore, there are strong inapproximability results which imply that  the approximation guarantees achieved by these simple primal-dual and greedy algorithms are essentially the best possible~\cite{stoc/DinurS14,siamcomp/DinurGKR05,ccc/KhotR03}.

\begin{table}
	\centering
	\begin{centering}
		\begin{tabular}{|c|c|c|c|c|}
			\hline 
			\textbf{Reference} & \textbf{Approximation} & \textbf{Update Time} & \textbf{Deterministic?} & \textbf{Weighted?} \tabularnewline
			\hline  
			\cite{GuptaKKP17} & $O(\log n)$	&	$O(f \log n)$ & yes & yes\\
			\hline 
			\cite{GuptaKKP17,BhattacharyaCH17} & $O(f^3)$	&	$O(f^2)$ & yes & yes\\
			\hline
			\cite{BhattacharyaHI15} & $O(f^2)$ & $O(f \log(m + n))$ & yes & yes\\
			\hline
			\cite{AbboudAGPS19} & $(1+\epsilon)f$ & $O\left(\frac{f^2}{\epsilon} \log n \right)$ & no & no\\ 
			\hline
			\cite{focs/BhattacharyaHN19} &  $(1+\epsilon)f$ & $O\left(\frac{f}{\epsilon^2} \log (Cn)\right)$ & yes & yes\\
			\hline
			{\small Our result (amortized)} & $(1+\epsilon)f$ & $O\left(\frac{f^2}{\epsilon^3}+ \frac{f}{\epsilon^2} \log C\right)$ & yes & yes\\
			\hline
			{\small Our result (worst case)} &  $(1+\epsilon)f$ & $O(f\log^2 (Cn)/\epsilon^3)$ & yes & yes\\
			\hline
		\end{tabular}
	\end{centering}
	\caption{Summary of known results on dynamic set cover. All the previous update times are amortized. The last column indicates if the result holds when different sets have different costs.}
	\label{tab:1}
\end{table}

In the dynamic setting, an important challenge is to match the approximation ratio of the (static) greedy and primal-dual algorithms, while minimizing the update time. In recent years, a series of papers on dynamic algorithms have been devoted to this topic. See Table~\ref{tab:1} for a concise description of the results obtained in these papers.  To summarize, we currently know how to get a $\Theta(\log n)$-approximation in $O(f \log n)$ update time, and a $(1+\epsilon)f$-approximation in $O(f \log (Cn)/\epsilon^2)$ update time. 
In addition, there is a strong conditional lower bound~\cite{AbboudAGPS19} which states that any  dynamic set cover algorithm with nontrivial approximation ratio must have an update time of $\Omega(f^{1-\delta})$, for any constant $\delta > 0$.  
This explains the $O(\poly (f))$ factor inherent in all the update time bounds of Table~\ref{tab:1}, but leaves open the following question.
%

\smallskip
\noindent {\bf (Question 1) Must we necessarily incur a $\polylog(m,n)$ factor in the update time if we want to aim for near-optimal approximation ratio?} 
%
%

\smallskip
The above question falls within the study of {\em constant update time} (see below). Besides helping us understand the best possible update time for a dynamic problem to its limit, this study is useful in ruling out non-trivial cell-probe lower bounds \cite{PatrascuD06,Larsen12,LarsenWY18}. Another important line of work in dynamic graph algorithms is achieving {\em worst-case} update time. All previous dynamic set cover algorithms can guarantee only {\em amortized} update time, leaving it widely open the following. 
%

\smallskip
\noindent
{\bf (Question 2): Is there a dynamic  algorithm with non-trivial worst-case update time?} 
%
%

\smallskip
When $f=2$ and $C=1$, the above questions are equivalent to asking whether there are $2$-approximation algorithms for {\em dynamic (unweighted) vertex cover} with (i) {constant} update time and (ii) non-trivial worst-case update time. There exists a long line of work on this dynamic (unweighted) vertex cover problem~\cite{OnakR10,BaswanaGS11,GuptaP13,NeimanS13,soda/BhattacharyaHI15,PelegS16,Solomon16,stoc/BhattacharyaHN16,soda/BhattacharyaK19}. Currently, the state of the art results on this problem are as follows.
\begin{itemize}[noitemsep]
	\item The deterministic  algorithm  of~\cite{soda/BhattacharyaK19}  achieves $(2+\epsilon)$-approximation in $O(1/\epsilon^2)$ amortized update time for unweighted vertex cover, and the  randomized algorithm of~\cite{Solomon16} achieves $2$-approximation in $O(1)$ amortized update time for unweighted vertex cover. 
	\item  The deterministic algorithm of~\cite{soda/BhattacharyaHN17} achieves $(2+\epsilon)$-approximation  in $O(\log^3 n/\poly(\epsilon))$ worst-case update time for unweighted vertex cover (also see \cite{icalp/CharikarS18,icalp/ArarCCSW18,soda/BernsteinFH19}).  
\end{itemize}
%
%

Our $O(f\log^2 (Cn))$ worst-case bound in Theorem~\ref{th:main:intro:both results},  when restricted to unweighted vertex cover, improves the  $O(\log^3 n)$ bound of \cite{soda/BhattacharyaHN17} by a logarithmic factor. Moreover, ours is the first non-trivial worst-case update time that holds for $f>2$. On the other hand, our amortized bound in Theorem~\ref{th:main:intro:both results} is the first generalization of the vertex cover results from \cite{Solomon16,soda/BhattacharyaK19}: 
When $f > 2$ and $C > 1$, a possible generalization  of the constant amortized update time obtained in~\cite{soda/BhattacharyaK19,Solomon16} is the one guaranteeing $(1+\epsilon)f$-approximation ratio and $O(\poly(f, C))$ update time. The only previous result of this kind is the $O(f^2)$ update time achieved by~\cite{GuptaKKP17,BhattacharyaCH17}; however this comes with a higher approximation ratio of $O(f^3)$. Our amortized update time is the first to achieve the target  $O(\poly (f,C))$ bound  {\em simultaneously} with a  $(1+\epsilon)f$-approximation ratio.

Finally, note that our amortized update time improves the previous one in \cite{focs/BhattacharyaHN19} in almost the whole range of parameters that we should be interested in: 
%
%
For a fixed $\epsilon > 0$, we get an update time of $O(f^2 + f \log C)$, whereas \cite{focs/BhattacharyaHN19} obtained an update time of $O(f \log (Cn))$. Note that in the {\em high-frequency range}, when $f = \omega(\log n)$, the $\Theta(\log n)$-approximation ratio obtained by \cite{GuptaKKP17} is already better than a $(1+\epsilon)f$-approximation. In other words, we are typically interested in getting an $(1+\epsilon)f$-approximation only in the {\em low-frequency range}, when $f=O(\log n)$. In this regime, our $O(f^2 + f \log C)$ update time strictly improves upon the update time of \cite{focs/BhattacharyaHN19} for most values of $f$, i.e. whenever $f=o(\log n)$.


\subsection{Techniques}

Both our results build on the recent algorithm of Bhattacharya, Henzinger and Nanongkai~\cite{focs/BhattacharyaHN19}. This algorithm and most previous deterministic algorithms for dynamic set cover and vertex cover (e.g. \cite{soda/BhattacharyaK19,soda/BhattacharyaHI15,stoc/BhattacharyaHN16,soda/BhattacharyaHN17}) are based on the following static primal-dual algorithm. (For ease of exposition, in this section we assume that $C = 1$; i.e., every set has the same cost.)

The static primal-dual algorithm assigns a fractional {\em weight} $w_e \geq 0$ to every element $e \in \U$, as follows. Initially, we set $w_e \leftarrow 0$ for all elements $e\in\U$ and $F \leftarrow \U$. Subsequently, the algorithm proceeds in rounds. In each round, we continuously raise the weights of all the elements in $F$ until some set $s \in \S$ becomes {\em tight} (a set $s$ becomes tight when its total weight $w_s = \sum_{e \in s} w_e$ becomes equal to $1$). At this point, we delete the elements contained in the newly tight sets from $F$, and after that we proceed to the next round. The process stops when $F$ becomes empty. At that point, we return the collection of tight sets as a valid set cover and the weights $\{w_e\}$ as the dual certificate. Specifically, it turns out that the weights $\{w_e\}$ returned at the end of the algorithm form a valid solution to the dual {\em fractional packing problem}, which asks us to assign a weight $w_e \geq 0$ to each element in $\U$ so as to maximize the objective $\sum_{e\in \U} w_e$, subject to the constraint that $\sum_{e \in s} w_e \leq 1$ for all sets $s \in \S$. We can also show that the collection of tight sets returned at the end of the static algorithm forms a valid set cover, whose cost is at most $f$ times the cost of the dual objective $\sum_{e \in \U} w_e$. This leads to an approximation guarantee of $f$. In the dynamic setting, the main challenge now is to maintain the (approximate) output of the static algorithm described above in small update time. This is where \cite{focs/BhattacharyaHN19} and previous deterministic algorithms use radically different approaches.

More specifically, previous deterministic algorithms (e.g. \cite{soda/BhattacharyaK19,soda/BhattacharyaHI15,stoc/BhattacharyaHN16,soda/BhattacharyaHN17}) follow some {\em local update rules} and maintain the approximate complementary slackness conditions. Thus, whenever the weight $w_s$ of a tight set $s \in \S$ becomes  too large (resp. too small) compared to $1$, these algorithm decrease (resp. increase) the weights of some of the elements contained in $s$. This step affects the weights of some other sets that share these elements with $s$, and hence it might lead to a chain of cascading effects. Using very carefully chosen potential functions, these algorithms are able to bound these cascading effects over any sufficiently long sequence of updates. 
For technical reasons, however, this approach seems to work only when $f = 2$.
Thus, although previous works could get constant and worst-case update time for maintaining a $(2+\epsilon)$-approximate vertex cover, it seems very difficult to extend their potential function analysis to the more general minimum set cover problem (or, equivalently, to minimum vertex cover on hypergraphs).

In contrast,~\cite{focs/BhattacharyaHN19} makes no attempt at maintaining the approximate complementary slackness conditions. It simply waits until the overall cost of the dual solution changes by a significant amount (compared to the cost of the set cover maintained by the algorithm). (This approach shares some similarities with the earlier randomized algorithm by \cite{AbboudAGPS19}, although \cite{AbboudAGPS19} is not based on the static algorithm described above.)  
At that point, the algorithm identifies a critical collection of affected elements and recomputes their weights from scratch using a {\em global rebuilding} subroutine. The time taken for this recomputation step is, roughly speaking, proportional to the number of critically affected elements, which   leads to a bound on the amortized update time. The strength of this framework is that this global rebuilding strategy extends seamlessly to the general set cover problem (where $f > 2$). Unfortunately this strategy incurs an additional $\Theta(\log n)$ factor in the update time that seems impossible to overcome. 

Our algorithm with amortized update time results from carefully combining the these two sharply different approaches, namely the algorithm of  \cite{soda/BhattacharyaK19} that uses some local update rules to obtain an $O(1/\epsilon^2)$ amortized update time, and the new approach of \cite{focs/BhattacharyaHN19}. 
%
In our {\em hybrid} approach, whenever the weight $w_s$ of a tight set $s$ becomes too large compared to $1$, we decrease the weights of some of the elements contained in $s$ using the same local rule as in~\cite{soda/BhattacharyaK19}. In contrast, whenever the weight $w_s$ of a tight set $s$ becomes too small compared to $1$, we follow a lazy strategy and try to wait it out. After some period of time, when the total cost of the dual solution becomes significantly small compared to the size of the set cover maintained by the algorithm, we apply a {\em global rebuilding} subroutine as in~\cite{focs/BhattacharyaHN19} to fix the weights of some critical elements. This hybrid approach allows us to combine the best of both worlds, leading to a dynamic algorithm that has $(1+\epsilon)f$-approximation ratio for any  $f \geq 2$ and an amortized update time of $O(f^2/\epsilon^3)$. 

Our algorithm with worst-case update time extends the approach of \cite{focs/BhattacharyaHN19} by having many {\em schedulers} working in parallel. This general idea has been used in many dynamic algorithms with worst-case update time (e.g. \cite{siamcomp/ChanPR11,GuptaP13,CharikarS18,NanongkaiSW17,NanongkaiS17,Wulff-Nilsen17}). The main challenge is typically how to make the schedulers {\em consistent} in what they maintain, especially if they maintain an overlapping part of the solution.
%
More specifically, we have $k=O(\log n)$ schedulers, where the $i^{th}$ scheduler is associated with an integer $r_i$ such that $r_1\geq r_2 \geq \ldots \geq r_k$. The $i^{th}$ scheduler is responsible for running the global rebuilding subroutine of \cite{focs/BhattacharyaHN19} on sets that get tight {\em at and after round $r_i$} in the static algorithm described above. Thus, the sets that the $i^{th}$ scheduler is responsible for are also under the responsibilities of the $j^{th}$ schedulers for all $j>i$. A complication arises when these schedulers want to rebuild these sets at the same time, since it is not clear which solution of which scheduler we should use as a final solution. Typically, this can be resolved by forcing all schedulers to be {\em consistent}; i.e. the $i^{th}$ and $j^{th}$ schedulers agree on what happens to each set that they are both responsible for. This seems very hard to achieve in our case. 
At a high level, we get around this issue by requiring the schedulers to be only {\em loosely consistent}: Schedulers may maintain drastically different {\em local views} on the sets they are responsible for, except that there are some mild consistency conditions that tie their behaviors together. This way, each scheduler can work independently while our conditions guarantee that we can still combine results from the schedulers when needed.  More specifically,  we use the solution that the $i^{th}$ scheduler maintains for level $r_i.$ Due to the consistency conditions this results in an approximately  minimum set cover.

%% file: overview.tex
\section{Preliminaries: A Static Primal-Dual Algorithm} \label{main:sec:prelim}



\noindent {\bf Uniform-cost case:}
Recall the  notations defined in the beginning of Section~\ref{main:sec:intro}. In order to highlight the main ideas behind our algorithms, in this extended abstract we  only consider the special case where every set has the same cost ($c_s = c_{s'}$ for all $s, s' \in \S$) and our goal is to compute a set cover in $(\U, \S)$ of minimum size. The full version of the paper is presented in the appendix. 


\smallskip
\noindent {\bf The dual:} In the  {\em maximum fractional packing} problem, we get a set system $(\U, \S)$ as input. We have to assign a weight $w_e \in [0,1]$ to every element $e \in \U$, subject to the constraint that $\sum_{e \in s} w_e \leq 1$ for all sets $s \in \S$. The goal is to maximize  $\sum_{e \in \U} w_e$. We let $w_s = \sum_{e \in s} w_e$ denote the total weight received by a set $s \in \S$. LP-duality  implies the following lemma.

\begin{lemma}
\label{main:lm:dual}
Consider any instance $(\U, \S)$ of the  set cover problem. Let $\opt$ denote the size of the minimum set cover on this instance, and let $\{w_e\}_{e \in \U}$ denote any feasible fractional packing solution on the same input instance. Then we have $\sum_{e \in \U} w_e \leq \opt$. 
\end{lemma}


\newcommand{\T}{\mathcal{T}}
\renewcommand{\A}{\mathcal{A}}
\newcommand{\F}{\mathcal{F}}

We now describe an $f$-approximation  algorithm for minimum set cover  in the static setting. The algorithm works as follows. There is a {\em time-variable} $t$ that is initially set to $-\infty$, and it keeps increasing continuously throughout the duration of the algorithm. At every time $t$, the algorithm maintains a partition of the universe of elements $\U$ into two subsets: $\A(t) \subseteq \U$ and $\F(t) = \U \setminus \A(t)$. The elements in $\A(t)$ and $\F(t)$ are respectively called {\em alive} and {\em frozen} at time $t$. In the beginning, we have $\A(-\infty) = \U$ and $\F(-\infty) = \emptyset$. As $t$ increases starting from $-\infty$, alive elements  become frozen one after the other. Specifically, we have $\A(t) \supseteq \A(t')$  and $\F(t) \subseteq \F(t')$ for all $t \leq t'$.
Let $t_e$ be the time at which  an element $e \in \U$ becomes frozen (i.e., moves from  $\A$ to $\F$). We refer to $t_e$ as the {\em freezing time} of $e$.  Note that  $t_e \leq t$ for all $e \in \F(t)$.  At any time $t$,  the weight  of an element $e \in \U$ is determined as follows: $\text{If } e \in \A(t), \text{ then } w_e(t) = (1+\epsilon)^{t}. \text{ Otherwise, } e \in \F(t) \text{ and } w_e(t) = (1+\epsilon)^{t_e}.$ 
 
 Let $w_s(t) = \sum_{e \in s} w_e(t)$ denote the weight of a set $s \in \S$ at time $t$. We say that the set $s$ is {\em tight} (resp. {\em slack}) at time $t$ if $w_s(t) = 1$ (resp. $w_s(t) < 1$). Let $\T(t) \subseteq \S$ denote the collection of tight sets at time $t$.  When $t = -\infty$, we have $w_e(t) = (1+\epsilon)^{t} = 0$ for all elements $e \in \A(t) = \U$, and hence $w_s(t) = 0$ for all sets $s \in \S$. This implies that $\T(-\infty) = \emptyset$.  Now, the following invariant completes the description of the algorithm: At any time $t$, we have $\F(t) = \bigcup_{s \in \T(t)} s$. 

To summarize, the algorithm starts at time $t = -\infty$. At that point every element is alive and has weight $0$, and all the sets are slack with weight $0$. As $t$ starts increasing continuously, the weights of the alive elements keep increasing according to the equation $w_e(t) = (1+\epsilon)^{t}$. Whenever a set $s$ becomes tight during this process, every alive element $e \in s$ becomes frozen at the same time-instant, which ensures that the weights of all the elements $e \in s$ (and that of the set $s$) do not increase any further as the value of $t$ keeps increasing. The process stops at time $t = 0$. Note that at time $t = 0$, if an element $e$ is alive, then all sets containing $e$ must be tight. This means that any element $e$ has freezing time $t_e \leq 0$, which leads to the following claim.

\begin{claim}
\label{main:cl:stopping}
At time $t = 0$, we have $\A(0) = \emptyset$ and $\F(0) = \U$.
\end{claim}


\noindent {\bf Levels of elements and sets:} Claim~\ref{main:cl:stopping} implies that the continuous process describing the static algorithm ends at time $t = 0$. At that point,  every element $e \in \U = \F(0)$ has a well-defined freezing time $t_e \leq 0$. We define the {\em level} of an element $e \in \U$ to be $\ell(e) = - t_e$. The {\em level} of a set $s \in \T(0)$ is defined as $\ell(s) = -t_s$, where $t_s$ is the time at which the set $s$ became tight. The levels of the remaining sets $s \in \S \setminus \T(0)$ are defined to be $\ell(s) = 0$. 

Henceforth, we use the symbol $w_e$ to denote the weight of an element $e$ at time $t = 0$ (i.e., $w_e = w_e(0)$). Similarly, we use the symbol $w_s$ to denote $w_s(0)$, and the symbol $\T$ to denote $\T(0)$.  Finally, when we say that a set $s$ is tight, we mean that it is tight at time $t = 0$. 


\begin{property}
\label{main:static:property:level}
We have $\ell(e) = \max_{s \in \S : e \in s} \ell(s)$  and $w_e = (1+\epsilon)^{-\ell(e)}$ for all  elements $e \in \U$. 
\end{property}

\begin{property}
\label{main:static:property:tight}
We have $w_s \leq 1$ for all   $s \in \S$. Further, every  set $s \in \S$ at level $\ell(s) > 0$ is tight.
\end{property}



\begin{lemma}
\label{main:cor:static:dual}
\label{main:static:approx}
The weights $\{w_e\}_{e\in \U}$ form a fractional packing and the collection of sets $\T$ forms an $f$-approximate minimum set cover in $(\U, \S)$. 
\end{lemma}
\begin{proof}
Since $w_s \leq 1$ for all $s \in \S$, the weights $\{w_e\}_{e \in \U}$ form a fractional packing. Consider any element $e \in \U$. If at least one set $s \in \S$ containing $e$ lies at level $\ell(s) > 0$, then $s \in \T$ and hence the element $e$ is covered by $\T$. Otherwise,  every set $s \in \S$ containing the element $e$ lies at level $\ell(s) = 0$. So Property~\ref{main:static:property:level} implies that $\ell(e) = 0$, and hence $w_e = (1+\epsilon)^{-0} = 1$. Thus, every set $s$ containing $e$ has weight $w_s \geq w_e = 1$. In other words, every set $s$ containing $e$ is tight, and so the element $e$ is again covered by $\T$. This implies that $\T$ forms a valid set cover. 

Since each element $e$ contributes to the weight $w_s$ of only the (at most $f$) sets that contain it, we have $\sum_{e \in \U} f \cdot w_e \geq \sum_{s \in \S} w_s \geq \sum_{s \in \T} w_s =   |\T|$. The last equality holds since $w_s = 1$ for all sets $s \in \T$. Now,  Lemma~\ref{main:lm:dual}  implies that $f \cdot \opt \geq f \cdot \sum_{e \in \U} w_e \geq |\T|$. 
\end{proof}




\section{Our Algorithm for Amortized Update time: An Overview}
\label{main:sec:algo}





\noindent {\bf Preprocessing:} We start by computing the solution returned by the static algorithm from Section~\ref{main:sec:prelim}. Let $\ell(s)$ be the level of a set $s \in \S$ in the output of this static algorithm. The level $\ell(e)$ and the weight $w_e$  of every element $e \in \U$ are determined by the levels of the sets containing it, in accordance with Property~\ref{main:static:property:level}. The weight of a set $s \in \S$ is defined as $w_s = \sum_{e \in s} w_e$.   In addition, we associate a variable $\phi_s$ with every set $s\in \S$. The value of $\phi_s$ is called  the {\em dead-weight} of $s$. In contrast, the value of $w_s$ denotes the {\em real-weight} of $s$. The {\em total-weight} of a set $s \in \S$ is given by $w^*_s = w_s + \phi_s$. Just after preprocessing, we have $\phi_s = 0$ for all $s \in \S$, so that the total-weight of every set is equal to its dead-weight. Throughout Section~\ref{main:sec:algo}, we will say that a set $s \in \S$ is {\em tight} if $w^*_s = 1$ and {\em slack} if $w^*_s < 1$. Accordingly, the total-weights $\{w^*_s\}_{s \in \S}$ satisfy Property~\ref{main:static:property:tight} just after preprocessing.  The significance of the notion of dead-weights will become clear shortly.



\subsection{Handling deletions of elements}
\label{main:sec:deletions}

When an element $e$ gets deleted, the real-weight $w_s$ of every set $s$ containing $e$ decreases by $w_e$. To compensate for this loss,  we set $\phi_s \leftarrow \phi_s + w_e$ for all sets $s \in \S$ that contained  $e$. Thus, the total-weight of every set remains unchanged due to an element-deletion.  It should now be apparent that our  algorithm satisfies Property~\ref{main:static:property:tight} if we replace the real-weights $w_s$ by the total-weights $w^*_s$. 
\begin{invariant}
\label{main:inv:deletions:weights}
We have $w^*_s \leq 1$ for all   $s \in \S$. Further, every  set $s \in \S$ at level $\ell(s) > 0$ is tight. 
\end{invariant}

As the elements keep getting deleted the size of minimum set cover keeps decreasing. But the set cover maintained by the  algorithm we have described so far remains unchanged. Hence, after sufficiently many deletions, the approximation ratio of our algorithm will degrade by a significant amount. To address this concern, our algorithm {\em rebuilds} part of the solution once the sum of the dead-weights of the sets becomes too large. Specifically, we maintain the following invariant.

\begin{invariant}
\label{main:inv:deletions:rebuild}
We have $\sum_{s \in \S} \phi_s \leq \epsilon \cdot f \cdot \sum_{e \in \U} w_e$.
\end{invariant}

After preprocessing, the above invariant holds since $\phi_s = 0$ for all $s \in \S$. Subsequently, after handling each element-deletion in the manner described above, we perform the following operations. 
\begin{itemize}
\item {\sc While} Invariant~\ref{main:inv:deletions:rebuild} is violated:
\begin{itemize}
\item Identify the smallest level $k \geq 0$ such that $\sum_{s \in \S : \ell(s) \leq k} \phi_s > \epsilon \cdot f \cdot \sum_{e \in \U : \ell(e) \leq k} w_e$. 
\item Call the subroutine {\sc Rebuild}($\leq k$) as described below.
\end{itemize}
\end{itemize}


\noindent{\bf The subroutine {\sc Rebuild}$(\leq k)$:} Let $\S'_k$ (resp.~$\U'_k$) be the collection of sets $s \in \S$ at levels $\ell(s) \leq k$ (resp.~the collection of elements $e \in \U$ at levels $\ell(e) \leq k$) just before the call to  {\sc Rebuild}$(\leq k)$. The subroutine works in two steps: Step I (clean-up) and Step II (rebuild). To simplify the analysis, we make the following crucial assumption in this extended abstract.

\begin{assumption}
\label{main:assume:rebuild}
Every set $s \in \S'_k$ contains at least one element from $\U'_k$. 
\end{assumption}

\noindent {\bf Step I (clean-up):}
We set $\phi_s \leftarrow 0$, $\ell(s) \leftarrow k$  for all $s \in \S'_k$. This  resets $\ell(e) \leftarrow k$ and $w_e \leftarrow (1+\epsilon)^{-k}$ for all $e \in \U'_k$, as per Property~\ref{main:static:property:level}. The real-weights $\{w_s\}_{s \in \S'_k}$  get updated accordingly. 

The clean-up step as described above can only decrease the weight $w_e$ of an element $e \in \U'_k$, since it moves up from its earlier level (which was $\leq k$) to level $k$. Hence, the real-weights $w_s$ of the sets $s \in \S'_k$ can also only decrease due to this step. Furthermore, since $\ell(e) = \max_{s \in \S : e \in s} \ell(s)$ for every element $e$ (see Property~\ref{main:static:property:level}),  all the sets containing an element $e \in \U'_k$ belong to $\S'_k$. Accordingly,  we do not change the real-weight $w_s$  of any set $s \in \S$ at level $\ell(s) > k$ during the clean-up step. Neither  do we change the dead-weight $\phi_s$ of any set $s$  or the level/weight of any element $e$  at level $> k$. Since  Invariant~\ref{main:inv:deletions:weights} was satisfied just before the clean-up step, we get:


\begin{observation}
\label{main:ob:rebuild}
Just after the clean-up step, every set $s \in \S$ at level $\ell(s) > k$ is tight. All the remaining sets $s \in \S'_k$ are at level $\ell(s) = k$ with real-weights $w_s \leq 1$ and dead-weights $\phi_s = 0$. 
\end{observation}

\noindent {\bf Step II (rebuild):} Recall that the static algorithm from Section~\ref{main:sec:prelim} starts at time $t = -\infty$ and stops when $t$ becomes equal to $0$. Observation~\ref{main:ob:rebuild} implies that after the clean-up step, we are in exactly the same state as the static algorithm from Section~\ref{main:sec:prelim} at time $t = -k$ (provided we  replace the real-weights $w_s$ by the total-weights $w^*_s$ for all sets $s$ at level $\ell(s) > k$). At this point, we perform the remaining steps prescribed by the static algorithm from Section~\ref{main:sec:prelim} as its time-variable moves from $t = -k$ to $t = 0$. We emphasize that while executing these remaining steps, we do not change the dead-weights of the sets in $\S'_k$ (these dead-weights were reset to zero during the clean-up step, and they continue to remain zero). This leads us to the following observation.

\begin{observation}
\label{main:ob:rebuild:2}
At the end of the call to the subroutine {\sc Rebuild}$(\leq k)$, Invariant~\ref{main:inv:deletions:weights} is satisfied. Furthermore, we also have $\phi_s = 0$ for all sets $s \in \S$ at levels $\ell(s) \leq k$. 
\end{observation}

We note that  using appropriate data structures this subroutine  can be implemented efficiently.

\begin{lemma}
\label{main:lm:rebuild:runtime} Under Assumption~\ref{main:assume:rebuild}, the subroutine {\sc Rebuild}$(\leq k)$ runs in $O(f \cdot |\U'_k|)$ time. 
\end{lemma}


\subsection{Handling insertions of elements}
\label{main:sec:insertions}

\newcommand{\N}{\mathcal{N}}

We handle the insertion of an element $e'$   by calling the procedure  in Figure~\ref{main:fig:insertion}, where  $\S_{e'} = \{ s \in \S : e' \in s\}$ denotes the collection of sets  containing  $e'$. From the outset,  we  often do not explicitly specify how the real-weight  and dead-weight of a set $s$ changes due to the execution of the procedure in Figure~\ref{main:fig:insertion}. Instead, they will be implicitly determined as: $w_s = \sum_{e \in s} w_e$ and $w^*_s = w_s + \phi_s$.

Step (01) assigns the element $e'$ a level and a weight in  accordance with Property~\ref{main:static:property:level}. This increases the real-weight  and the total-weight  of every set $s \in \S_{e'}$ by $w_{e'}$. So the sets in $\S_{e'}$ can now potentially violate Invariant~\ref{main:inv:deletions:weights}.  We say that a set $s$ is {\em dirty} if it violates Invariant~\ref{main:inv:deletions:weights}, and  {\em clean} otherwise. Note that Observation~\ref{main:ob:dirty:1} is satisfied at this juncture.   The {\sc For} loop in Step (02) takes care of these dirty sets. Before proceeding  further, we need to define a few important notations.

\smallskip
\noindent {\bf Notations:} For any set $s \in \S$, we let $\N(s) = \{s' \in \S \setminus \{s\} : s \cap s' \neq \emptyset \}$ denote the {\em neighbors} of $s$. Next, fix any set $s \in \S$ and consider the following thought experiment. Suppose that we move the set $s$ from its current level   to some other level $j$, while keeping the levels of all the remaining sets $s' \in \S \setminus \{s\}$ unchanged. This potentially changes the levels and weights of some of the elements $e \in s$ in accordance with Property~\ref{main:static:property:level}, and hence the real-weights $w_{s'}$ of some of the sets  $s' \in \N(s)$  also get changed. Let $w_{s'}(s \rightarrow j)$ denote the resulting real-weight of a set $s'$ after  $s$ has moved  to level $j$.  It is easy to check that $w_{s'}(s \rightarrow j)$ is a continuous non-increasing function of $j$ for all $s', s \in \S$, and that $w_{s}(s \rightarrow \infty) = 0$ for all $s \in \S$. This leads us to the  concept of the {\em target level} $\ell^*(s)$ of a set $s \in \S$ with real-weight $w_s > 1$: If a set $s \in \S$ has real-weight $w_s > 1$, then  $\ell^*(s) = \min\{ j : w_s(s \rightarrow j) = 1\}$.  Note that $\ell^*(s) > \ell(s)$. 

\begin{observation}
\label{main:ob:dirty:1}
A set $s'$ is dirty only if $s' \in \S_{e'}$ and $w^*_{s'}  > 1$. 
\end{observation}

\begin{figure}[htbp]
                                                \centerline{\framebox{
                                                                \begin{minipage}{5.5in}
                                                                        \begin{tabbing}                                                                            
                                                                                01.  \=  Assign the element $e'$ a level $\ell(e') \leftarrow \max_{s \in \S_{e'}} \ell(s)$ and weight $w_{e'} \leftarrow (1+\epsilon)^{-\ell(e')}$. \\     
                                                                                02. \> {\sc For} every set $s \in \S_{e'}$:   Call the subroutine FIX$(s)$. \\                                                                                 
                                                                                03. \> {\sc While} Invariant~\ref{main:inv:deletions:rebuild} is violated: \\
                                                                                04. \> \ \ \ \ \ \= Identify the smallest level $k \geq 0$ such that $\sum_{s \in \S : \ell(s) \leq k} \phi_s > \epsilon \cdot f \cdot \sum_{e \in \U : \ell(e) \leq k} w_e$. \\
                                                                                05. \> \> Call the subroutine {\sc Rebuild}($\leq k$) as described in Section~\ref{main:sec:deletions}.
                                                                        \end{tabbing}
                                                                \end{minipage}
                                                        }}
                                                        \caption{\label{main:fig:insertion} Handling the insertion of an element $e'$.}
                                                \end{figure}

\noindent 
{\bf The subroutine {\sc Fix}$(s)$:} By induction, suppose that Observation~\ref{main:ob:dirty:1} holds at the start of a given call to  {\sc Fix}$(s)$. The subroutine  first checks if $w^*_s > 1$. If not, then Observation~\ref{main:ob:dirty:1} implies that the set $s$ is already clean and hence the subroutine finishes execution and returns the call. From now onward, we assume that $w^*_s = 1+\mu_s$ for some $\mu_s > 0$ at the beginning of the call. If $\phi_s \geq \mu_s$, then we set $\phi_s \leftarrow \phi_s - \mu_s$. This makes the set $s$ clean, and again the subroutine finishes execution. Hence, from now onward, we assume that $\mu_s > \phi_s$ at the beginning of the call. We first set $\phi_s \leftarrow 0$, in order to reduce the total-weight of $s$ as much as possible. At this stage, we have $w^*_s = w_s > 1$ and $\phi_s = 0$. The subroutine now moves the set $s$ up to its target-level $\ell^*(s) = j$ (say). This upward movement is achieved via a continuous process. Informally, as the set $s$ keeps moving up, some of its neighbors $s' \in \N(s)$ keep losing their real-weights (because the weights of the some of the elements  $e \in s' \cap s$  keep decreasing). In order to compensate for this loss, the affected neighbors $s' \in \N(s)$ keep increasing their dead-weights $\phi_{s'}$ in a continuous manner, whenever possible.

To be more specific, consider an infinitesimal time-interval during this continuous process when the set $s$ moves up from level $\lambda$ to $\lambda+d\lambda$.  As a result, some of the elements $e \in s$ have their weights decreased. This in turn change the real-weights $w_{s'}$ of some of the neighbors $s' \in \N(s)$ by $d w_{s'}(s \rightarrow \lambda)$. Note that $d w_{s'}(s \rightarrow \lambda) \leq 0$. If $w_{s'}(s \rightarrow \lambda) \leq 1$, then  we set  $\phi_{s'} \leftarrow \phi_{s'} - d w_{s'}(s \rightarrow \lambda)$, in order to compensate for the loss of real-weight of $s'$ during this infinitesimally small time-interval.

Thus, from the perspective of a neighbor $s' \in \N(s)$, the process looks like this: As the set $s$ keeps moving up, the real-weight of $s'$ keep decreasing  in a continuous manner, until  $w_{s'}$ becomes $\leq 1$. From this point onward, the dead-weight $\phi_{s'}$ keeps increasing at the same rate at which the real-weight $w_{s'}$ decreases (thereby keeping the total-weight $w^*_{s'}$ constant). 

 
\begin{observation}
\label{main:ob:dirty:2}
A call to {\sc Fix}$(s)$ never leads to an already clean set becoming dirty. Furthermore, if Observation~\ref{main:ob:dirty:1} holds in the beginning of the call, then it continues to hold at the end of the call.
\end{observation}

At the end of the {\sc For} loop in Figure~\ref{main:fig:insertion}, every set is clean and hence Invariant~\ref{main:inv:deletions:weights} is satisfied.  However, the dead-weights of some of the sets are increased due to the calls to {\sc Fix}$(s)$. This might lead to a violation of Invariant~\ref{main:inv:deletions:rebuild}. This is addressed by the {\sc While} loop in steps (03)-(05). Observation~\ref{main:ob:rebuild:2} implies that both the invariants hold at the end of procedure in Figure~\ref{main:fig:insertion}.

\subsection{Bounding the approximation ratio and amortized update time}
\label{main:sec:approx}
\label{main:sec:time}


The following theorem upper bounds the approximation ratio of our dynamic algorithm. 

\begin{theorem}
\label{main:th:approx}
The collection of tight sets $\T^* = \{ s \in  \S : w^*_s = 1\}$ forms a $(1+\epsilon) f$-approximate minimum set cover in $(\U, \S)$.
\end{theorem}
\begin{proof}
Following the argument in the proof of Lemma~\ref{main:cor:static:dual},  Invariant~\ref{main:inv:deletions:weights} implies that the collection of tight sets $\T^*$  forms a set cover in $(\U, \S)$, and the element-weights $\{ w_e \}_{e \in \U}$ form a  fractional packing in $(\U, \S)$. Next, as in the proof of Lemma~\ref{main:static:approx}, we first derive that $\sum_{e \in \U} f \cdot w_e \geq \sum_{s \in \T^*} w_s$. Adding the term $\sum_{s \in \T^*} \phi_s$ to both sides of this inequality, we get: $\sum_{e \in \U} f \cdot w_e + \sum_{s \in \T^*} \phi_s \geq \sum_{s \in \T^*} (w_s+\phi_s) = \sum_{s \in \T^*} w^*_s = |\T^*|$. Next, from Invariant~\ref{main:inv:deletions:rebuild} we derive that: $(1+\epsilon)f \cdot \sum_{e \in \U} w_e \geq \sum_{e \in \U} f \cdot w_e + \sum_{s \in \T^*} \phi_s \geq  |\T^*|$. In other words, there is a fractional packing $\{w_e\}_{e \in \U}$ whose value is within a multiplicative $(1+\epsilon) f$ factor of the size of a valid set cover $\T^*$.  Hence, $\T^*$ is a $(1+\epsilon)f$-approximate minimum set cover in $(\U, \S)$ according to Lemma~\ref{main:lm:dual}.
\end{proof}

\renewcommand{\E}{\mathcal{E}}

We spend the rest of this section explaining the main ideas behind the analysis of the amortized update time of our algorithm. We start with an assumption that helps simplify this analysis. 
\begin{assumption}
\label{main:assume:nice:insertion}
Suppose that an element $e'$ getting inserted  is assigned to a level $\ell(e') = j'$ in step (01) of Figure~\ref{main:fig:insertion}. After that, the level of  $e'$ does not change  during the {\sc For} loop in step (02). 
\end{assumption}

The update time of our algorithm is dominated by two main types of operations: (1) an iteration of the {\sc For} loop in Figure~\ref{main:fig:insertion}  where a set $s$ potentially moves up to its target-level, and (2) a call to the subroutine {\sc Rebuild}$(\leq k)$. For an operation of type  (1), in this section we bound the {\em fractional work} done by our algorithm instead of the actual time taken to implement it. We  give an intuitive justification as to why  fractional work is a useful proxy for the actual running time that is analyzed in the full version. In order to bound the time spent on operations of type (2), we introduce the notion of {\em down-tokens}. We now explain each of these concepts in  more details.

\smallskip
\noindent {\bf Fractional work:} Consider an event where a set $s \in \S$ moves up from level $j_0$ to level $j_1$, and the level of every other set remains unchanged. This event can change the level of an element $e \in \U$ only if $e \in s$. For all $e \in s$, let $\ell_0(e)$ and $\ell_1(e)$ respectively denote the level of $e$ just before and just after the event. Then the total {\em fractional work} done  during this event  $= \sum_{e \in s} f \cdot (\ell_1(e) - \ell_0(e))$. 

\smallskip
\noindent {\bf Justification for fractional work:} In the full version our starting point will be a {\em discretized variant} of the static algorithm from Section~\ref{main:sec:prelim}, where in each round the weights of of the alive elements increase by a multiplicative factor of $(1+\epsilon)$ and the level of every set and element is an integer in the range $\{0, 1, \ldots, \lceil \log_{(1+\epsilon)} n \rceil\}$. Using appropriate data structures, we can ensure that our algorithm spends $O(f)$ time for each element increasing its level by one unit.  This precisely corresponds to the notion of fractional work defined above (when the levels are integers).

\smallskip
\noindent {\bf Down-tokens:} We associate  $(1+\epsilon)^{\ell(s)} \cdot \phi_s$ amount of {\em down-tokens} with each set $s \in \S$. The {\em total volume} of down-tokens is given by $\sum_{s \in \S} (1+\epsilon)^{\ell(s)} \cdot \phi_s$.  

\smallskip
\noindent {\bf The parameter $\gamma_{\epsilon}$:} In the rest of this section, to ease notations we define  $\gamma_{\epsilon} = (\ln(1+\epsilon))^{-1} = O(\frac{1}{\epsilon})$. 

\smallskip
\noindent {\bf Overview of our analysis:} By Lemma~\ref{main:lm:insertions:2},  the total fractional work done per update  due to operations of type (1) is at most $O(f^2 \gamma_{\epsilon}) = O(f^2/\epsilon)$.
It now remains to bound the total time spent on operations of type (2). Towards this end, we make the following important observations: (a) Excluding the calls to {\sc Rebuild}$(\leq k)$, the procedure for handling the insertion of an element increases the total volume of down-tokens by at most $O(f^2)$ (see Corollary~\ref{main:cor:lm:insertions:1}). (b) Excluding the calls to {\sc Rebuild}$(\leq k)$, the procedure for handling  the deletion of an element increases the total-volume of down-tokens by at most $O(f)$. This holds because when an element $e$ gets deleted, the dead-weight $\phi_s$ associated with each set $s$ containing it increases by $w_e = (1+\epsilon)^{-\ell(e)}$ (see the first paragraph in Section~\ref{main:sec:deletions}).  Hence, the total volume of down-tokens increases by $\sum_{s \in \S : e \in s}(1+\epsilon)^{\ell(s)} \cdot (1+\epsilon)^{-\ell(e)} \leq \sum_{s \in \S : e \in s} (1+\epsilon)^{\ell(e)} \cdot (1+\epsilon)^{-\ell(e)} \leq f$. (c) The total time spent on all the calls to {\sc Rebuild}$(\leq k)$ is at most $O(1/\epsilon)$ times the decrease in the total volume of down-tokens (see Lemma~\ref{main:lm:deletions}). Since the total volume of down-tokens is always nonnegative, all these observations taken together imply an amortized update time of $O(f^2/\epsilon)$. This is slightly better than the bound  in Theorem~\ref{th:main:intro:both results}, because in the full version we do not have the luxury of analyzing {\em fractional work}.

\begin{lemma}
\label{main:lm:insertions:1}
Consider a call to the subroutine {\sc Fix}$(s)$ in Figure~\ref{main:fig:insertion}, which moves the set $s$ up to its target-level. Consider an infinitesimally small time-interval during this iteration when the set $s$ moves up from level $\lambda$ to level $\lambda+d\lambda$. During this infinitesimally small time-interval, (a) the  fractional work done   $= -f \gamma_{\epsilon} (1+\epsilon)^{\lambda} \cdot d w_s(s \rightarrow \lambda)$, and (b) the total volume of down-tokens increases by  $\leq -f  (1+\epsilon)^{\lambda} \cdot d w_s(s \rightarrow \lambda)$. Here, we have $dw_s(s \rightarrow \lambda) = w_s(s \rightarrow \lambda+d\lambda) - w_s(s \rightarrow \lambda)$. 
\end{lemma}
\begin{proof}(Sketch) Consider the collection of elements $X_{s}(\lambda) = \{ e \in \U : e \in s, \text{ and } \ell(s') \leq \lambda \text{ for all } s' \in \S \setminus \{s\} \text{ with } e \in s'  \}$. As the set $s$ moves up from level $\lambda$ to level $\lambda+d\lambda$, each element $e \in X_{s}(\lambda)$ changes its level by $d\lambda$ and no other element changes its level.\footnote{Since we consider an infinitesimally small interval and  the collection $\{ \ell(e) : e \in \U\}$ is finite,  $X_{s}(\lambda) = X_{s}(\lambda+d\lambda)$. } Hence, we get: 
\begin{equation}
\label{main:eq:work:1}
\text{The fractional work done during this interval }= f |X_s(\lambda)| \cdot d\lambda.
\end{equation} 
When the set $s$ is at level $\lambda$, each element $e \in X_s(\lambda)$ has weight $(1+\epsilon)^{-\lambda}$. Starting from $\lambda$, if we increase the level of $s$ by an infinitesimal amount, then only the elements $e \in X_s(\lambda)$ change their weights, while the weights of every other element in $s$ remains unchanged. Hence, we derive that:
$$\frac{d w_{s}(s \rightarrow \lambda)}{d\lambda} = \frac{d}{d\lambda} \left( |X_s(\lambda)| \cdot (1+\epsilon)^{-\lambda}\right) = - \gamma_{\epsilon}^{-1} \cdot (1+\epsilon)^{-\lambda} \cdot   |X_s(\lambda)|.$$
Thus, we get: $|X_s(\lambda)| \cdot d\lambda = - \gamma_{\epsilon} \cdot (1+\epsilon)^{\lambda} \cdot d w_{s}(s \rightarrow \lambda)$. Part (a) of the lemma now follows from~(\ref{main:eq:work:1}).

Let  $dw$ be the change in the weight of an element $e \in X_s(\lambda)$  as the set $s$ moves up from level $\lambda$ to level $\lambda+d\lambda$. Note that $dw < 0$. From the preceding discussion, it follows that: 
\begin{equation}
\label{main:eq:work:2}
d w_s(s \rightarrow \lambda) = |X_s(\lambda)| \cdot dw.
\end{equation} 
Consider the collection of sets $\S^*(\lambda) = \{ s' \in \S \setminus \{s\} :  s' \cap X_s(\lambda) \neq \emptyset \text{ and } \ell(s') \leq \lambda\}$. As the set $s$ moves up from level $\lambda$ to level $\lambda+d\lambda$, each set $s' \in  \S^*(\lambda)$ decreases its real-weight $w_{s'}$ by $|s' \cap X_s(\lambda)| \cdot (-dw) \geq 0$, and the real-weight of every other set $s' \notin \S^*(\lambda) \cup \{s\}$ remains unchanged. For each of these sets $s' \in \S^*(\lambda)$, the increase in its dead-weight $\phi_{s'}$  is upper bounded by the decrease in its real-weight (see the description of  {\sc Fix}$(s)$ in Section~\ref{main:sec:insertions}). Hence, for each set $s' \in \S^*(\lambda)$, we get $0 \leq d\phi_{s'} \leq |s' \cap X_s(\lambda)| \cdot (-dw)$. None of the other sets change their dead-weights as $s$ moves up from level $\lambda$ to level $\lambda+d\lambda$. Thus,  total volume of the down-tokens increases by:
\begin{eqnarray*}
\sum_{s' \in \S^*(\lambda)} (1+\epsilon)^{\ell(s')} \cdot d\phi_{s'} \leq (1+\epsilon)^{\lambda}  \cdot \sum_{s' \in \S^*(\lambda)}  d\phi_{s'} \leq (1+\epsilon)^{\lambda}  \cdot (-dw) \cdot \sum_{s' \in \S^*(\lambda)} |s' \cap X_s(\lambda)| \\
\leq  (1+\epsilon)^{\lambda}  \cdot (-dw) \cdot f \cdot |X_s(\lambda)|.
\end{eqnarray*}
The last inequality holds since each element in $X_s(\lambda)$ is contained in at most $f$ sets from $\S^*(\lambda)$. By~(\ref{main:eq:work:2}),  the increase in the total volume of down-tokens is $\leq -f (1+\epsilon)^{\lambda} \cdot d w_s(s \rightarrow \lambda)$. 
\end{proof}

\begin{lemma}
\label{main:lm:insertions:2}
During steps (01)-(02) in Figure~\ref{main:fig:insertion}, the total fractional work done is  $O(f^2 \gamma_{\epsilon})$.
\end{lemma}
\begin{proof}
Suppose that the {\sc For} loop in step (02) runs for $r$ iterations, where in each iteration $i \in \{1, \ldots, r\}$ it deals with a distinct set $s_i \in \S_{e'}$. We will show that the fractional work done during each iteration is $O(f \gamma_{\epsilon})$. Since $r = |\S_{e'}| \leq f$, this will imply the lemma. 

For the rest of the proof, focus on any iteration $i \in \{1, \ldots, r\}$, and the call to {\sc Fix}$(s_i)$. At the start of this call, we have $w^*_s = w_s = 1+\delta_i$, for some $\delta_i > 0$, even after resetting the dead-weight $\phi_{s_i} \leftarrow 0$. (Otherwise, the set $s_i$ does not change its level and the fractional work done $= 0$).  At the end of this iteration, the set $s_i$ has moved up to its target-level $\ell^*(s) \leq \ell(e')$ (this inequality follows from Assumption~\ref{main:assume:nice:insertion}), and it is clean with weights $w^*_s = w_s = 1$. Part (a) of Lemma~\ref{main:lm:insertions:1} now implies that the fractional work done during this iteration is:
\begin{eqnarray}
= \int_{1+\delta_i}^{1} -f\gamma_{\epsilon} (1+\epsilon)^{\lambda} dw_{s_i}(s_i \rightarrow \lambda) \leq f\gamma_{\epsilon}(1+\epsilon)^{\ell^*(s)} \int_{1+\delta_i}^{1} -dw_{s_i}(s_i \rightarrow \lambda) = \delta_i f \gamma_{\epsilon} (1+\epsilon)^{\ell^*(s)} \nonumber \\
\leq \delta_i f \gamma_{\epsilon} (1+\epsilon)^{\ell(e')}. \label{main:eq:100}
\end{eqnarray}
  Let $w_{s_i}^{'}, w_{s_i}^{''}, w_{s_i}^{'''}, w_{s_i}^{''''}$ respectively denote the real-weight of the set $s_i$ just before the insertion of the element $e'$, just after step (01), just before it starts moving up towards its target-level during the call to {\sc Fix}$(s_i)$, and  just after the call to {\sc Fix}$(s_i)$. Thus, we have $w_{s_i}^{'''} = 1+\delta_i$ and $w_{s_i}^{''''} = 1$.  Since a call to {\sc Fix}$(.)$ never  increases  the real-weight of any set, it follows that $w_{s_i}^{'''} \leq w_{s_i}^{''}$.  Since the set $s_i$ was clean just before the insertion of the element $e'$,  we get $w_{s_i}^{'} \leq 1$.  Finally, step (01) in Figure~\ref{main:fig:insertion}  implies that  $w_{s_i}^{''} = w_{s_i}^{'}+(1+\epsilon)^{-\ell(e')}$. To summarize, we have:
\begin{equation}
\label{main:eq:200}
w_{s_i}^{'} \leq 1 = w_{s_i}^{''''} < w_{s_i}^{'''} = 1+\delta_i \leq w_{s_i}^{''} = w_{s_i}^{'}+(1+\epsilon)^{-\ell(e')}.
\end{equation}
From~(\ref{main:eq:200}) we derive that $1+\delta_i \leq w_{s_i}^{'}+(1+\epsilon)^{-\ell(e')} \leq 1 + (1+\epsilon)^{-\ell(e')}$, which gives us: $\delta_i \leq (1+\epsilon)^{-\ell(e')}$. This observation, along with~(\ref{main:eq:100}), concludes the proof of the lemma.
\end{proof}


\begin{corollary}
\label{main:cor:lm:insertions:1}
 Steps (01)-(02) in Figure~\ref{main:fig:insertion} increase the total volume of down-tokens by  $O(f^2)$.
\end{corollary}
\begin{proof}
By Lemma~\ref{main:lm:insertions:1},   the  increase in the total volume of down-tokens during steps (01)-(02) is at most $\gamma_{\epsilon}^{-1}$ times the total fractional work done. The corollary now  follows from Lemma~\ref{main:lm:insertions:2}. 
\end{proof}


\begin{claim}
\label{main:cl:imp}
Let $\alpha_1 \ldots \alpha_j$ and $\beta_1 \ldots \beta_j$ be nonnegative real numbers satisfying the following property: $j$ is the smallest index  $j' \in \{1, \ldots, j\}$ such that  $\sum_{i=1}^{j'} \alpha_i >  \sum_{i=1}^{j'} \beta_i$. Then for all $0 \leq \lambda_1 \leq \cdots \leq \lambda_j$, we have  $\sum_{i=1}^j (1+\epsilon)^{\lambda_i} \cdot \alpha_i \geq \sum_{i=1}^j (1+\epsilon)^{\lambda_i} \cdot \beta_i$.
\end{claim}

\renewcommand{\L}{\mathcal{L}}

\begin{lemma}
\label{main:lm:deletions}
The  time spent to implement a call to the subroutine {\sc Rebuild}$(\leq k)$ is at most $O(1/\epsilon)$ times the decrease in the total volume of down-tokens during the same call.
\end{lemma}
\begin{proof}
Unless specified otherwise, throughout this proof we focus on the time-instant just before the  call to {\sc Rebuild}$(\leq k)$.  At that time, from Section~\ref{main:sec:deletions} and Figure~\ref{main:fig:insertion} we infer that: 
\begin{equation}
\label{main:eq:level:1}
k \text{ is the smallest level such that } \sum_{s \in \S : \ell(s) \leq k} \phi_s >  \sum_{e \in \U : \ell(e) \leq k} \epsilon f \cdot w_e. 
\end{equation} Define $\L_k = \{ \lambda : \lambda \leq k \text{ and } \lambda = \ell(s) \text{ for some set } s \in \S\}$.  Since each set gets assigned to exactly one level, we have $|\L_{k}| \leq |\S|$. In particular, the collection $\L_{k}$ is finite. Let $\L_k = \{\lambda_1, \ldots, \lambda_j\}$ where $\lambda_1 < \lambda_2 < \cdots < \lambda_j  = k$. For any $\lambda_i \in \L_k$, let $\phi_{i} = \sum_{s \in \S : \ell(s) = \lambda_i} \phi_s$ and $w_{i} = \sum_{e \in \U : \ell(e) = \lambda_i} w_e$ respectively denote the total dead-weight and element-weight stored at level $\lambda_i$.  Note that  $\sum_{s \in \S : \ell(s) \leq k} \phi_s = \sum_{i=1}^j \phi_{i}$ and $\sum_{e \in \U : \ell(e) \leq k} w_e = \sum_{i=1}^j w_{i}$. From~(\ref{main:eq:level:1}) we get:
$j \text{ is the smallest index } j' \in \{1, \ldots, j\} \text{ s.t. } \sum_{i=1}^{j'} \phi_i > \sum_{i=1}^{j'} \epsilon f \cdot w_i.$  Now,  Claim~\ref{main:cl:imp} gives us: 
\begin{equation}
\label{main:eq:updatetime:1}
\sum_{i=1}^{j} (1+\epsilon)^{\lambda_i} \cdot \phi_i \geq \sum_{i=1}^{j'} (1+\epsilon)^{\lambda_i} \cdot \epsilon f \cdot w_i
\end{equation}
Each element $e \in \U$ at level $\ell(e) = \lambda_i$ has weight $w_e = (1+\epsilon)^{-\lambda_i}$.  Hence, the quantity $(1+\epsilon)^{\lambda_i} w_i$  equals the number of elements at levels $\lambda_i$. Summing over all the levels in $
\L_k$, we infer that the right hand side (RHS) of~(\ref{main:eq:updatetime:1}) equals $\epsilon f \cdot \U'_{\leq k}$, where $\U'_{\leq k}$ is the collection of elements at levels $\leq k$ just before the call to {\sc Rebuild}$(\leq k)$. In contrast, the left hand side (LHS) of~(\ref{main:eq:updatetime:1}) equals the total volume of down-tokens at level $\leq k$ just before the call to {\sc Rebuild}$(\leq k)$. The call to {\sc Rebuild}$(\leq k)$ does not change the dead-weight of any set $s \in \S$ at level $\ell(s) > k$, and Observation~\ref{main:ob:rebuild:2} states that at the end of the call to this subroutine every set at level $\leq k$ has zero dead-weight. To summarize, the LHS of~(\ref{main:eq:updatetime:1}) equals the amount by which the total volume of down-tokens decreases during the call to {\sc Rebuild}$(\leq k)$, whereas the RHS of~(\ref{main:eq:updatetime:1}) equals $\epsilon$ times the total time spent by our algorithm to implement this call (see Lemma~\ref{main:lm:rebuild:runtime}). This completes the proof of the lemma.
\end{proof}



\section{Our Algorithm for Worst Case Update Time: An Overview}
\label{main:sec:worstcase}


Our complete algorithm needs to deal with a lot of subtle issues, and is presented in the full version. Here, to highlight the main ideas, we only focus on the decremental (deletions only) setting. 

\begin{observation}
\label{main:lm:static:level}
The static algorithm in Section~\ref{main:sec:prelim} never assigns a set $s \in \S$ to a level $\ell(s) \geq L =  1+ \lceil \log_{(1+\epsilon)} n \rceil$, where $n$ is the total number of elements.
\end{observation}
\begin{proof}
Recall the continuous process from Section~\ref{main:sec:prelim} that starts at time $t = -\infty$. When $t = L$,  every element $e \in s$ has weight $w_e  \leq (1+\epsilon)^{-L} < 1/n$, and hence  $w_s = \sum_{e \in s} w_e < |s| \cdot (1/n) \leq 1$. Since  $s$ is not yet tight at time $t = L$,  it gets assigned to a level $< L$ at the end of the algorithm.
\end{proof}

\noindent {\bf Schedulers:} Our dynamic algorithm uses $L$ subroutines called {\sc schedulers} -- one for each level in $[L] = \{1, \ldots, L\}$. Informally, for each $k \in [L]$ the subroutine  {\sc Scheduler}$(k)$ is  responsible for all the sets and elements at levels $\leq k$. Each scheduler works on its own local memory that is disjoint from the memory locations used by the other schedulers.  Another key feature of our algorithm is that we allow different schedulers to hold mutually inconsistent views regarding the level of an individual set or element. Before proceeding any further, we  introduce some important concepts and notations. Most of the concepts defined below -- such as the notions of real-weights, dead-weights and total-weights -- closely mirror their counterparts from Section~\ref{main:sec:algo}.

 Let $\S^{(k)} \subseteq \S$ and $\U^{(k)} \subseteq \U$ respectively denote the collection of sets and elements {\sc Scheduler}$(k)$ is responsible for. Let $\ell^{(k)}(e)$ and $\ell^{(k)}(s)$ respectively denote the level of an element $e \in \U^{(k)}$ and a set $s \in \S^{(k)}$ from the perspective of {\sc Scheduler}$(k)$.  As usual, the level of an element $e \in \U^{(k)}$ according to {\sc Scheduler}$(k)$ is completely determined by the levels of the sets in $\S^{(k)}$ that contain it: We have $\ell^{(k)}(e) = \max_{s \in \S^{(k)} : e \in s} \ell^{(k)}(s) \leq k$ for all $e \in \U^{(k)}$. Let  $w^{(k)}_{e} = (1+\epsilon)^{-\ell^{(k)}(e)}$  be the  weight of an element $e \in \U^{(k)}$ according to {\sc Scheduler}$(k)$.  The {\em real-weight} of a set $s \in \S^{(k)}$ according to {\sc Scheduler}$(k)$ equals $w^{(k)}_{s} = \sum_{e \in s \cap \U^{(k)}} w^{(k)}_{e} + \delta^{(k)}_{s}$, where  $\delta^{(k)}_{s}$ is the {\em extra-weight} of $s$. We will shortly see that the concept of extra-weight has a natural  explanation. Intuitively, the quantity $\delta_{s}^{(k)}$ measures the weight received by a set $s \in \S^{(k)}$ from elements that are at level $> k$. Each set $s \in \S^{(k)}$ has a {\em dead-weight} $\phi_{s}^{(k)}$, and its {\em total-weight} is given by $w^{*(k)}_s = w^{(k)}_s + \phi^{(k)}_s$. Finally, {\sc Scheduler}$(k)$  maintains a  collection  $D^{(k)}$ of some elements that got deleted from $\U^{(k)}$  in the past due to an external update operation. The elements in $D = \bigcup_{k=1}^{L} D^{(k)}$ are called {\em dead elements}. 
 
 


\smallskip
\noindent {\bf Preprocessing:} We first run the static algorithm from Section~\ref{main:sec:prelim}.  Let  $\ell(s), \ell(e), w_s, w_e$ denote the levels and weights of elements and sets returned by this static algorithm. At this stage, all the different schedulers completely agree with each other regarding the level of each element and set. Specifically, consider any level $k \in [L]$. At this point in time, we have $\S^{(k)} = \{ s \in \S : \ell(s) \leq k\}$, $\U^{(k)} = \{ e \in \U : \ell(e) \leq k\}$ and $\D^{(k)} = \emptyset$. For all elements $e \in \U^{(k)}$ and sets $s \in \S^{(k)}$,  we have $\ell^{(k)}(e) = \ell(e)$, $w^{(k)}_{e} = w_e$ and $\ell^{(k)}(s) = \ell(s)$. For all $s \in \S^{(k)}$, we set  $\delta^{(k)}_{s} = \sum_{e \in s : \ell(e) > k} w_e$ and $\phi^{(k)}_s = 0$. Hence,  we have  $w^{*(k)}_s = w^{(k)}_{s} = w_s$ for all $s \in \S^{(k)}$ just after  the end of preprocessing. 

\smallskip
\noindent {\bf Invariants:}
We now describe three important invariants that are satisfied by our dynamic algorithm. Invariant~\ref{main:inv:worstcase:tight} closely mirrors Invariant~\ref{main:inv:deletions:weights} from Section~\ref{main:sec:deletions}, and it clearly holds at the end of preprocessing.  Invariant~\ref{main:inv:worstcase:trigger} dictates that the number of dead elements in the control of a {\sc Scheduler}$(k)$ is really small compared  to the number of elements in $\U^{(k)}$. At the end of preprocessing, this invariant trivially holds because $D^{(k)} = \emptyset$. Invariant~\ref{main:inv:laminar} says that the sets $\S^{(k)}, \U^{(k)}$ and $D^{(k)}$ form a very nice laminar structure as $k$ ranges from $L$ to $1$. Specifically, the  sets/elements a {\sc Scheduler}$(k-1)$ is responsible for are precisely the ones that lie at level $\leq k-1$ according to the next {\sc Scheduler}$(k)$. Furthermore, {\sc Scheduler}$(L)$ is responsible for all the elements and sets. Again, it is easy to check that this invariant holds at the end of preprocessing. 


 \begin{invariant}
 \label{main:inv:worstcase:tight}
  Consider any  $k \in [L]$. Every set $s \in \S^{(k)}$ at level $\ell^{(k)}(s) > 0$  has total-weight $w^{*(k)} = 1$. Furthermore, every set $s \in \S^{(k)}$ at level $\ell^{(k)}(s) = 0$ has total-weight $w^{*(k)} \leq 1$.  
  \end{invariant}
  
 \begin{invariant}
\label{main:inv:worstcase:trigger}
Consider any $k \in [L]$. We have $\left| D^{(k)} \right| \leq 2\epsilon \cdot \left| \U^{(k)} \right|$.
\end{invariant}

 \begin{invariant}
 \label{main:inv:laminar}
For every $k \in [2, L]$, we have $\S^{(k-1)} = \{ s \in \S^{(k)} : \ell^{(k)}(s) \leq k-1\}$, $\U^{(k-1)} = \{ e \in \U^{(k)} : \ell^{(k)}(e) \leq k-1\}$ and $\D^{(k-1)} = \{ e \in \D^{(k)} : \ell^{(k)}(e) \leq k-1\}$. Furthermore, we have  $\S^{(L)} = \S$ and $\U^{(L)} = \U$. 
 \end{invariant}

 \noindent{\bf Ownerships:} We  say that an element $e \in \U$ (resp. $e \in D$) is {\em owned} by a level $k \in [L]$ iff $e \in \U^{(k)}$ (resp. $e \in D^{(k)}$) and $e \notin \U^{(j)}$ (resp. $e \in D^{(j)}$) for all $j \in [k-1]$. A set $s \in \S$ is {\em owned} by a level $k \in [L]$ iff $s \in \S^{(k)}$ and $s \notin \S^{(j)}$ for all $j \in [k-1]$.  Note that each element/set is owned by a unique level. We now describe how to handle a sequence of element deletions after preprocessing.

\smallskip
 \noindent {\bf Handling the deletion of an element:} Suppose that an element $e$, which was owned by level $j$, gets deleted.  Accordingly, we {\em feed} this deletion to {\sc Scheduler}$(j)$, \ldots, {\sc Scheduler}$(L)$.  Note that these are precisely the schedulers that are responsible for this element and are affected by the concerned deletion. Each of these affected schedulers works within its own local memory, independently of others. We describe the actions taken by  {\sc Scheduler}$(k)$, for any $k \in \{j, \ldots, L\}$.
 
Suppose that the element $e$ was at level $\ell^{(k)}(e) = i$  just before its deletion.  {\sc Scheduler}$(k)$ moves the element $e$ from $\U^{(k)}$ to $D^{(k)}$, without changing the its level $\ell^{(k)}(e) = i$ or weight $w^{(k)}_e = (1+\epsilon)^{-i}$.  For each set $s\in \S^{(k)}$ containing the element $e$, this reduces its real-weight $w^{(k)}_s$  by $(1+\epsilon)^{-i}$. To compensate for this loss, we increase its  dead-weight $\phi^{(k)}_s$ by the same amount $(1+\epsilon)^{-i}$. We do not change the extra-weights of the sets. Thus, the total-weight of every set also remains unchanged.

 \smallskip
 \noindent {\bf Triggering a rebuild:} It is easy to check that the above actions do not violate Invariant~\ref{main:inv:worstcase:tight} and Invariant~\ref{main:inv:laminar}. However, if we keep acting in this lazy manner, then the number of dead elements keep growing with time, and so after some number of updates Invariant~\ref{main:inv:worstcase:trigger} will get violated. Furthermore, unlike in Section~\ref{main:sec:algo} here we cannot even afford to wait until the moment Invariant~\ref{main:inv:worstcase:trigger} is violated before triggering a rebuild, because we are shooting for worst-case update time and the rebuild needs to occur {\em in the background} -- a few steps at a time. Accordingly, for each $k \in [L]$, whenever {\sc Scheduler}$(k)$ finds that $|D^{(k)}| \geq \epsilon \cdot |\U^{(k)}|$ (note that this is still {\em $\epsilon$-far} from violating Invariant~\ref{main:inv:worstcase:trigger}),  it starts {\em rebuilding} its part of the input $(\U^{(k)}, \S^{(k)}, D^{(k)})$ {\em in the background} in a separate memory location that is not affected by the happenings elsewhere.  We now explain this rebuilding procedure in a bit more details.
 
 \smallskip
 \noindent {\bf Rebuilding done by {\sc Scheduler}$(k)$:} Suppose that {\sc Scheduler}$(k)$ has triggered a rebuild at the present moment because $|D^{(k)}| = \epsilon \cdot |\U^{(k)}|$. In order to see the high level idea, for now assume that {\sc Scheduler}$(k)$ does not need to handle an element deletion due to any external update operation while it is performing the rebuild.\footnote{According to this assumption, during the same time-interval when {\sc Scheduler}$(k)$ is performing a rebuild, some other {\sc Scheduler}$(j)$ with $j > k$ might still have to handle element deletions due to external updates.} Then the rebuild subroutine will {\em clean-up} all the dead elements in $D^{(k)}$, by resetting $D^{(k)} = \emptyset$ and $\phi^{(k)}_s = 0$ for all $s \in \S^{(k)}$, and then run the static algorithm (see Section~\ref{main:sec:algo}) on input $(\S^{(k)}, \U^{(k)})$ starting from time $t = -k$ onward. During this run of the static algorithm, the extra-weights $\delta_{s}^{(k)}$ of the sets $s \in \S^{(k)}$ will not change, because these extra-weights are coming from elements that are at levels $> k$. Note that this is exactly the same principle that underpins the rebuild subroutine described in Section~\ref{main:sec:deletions}. At the end of this static algorithm, we will get the following guarantees: Invariant~\ref{main:inv:worstcase:tight} is satisfied by {\sc Scheduler}$(k)$, and $D^{(k)} = \emptyset$. At this point, {\sc Scheduler}$(k)$ will {\em order} all the schedulers corresponding to levels $j \in [k-1]= \{1, \ldots, k-1\}$ to abort whatever they are doing and synchronize their perspectives with the perspective of {\sc Scheduler}$(k)$. We refer to this as the {\em synchronization} event. At the end of this,  each {\sc Scheduler}$(j), j \in [k-1],$ will satisfy:  $\S^{(j)} = \{ s \in \S^{(k)} : \ell^{(k)}(s) \leq j\}$,  $\U^{(j)} = \{ e \in \U^{(k)} : \ell^{(k)}(s) \leq j \}$, $D^{(j)}  = \emptyset$, $\ell^{(j)}(e) = \ell^{(k)}(e)$ for all $e \in \U^{(j)}$, and $w^{*(j)}_s = w^{*(k)}_s$ for all  $s \in \S^{(j)}$. Our algorithm ensures that this synchronization is achieved  in $O(k)$ time {\em on the fly}, as follows. The rebuild subroutine of {\sc Scheduler}$(k)$ will run $k$ different {\em threads} -- one for each level $j \in \{1, \ldots, k\}$. Each of these threads will run in a designated separate region of memory, independent of others. It will be the responsibility of thread $j$ to prepare everything  in its allocated memory region, which can then be handed over to {\sc Scheduler}$(j)$ at the time of synchronization by simply passing a pointer to the beginning of the concerned memory region. Thus, synchronization involves the passing of $k$ pointers, and hence takes $O(k) = O(L) = O(\log n/\epsilon)$ time. It is easy to check that Invariants~\ref{main:inv:worstcase:tight} and~\ref{main:inv:laminar} continue to remain satisfied at the end of synchronization. Furthermore, the remaining Invariant~\ref{main:inv:worstcase:trigger} is trivially satisfied under our working assumption that no element deletion occurs in $\U^{(k)}$ as {\sc Scheduler}$(k)$ is rebuilding in the background. Next, we  analyze the worst-case update time of this algorithm, and outline what happens when this working assumption does not hold.

  \smallskip
  \noindent {\bf Worst-case update time:} We measure the time taken to implement a procedure in terms of {\em units of work}. The update time of our algorithm is dominated by the time spent on the rebuild subroutines of the individual schedulers. There are $L$ schedulers running in the background, and we  ensure that each of these individual schedulers perform $O(f L/\epsilon)$ units of work after every external update (element deletion). This implies a worst-case update time of $O(L \cdot f L/\epsilon) =  O(f \log^2 n/\epsilon^3)$. 
  
  We now explain why each {\sc Scheduler}$(k)$ needs to perform $O(f L/\epsilon)$ units of work per update. The factor $O(L)$ comes from the fact that the rebuild subroutine of  {\sc Scheduler}$(k)$ needs to run $k = O(L)$ threads in the background so as to ensure that it can synchronize on the fly when it finishes execution. It now remains to explain the rationale behind  the remaining $O(f/\epsilon)$ factor.  Suppose that $d = |D^{(k)}|$ and $u = |\U^{(k)}|$ when the rebuild subroutine of {\sc Scheduler}$(k)$ gets triggered. We have $d = \epsilon u$, and hence the subroutine needs to perform $O(f (u+d)) = O(fu)$ units of work overall (see Lemma~\ref{main:lm:rebuild:runtime}). According to our scheme, the subroutine  splits this work across a sequence of external updates (element deletions), performing $c f/\epsilon$ units of work per update for some large constant $c > 1$. This ensures that in the general setting (when the working assumption in the previous paragraph does not hold), at most $fu/(cf/\epsilon) = \epsilon u/c$ many elements get deleted from $\U^{(k)}$ when the rebuild subroutine is in progress. In the full version, we show that there is a way to handle these incoming deletions {\em on the fly} during the rebuild subroutine. However, this comes at a cost: When an element $e$ gets deleted that has already been processed by the rebuild subroutine, it gets classified as a dead element. But since the subroutine is working at a sufficient fast rate, at most $\epsilon u/c$ many new dead elements might get created in this manner at the time of synchronization, whereas the old dead elements that were present at the time the rebuild subroutine was triggered are removed anyway. Taking everything together, just after synchronization we end up having $|D^{(j)}| \leq (\epsilon/c) \cdot |\U^{(j)}|$ for all $j \in [k]$ and hence Invariant~\ref{main:inv:worstcase:trigger} continues to remain satisfied.

 \smallskip
  \noindent {\bf Approximation ratio:} The main challenge here is to show that all the same element/set might get assigned to different levels by different schedulers, there is a way to come up with a {\em consistent} assignment of levels to all the elements in $\U \cup D$ and all the sets in $\S$. Furthermore, the sets with weight $=1$ in this consistent assignment form a $(1+\epsilon)f$-approximate minimum set cover in $(\U, \S)$. 
  
The consistent assignment is as follows: Each set $s \in \S$ gets assigned to the level that owns it. Let $\ell(s)$ denote the level of a set $s \in \S$ in the consistent assignment. This automatically defines the level $\ell(e)$ and weight $w_e$ of every element $e \in \U \cup D$, since $\ell(e) = \max_{s \in \S: e \in s} \ell(s)$ and $w_e = (1+\epsilon)^{-\ell(s)}$. For all $s \in \S$, we  define its weight to be $w_s^{(\U \cup \D)} = \sum_{e \in \U \cup D : e \in s} w_e$. 


For each $k \in [L]$, define $\S(\leq k) = \{ s \in \S : \ell(s) \leq k\}$, $D(\leq k) = \{ e \in D : \ell(e) \leq k\}$ and $\U(\leq k) = \{ e \in \U : \ell(e) \leq k\}$.   The lemma below shows that the levels of elements/sets in the consistent assignment are nicely aligned with the levels of the same elements/sets according to the individual schedulers. Lemma~\ref{main:lm:worstcase:approx:1} and Invariant~\ref{main:inv:worstcase:trigger} together imply Corollary~\ref{main:cor:worstcase:approx:1}.
\begin{lemma}
\label{main:lm:worstcase:approx:1}
For each $k \in [L]$, we have $\S(\leq k) = \S^{(k)}$, $\U(\leq k) = \U^{(k)}$, and $D(\leq k) = D^{(k)}$.
\end{lemma}
\begin{proof}(Sketch)
The lemma follows from induction: It is easy to check that the lemma holds just after preprocessing. When an element $e$ gets deleted, each {\sc Scheduler}$(k)$ responsible for  $e$ simply moves it from $\U^{(k)}$ to $D^{(k)}$ without changing its weight or level. Hence, the lemma continues to remain satisfied. Finally, a moment's thought reveals that the lemma continues to hold after a rebuild subroutine of some {\sc Scheduler}$(k)$ executes its synchronization step. 
\end{proof}

 \begin{corollary}
  \label{main:cor:worstcase:approx:1}
  For each $k \in [L]$, we have $\left| D(\leq k) \right| \leq 2\epsilon \cdot \left| \U(\leq k) \right|$.
  \end{corollary}

   \begin{corollary}
  \label{main:cor:lm:worstcase:approx:2}
  We have $\sum_{e \in D} w_e \leq 2 \epsilon (1+\epsilon) \cdot \sum_{e \in \U} w_e$. 
  \end{corollary}
  \begin{proof}(Sketch)
  The proof is almost the same as the proof of Lemma 4.8 in the Arxiv version of~\cite{focs/BhattacharyaHN19}. Basically, consider any dead-element $e' \in D$ and any actual element $e \in \U$ such that both their levels lie within the interval $[k-1, k]$, i.e., $k-1 \leq \ell(e), \ell(e') \leq k$. Then it follows that their weights $w_{e'}$ and $w_{e}$ are within a $(1+\epsilon)$ multiplicative factor of each other.  This observation, along with Corollary~\ref{main:cor:worstcase:approx:1}, is sufficient to ensure that $\sum_{e \in D} w_e \leq 2 \epsilon (1+\epsilon) \cdot \sum_{e \in \U} w_e$.
 \end{proof}

  \begin{lemma}
  \label{main:lm:worstcase:approx:2}
  For every set $s \in \S$ at level $\ell(s) > 0$, we have $w_s^{(\U \cup \D)} = 1$. Furthermore, for every set $s \in \S$ at level $\ell(s) = 0$, we have $w_s^{(\U \cup \D)} \leq 1$.
  \end{lemma}
  
Lemma~\ref{main:lm:worstcase:approx:2} (whose proof follows from induction) closely mirrors Lemma~\ref{main:cor:static:dual} from Section~\ref{main:sec:prelim}. Consider the collection of sets $\T = \{ s \in \S : w_s^{(\U \cup D)} = 1\}$.  Lemma~\ref{main:lm:worstcase:approx:2} implies that  (see the proof of Lemma~\ref{main:cor:static:dual}) the element-weights $\{w_e\}$ form a valid fractional packing in $(\U \cup D, \S)$ and  $\T$ forms a valid set cover in $(\U \cup D, \S)$. Furthermore, following the proof of Lemma~\ref{main:static:approx}, we get: $f \cdot \sum_{e \in \U \cup D} w_e \geq |\T|$. Applying Corollary~\ref{main:cor:lm:worstcase:approx:2}, we now derive that:
\begin{equation}
\label{main:eq:last}
f \cdot \sum_{e \in \U} w_e \geq (1+2\epsilon(1+\epsilon))^{-1} \cdot f \cdot \sum_{e \in \U \cup D} w_e  \geq (1+2\epsilon(1+\epsilon))^{-1} \cdot |\T|.
\end{equation}
Since $\T$ is a set cover in $(\U \cup D, \S)$, it  also forms a set cover in $(\U, \S)$. Similarly, since the element-weights  $\{w_e\}_{e \in \U \cup D}$ form a fractional packing in $(\U \cup D, \S)$, the weights $\{w_e\}_{e \in \U}$ form a fractional packing in $(\U, \S)$. Thus, we have a set cover $\T$ and a fractional packing $\{w_e\}$ in $(\U, \S)$ whose sizes are within a $(1+2\epsilon(1+\epsilon))f$ multiplicative factor of each other, according to~(\ref{main:eq:last}). Hence, Lemma~\ref{main:lm:dual} implies that $\T$ forms a $(1+2\epsilon(1+\epsilon))f$-approximate minimum set cover in $(\U, \S)$. 

%% file: Preliminaries.tex
\paragraph{Remark:} There are a few  minor notational inconsistencies between the extended abstract and this part of the paper. However, we emphasize that {\em this part consists of a  self-contained full version of our two dynamic algorithms}, with every necessary notation and concept defined from scratch.

\paragraph{Organization:}
In Section~\ref{sec:preliminaries} we describe the preliminary concepts and notations, along with a static primal-dual algorithm we will be building upon in subsequent sections. We present our dynamic algorithm for amortized update time in Section~\ref{sec:amortized}, where Section~\ref{ssec:amortized-algorithm} contains the algorithm description and Section~\ref{ssec:potential-analysis} analyzes its  amortized update time. Our dynamic algorithm for worst case update time is presented in Section~\ref{sec:worst-case}.

\section{Preliminaries} \label{sec:preliminaries}

In the {\em minimum set cover} problem, we get a universe of elements $\U$ and a collection of sets $\S \subseteq 2^{\U}$  as input, where $\bigcup_{s \in \S} s = \U$ and each set $s \in \S$ has a {\em cost} $c_s > 0$ associated with it.
Without loss of generality, we assume there exists a parameter $C > 1$ such that $1/C < c_s < 1$ for all sets $s \in \S$, 
A collection of sets $\S' \subseteq \S$ forms a {\em set-cover} of $\U$ iff $\bigcup_{s \in \S'} s = \U$. The goal is to compute a set cover $\S'$ of $\U$ with minimum total cost $c(\S') = \sum_{s \in \S'} c_s$. 

\paragraph{Dynamic Set Cover.}
Initially, the algorithm receives as input a universe of elements $\U$, a collection of sets $\S \subseteq 2^{\U}$ with $\bigcup_{s \in \S} s = \U$, and a cost $c_s \geq 0$ for each set $s \in \S$. 
Subsequently, the input keeps changing via a sequence of updates, where each update either (1) deletes an element $e$ from the universe $\U$ and from every set $s \in \S$ that contains $e$, or (2) inserts an element $e$ into the universe $\U$ and specifies the sets in $\S$ that the element $e$ belongs to.
After each update, we would like to maintain an approximate cost of the optimal set cover of the updated set system.
The time taken by a dynamic algorithm to handle the preprocessing step is referred to as its {\em preprocessing time}.
We say that a dynamic algorithm has an {\em amortized update time} of $O(\lambda)$ iff the algorithm takes $O(\Gamma + t \cdot \lambda)$ total time to handle any sequence of $t \geq 1$ updates after preprocessing, where $O(\Gamma)$ is the preprocessing time.
We want to design a dynamic algorithm whose approximation ratio and amortized update time are as small as possible.

\medskip

Through out the remaining part of this paper, we use $m$ to denote the number of sets in $\S$, and $n$ to denote the maximum number of elements in the universe $\U$ across all the updates.
We use $f$ to denote an upper bound on the maximum frequency of any element across all the updates, where the frequency of an element $e \in \U$ is defined as the number of sets in $\S$ that contain $e$.

Fix $\epsilon > 0$ to be a sufficiently small constant, e.g., $\epsilon < 0.1$.

\paragraph{Hierarchical Decomposition.}
Let $L = \lceil \log_{1+\epsilon}(C n) \rceil + 1$ and $[L]=\{ 0,1,\ldots,L \}$.
We call $L$ the \emph{highest} level and $0$ the \emph{lowest}.
In the following we describe a static primal-dual algorithm assigns a level $l(s) \in [L]$ to every set $s\in \S$.
We also define the level $l(e)$ of element $e\in\U$ as $l(e) = \max_{s\in\S: e\in S} \{l(s)\}$.
The algorithm guarantees that each element $e$ at level $l(e) = k$ has weight $(1+\epsilon)^{-k}$.
We use $w(s) = \sum_{e\in s}w(e)$ to denote the weight of $s\in \S$, which is the total weight of elements it contains.
We call a set $s$ \emph{tight} if $w(s) > \frac{c_s}{1+\epsilon}$.

Initially, we set $l(e) \leftarrow L$ and $w(e) \leftarrow (1+\epsilon)^{-L}$ for all elements $e\in\U$ and let $S\leftarrow \S$ and $E \leftarrow \U$.
Note that at this point any set $s\in \S$ has weight 
\begin{equation*}
w(s) = \sum_{e\in s} w(e) \leq n\cdot (1+\epsilon)^{-L} \leq \frac{1}{C} \leq c_s.
\end{equation*}
Subsequently, the algorithm proceeds in rounds $i=L,L-1,\ldots,0$.
In round $i$, we identify sets in $S$ that become tight, and move them from $S$ to $S_i$.
We set $l(s) \leftarrow i$ for all $s\in S_i$.
Then we move the elements contained in the newly tight sets from $E$ to $E_i$.
Let $l(e) \leftarrow i$ for each $e\in E_i$.
Then we raise the weights of all the remaining elements in $E$ by a factor of $(1+\epsilon)$, and proceed to the next round.
The process stops when $E$ becomes empty.
Observe that  in round $0$, all elements in $E$ has weight $(1+\epsilon)^{0} = 1 > c_s$.
Hence every element will eventually be assigned some level in $[L]$.
We claim that the collection of tight sets is a valid set cover and is a $(1+\epsilon)f$-approximation.


\paragraph{Primal Dual Analysis.}
We abuse the notation slightly and use $w(X)$ to denote the total weight for any collection $X$ of sets or elements.
For any collection $X$ of sets, we use $c(X)$ to denote their total cost.
Let $\opt$ be the cost of the minimum set cover. We have the following lemma.

\begin{lemma}\label{lemma:primal-dual}
	If we have $w(s) \leq (1+\epsilon)\cdot c_s$ for every set $s\in \S$, then $w(\U) \leq (1+\epsilon)\cdot\opt$.
\end{lemma}
\begin{proof}
	Let $\S^*$ be the minimum set cover, i.e., $\opt = c(\S^*)$.
	Then we have
	\begin{equation*}
	w(\U) = \sum_{e\in \U}w(e) \leq \sum_{s\in \S^*}\sum_{e \in s}w(e) \leq (1+\epsilon)\cdot c(\S^*) = (1+\epsilon)\cdot \opt.
	\end{equation*}
	The first inequality comes from the fact that $\bigcup_{s\in \S^*} s = \U$, and the second inequality comes from the assumption of the lemma.
\end{proof}


Hence any set cover $S$ with $c(S)\leq \alpha\cdot w(\U)$ is an $\alpha(1+\epsilon)$-approximation.
In particular, in the above hierarchical decomposition, the collection of tight sets $T\subseteq \S$ is a valid set cover since each element is contained in at least one tight set.
Moreover, for every $s\in T$, we have $w(s) > \frac{c_s}{1+\epsilon}$.
Hence $T$ is a $(1+\epsilon)^2 f$-approximate set cover since
\begin{equation*}
c(T) \leq (1+\epsilon)\sum_{s\in T} w(s) \leq (1+\epsilon)f\cdot \sum_{e\in \U}w(e) \leq (1+\epsilon)^2 f\cdot \opt.
\end{equation*}

We will show that our dynamic algorithms maintains a similar hierarchical decomposition described above, and the collection of tight sets is a $(1+O(\epsilon))f$-approximate set cover.


%% file: Algorithm_Amortized.tex
\section{Our Algorithm for Amortized Update Time} \label{sec:amortized}

In this section we present the data structure and the algorithm with amortized $O\left(\frac{f^2}{\epsilon^3}+\frac{f}{\epsilon^2}\log C\right)$ update time.
Specifically, we prove the following theorem.

\begin{theorem} \label{th:amortized_time}
	There are deterministic dynamic algorithms for the minimum set cover problem with $(1+\epsilon)f$-approximation ratio and an amortized update time of $O\left(\frac{f^2}{\epsilon^3}+\frac{f}{\epsilon^2}\log C\right)$.
\end{theorem}

Our algorithm maintains the $(1+\epsilon)f$-approximate set cover \emph{value} after every update and can return it in constant time.
The algorithm also maintains a \emph{solution} of such value throughout the updates, and can output the change of the maintained solution after every update.

\paragraph{Notations.}
In addition to the \emph{real weight} $w(s) = \sum_{e\in s} w(e)$, we also maintain for every set $s$ a \emph{dead weight} $\phi(s)$.
We use $w^*(s) = w(s)+\phi(s)$ to denote its total weight.
The idea of introducing a dead weight for each set $s$ is to keep track of the weight that set $s$ has lost due to decreases of element weight and element deletions.
More specifically, when $w(s)$ decreases, we increase $\phi(s)$ such that $w(s)+\phi(s)$ does not decrease too much (so that we do not need to update our data structure immediately).
As long as the total dead weight of sets is small, e.g., at most an $\epsilon$ fraction of the total weight of sets, then our data structure maintains a solution with bounded approximation ratio.
We rebuild part of the data structure only when the fraction of dead weight gets too large.

\subsection{Invariants}

We define the weight of set $s$ \emph{at level $i$} as
\begin{equation*}
w(s,i) = \sum_{e\in s} (1+\epsilon)^{ -\max\{i, \max_{s'\neq s: e\in s'} l(s') \} }.
\end{equation*}
In other words, $w(s,i)$ is the weight of $s$ if we place $s$ at level $i$.
Note that $w(s) = w(s,l(s))$.

We maintain the following invariants.

\begin{invariant}[Bounded Weight Invariant]\label{inv:set-weight}
	$\forall s\in \S$, $w(s,l(s)+1) < c_s$.
\end{invariant}

For convenience we extend the definition of $w(s,i)$ to all positive integers $i$.
Hence the invariant holds for all sets $s$ at level $L$ since $w(s,L+1) \leq n\cdot (1+\epsilon)^{-(L+1)} < \frac{1}{C} < c_s$.

\begin{definition}[Tight Sets]
	We call a set $s$ \emph{tight} if $w(s)+\phi(s) > \frac{c_s}{1+\epsilon}$; \emph{slack} otherwise.
\end{definition}

\begin{invariant}[Tightness Invariant]\label{inv:level-of-tight}
	All sets at level at least $1$ are tight.
\end{invariant}

Note that a tight set can be at level $0$, but slack sets cannot be at level other than $0$.
With Invariant~\ref{inv:level-of-tight} maintained, the collection of tight sets is a feasible set cover.

\begin{corollary}\label{lemma:contain-tight}
	Every element is contained in at least one tight set.
\end{corollary}
\begin{proof}
	Suppose there exits an element $e$ such that all sets containing it are slack, then by Invariant~\ref{inv:level-of-tight}, we have $l(e) = 0$ and $w(e) = 1$.
	Hence each set $s$ containing $e$ has weight at least $1\geq c_s$, which contradicts the definition of slack sets.
\end{proof}

\begin{invariant}[Local $\phi$ Invariant]\label{inv:local-phi}
	If $w(s)+\phi(s) > c_s$ then $\phi(s) =0$.
\end{invariant}

Invariant~\ref{inv:set-weight} and~\ref{inv:local-phi} imply the following immediately.

\begin{corollary}\label{corollary:bounded-set-weight}
	We have $w(s) + \phi(s) < (1+\epsilon)\cdot c_s$ for all sets $s\in \S$.
\end{corollary}
\begin{proof}
	If $w(s)+\phi(s) \leq c_s$ then the corollary trivially holds.
	Otherwise by Invariant~\ref{inv:local-phi} we have $\phi(s) = 0$.
	By definition it is easy to prove that $w(s,i) \leq (1+\epsilon)\cdot w(s,i+1)$ for all $i\in [L-1]$.
	Hence by Invariant~\ref{inv:set-weight} we have $w(s) = w(s,l(s))\leq (1+\epsilon)\cdot w(s,l(s)+1) < (1+\epsilon)\cdot c_s$.
\end{proof}

Let $T_i$ be the collection of tight sets at level $i$.
Let $\Phi_i = \sum_{s: l(s)=i}\phi(s)$ be the total dead weight of sets at level $i$.
Let $\Phi_{\leq k} = \sum_{i= 0}^k \Phi_i$ and $T_{\leq k} = \sum_{i = 0}^k T_i$.
Let $\Phi = \Phi_{\leq L}$ and $T = T_{\leq L}$.
We maintain the following invariant, which guarantees that the total dead weight is bounded.

\begin{invariant}[Global $\phi$ Invariant]\label{inv:global-phi}
	We have $\Phi \leq \epsilon\cdot \big(c(T)+f\cdot w(\U) \big)$.
\end{invariant}

With all the invariants maintained, we have the following approximation ratio guarantee.

\begin{lemma}\label{lemma:approximation-ratio}
	The collection of tight sets $T$ is a $(1+5\epsilon)f$-approximate set cover.
\end{lemma}
\begin{proof}
	By definition of tight sets, we have $w(s)+\phi(s) > \frac{c_s}{1+\epsilon}$ for all $s\in T$.
	Hence the cost of the set cover is (where the third inequality comes from Invariant~\ref{inv:global-phi})
	\begin{align*}
	c(T) & < (1+\epsilon)\cdot\sum_{s\in T}(w(s)+\phi(s)) \leq (1+\epsilon)\cdot w(\S) + (1+\epsilon)\cdot \Phi \\
	& \leq (1+\epsilon)f\cdot w(\U) + \epsilon(1+\epsilon)\cdot c(T) + \epsilon(1+\epsilon)f\cdot w(\U) \\
	& = (1+\epsilon)^2 f\cdot w(\U) + \epsilon(1+\epsilon)\cdot c(T).
	\end{align*}
	
	Hence we have (for every $\epsilon\in (0,0.1)$)
	\begin{equation*}
	c(T) < \frac{(1+\epsilon)^2 f}{1-\epsilon(1+\epsilon)}\cdot w(\U) \leq \frac{(1+\epsilon)^3 f}{1-\epsilon(1+\epsilon)}\cdot \opt \leq (1+5\epsilon)f\cdot \opt,
	\end{equation*}
	where the second inequality comes from the Lemma~\ref{lemma:primal-dual}
\end{proof}

Recall that $E_i = \{e\in\U : l(e)=i\}$ contains the elements at level $i$.
Let $E_i(s) = s\cap E_i$ be elements in $s$ that are at level $i$.
We maintain for every $s$ the sets $E_i(s)$ for all $i\in[L]$.
Recall that each element appears in at most $f$ sets, and our data structure maintains for each element a pointer to each set containing it.
It is easy to check throughout the algorithm that the data structure can be maintained efficiently when elements change their levels.
Note that for all $i<l(s)$, $E_i(s) = \emptyset$.

Finally, we introduce the notion of base level.

\begin{definition}[Base Level]
	For each set $s$, let $b(s) = \lfloor \log_{1+\epsilon}\frac{1}{c_s} \rfloor$ be the \emph{base level} of $s$.
\end{definition}

Note that all base levels are at most $\log_{1+\epsilon}C = O(\frac{1}{\epsilon}\log C)$, since $c_s\in (\frac{1}{C},1)$ for all $s\in \S$.

\begin{lemma}\label{lemma:not-below-base}
	Invariant~\ref{inv:set-weight} implies $E_{i}(s) = \emptyset$ for all $s\in \S$ and $i < b(s)$.
\end{lemma}
\begin{proof}
	Fix any set $s$. Each element at level $i \leq b(s)-1$ has weight at least
	\begin{equation*}
	(1+\epsilon)^{-(b(s)-1)} \geq (1+\epsilon)\cdot c_s.
	\end{equation*}
	
	Hence if $E_{i}(s)$ is not empty, then $w(s) \geq (1+\epsilon)\cdot c_s$, which violates Corollary~\ref{corollary:bounded-set-weight}.
\end{proof}

As we will present in the next section, our main algorithm consists of four subroutines, namely \textsf{Insert}, \textsf{Delete}, \textsf{Promote} and \textsf{Rebuild}.
We summarize the invariants in the following table.

\begin{center}
	\begin{tabular}{ |c|c|c|c| }
		\hline
		Invariant & Property & Affected by & Fixed (by) \\
		\hline
		Bounded Weight (Inv.~\ref{inv:set-weight}) & $\forall s$: $w(s,l(s)+1)<c_s$ & \textsf{Insert} & \textsf{Promote} \\
		\hline
		Tightness (Inv.~\ref{inv:level-of-tight}) & $\forall u$: $l(s)\geq 1 \rightarrow w(s)+\phi(s) >\frac{c_s}{1+\epsilon}$ & All subroutines & immediately \\
		\hline
		Local $\phi$ (Inv.~\ref{inv:local-phi}) & $\forall s$: $w(s)+\phi(s)>c_s \rightarrow \phi(s)=0$ & All subroutines & immediately \\
		\hline
		Global $\phi$ (Inv.~\ref{inv:global-phi}) & $\Phi \leq \epsilon\cdot \big(c(T)+f\cdot w(\U) \big)$ & \textsf{Promote}, \textsf{Delete} & \textsf{Rebuild} \\
		\hline
	\end{tabular}
\end{center}

\subsection{Algorithm} \label{ssec:amortized-algorithm}

In this section we describe the algorithms to handle the updates.
In addition to \textsf{Insert}($e$) and \textsf{Delete}($e$) that handle element insertions and deletions, we introduce two subroutines, \textsf{Promote}$(s)$ and \textsf{Rebuild}$(k)$, the former of which promotes a set $s$ when its total weight gets too large, and the later rebuilds all levels on or below $k$ when there are too much dead weight at levels at most $k$.
Roughly speaking, the promotions of sets maintain the bounded weight invariant (Invariant~\ref{inv:set-weight}) and the rebuilds of levels maintain the global $\phi$ invariant (Invariant~\ref{inv:global-phi}).
We maintain the tightness invariant (Invariant~\ref{inv:level-of-tight}) and local $\phi$ invariant (Invariant~\ref{inv:local-phi}) in all subroutines (namely, \textsf{Delete}, \textsf{Insert}, \textsf{Promote} and \textsf{Rebuild}).

The main algorithm is shown as follows (refer to Algorithm~\ref{alg:complete}).

\begin{algorithm}[H]
	\caption{\textsf{DynamicsetCover}}
	\label{alg:complete}
	\begin{algorithmic}[1]
		\State initialize $l(s) = 0$ for all $s\in \S$
		\For{each update}
		\If{element $e$ is deleted}
		\State \textsf{Delete}$(e)$.
		\EndIf
		\If{element $e$ is inserted}
		\State \textsf{Insert}$(e)$.
		\While{exits set $s$ for which Invariant~\ref{inv:set-weight} is violated}
		\State \textsf{Promote}$(s)$.
		\EndWhile
		\EndIf
		\While{Invariant~\ref{inv:global-phi} is violated}
		\State find the smallest $k$ for which $\Phi_{\leq k} > \epsilon\cdot (c(T_{\leq k})+f\cdot w(E_{\leq k}))$.
		\State \textsf{Rebuild}($k$).
		\EndWhile
		\EndFor
	\end{algorithmic}
\end{algorithm}

As we will show later, the bounded weight invariant (Invariant~\ref{inv:set-weight}) can be violated for a set $s$ only after an element contained in $s$ is inserted.
In other words, line 8 of Algorithm~\ref{alg:complete} will be executed on set $s$ only if $s$ contains the element $e$ that is inserted.

In the following we describe the four subroutines one by one.

\subsubsection{Deletions}

Suppose element $e$ is deleted.
By definition, each $s$ containing $e$ has its real weight $w(s)$ decreases by $w(e)$.
On the other hand, we increase its dead weight $\phi(s)$ accordingly.
We also exclude $e$ from the list of elements $E_{l(e)}(s)$ maintained by each set $s$ containing $e$\footnote{Every time when an element changes its level, or gets inserted or deleted, we change these element collections (and necessary pointers) accordingly. For ease of presentation we omit these steps in the description of all our algorithms.}.
Since each element is contained in at most $f$ sets and it takes $O(1)$ time to handle each set containing $e$, the following algorithm runs in $O(f)$ time.

\begin{algorithm}
	\caption{\textsf{Delete$(e)$}}
	\label{alg:deletion}
	\begin{algorithmic}[1]
		\For{each $s$ containing $e$}
		\State $w(s) \leftarrow w(s)-w(e)$.
		\If{$l(s)>0$}
		\State $\phi(s) \leftarrow \phi(s) + w(e)$.
		\If{$w(s)+\phi(s) > c_s$}
		\State $\phi(s) \leftarrow \max\{0,c_s-w(s)\}$.
		\EndIf
		\EndIf
		\EndFor
	\end{algorithmic}
\end{algorithm}

When an element $e$ is deleted, for each $s$ containing $e$ our algorithm changes $w(e)$ unit of real weight of $s$ to dead weight.
If local $\phi$ invariant (Invariant~\ref{inv:local-phi}) is violated then we decreases $\phi(s)$: we have either $w(s)+\phi(s) \leq c_s$ or $\phi(s) = 0$ after line 6 of Algorithm~\ref{alg:deletion}.
Note that if a set is at level $0$, then we do not increase its dead weight (which stays $0$) since we do not need to maintain its tightness.
Consequently each set at level at least $1$ remains tight after the deletion.
That is, the tightness invariant (Invariant~\ref{inv:level-of-tight}) is maintained.
Note that the bounded weight invariant (Invariant~\ref{inv:set-weight}) also holds since $w(s,l(s)+1)$ does not increase after the deletion.
However, since the total dead weight is increased, the global $\phi$ invariant (Invariant~\ref{inv:global-phi}) can possibly be violated.
As we will show later, the invariant will be maintain by rebuilding some levels.

\subsubsection{Insertions}

Suppose element $e$ is inserted.
We consider any $s$ containing $e$.
If we have $w(s) \leq c_s$ after the insertion then we can easily fix all the invariants (regarding set $s$) by decreasing $\phi(s)$ (if necessary).
Suppose that after the insertion $w(s) > c_s$.
As we will show in our potential analysis, we can upper bound the increase in the potential of $s$ by $O(\frac{f^2}{\epsilon^3}+\frac{f}{\epsilon^2}\log C)$ if $E_{l(s)}(s)\neq \emptyset$ \emph{before} the insertion.

We call a set \emph{good} if one of the following properties holds, and \emph{bad} otherwise.

\begin{definition}[Good Sets]
	Suppose we insert a new element $e$ at level $k$. We call set $s$ containing $e$ \emph{good} before the insertion if (1) $w(s)+(1+\epsilon)^{-k} \leq c_s$; or (2) $E_{l(s)}(s)\neq \emptyset$.
\end{definition}

Before inserting a new element $e$, we would like to make sure that all sets containing $e$ are good.
Note that the definition of ``good'' depends on the level the new element is inserted into.
The higher level the element $e$ is inserted into, the more sets containing $e$ will be good.
We can insert $e$ at a high level only by increasing the levels of sets containing $e$, before the element is inserted.

Our algorithm increases the level of each bad set $s$ to the lowest level $i$ on which it becomes good.
As $s$ does not contain any element at the level on which it is bad, increasing $l(s)$ does not change the weight of any element or set, and thus can be done in $O(1)$ time.

We call such an increase of $l(s)$ a \emph{lift-up} of $s$, and we charge $s$ one unit of \emph{lift-up cost}.
Note that if $l(s) < b(s)$, then we can directly increase $l(s)$ to $b(s)$, since by Lemma~\ref{lemma:not-below-base} $E_i(s) = \emptyset$ for all $i < b(s)$, and $w(s)+(1+\epsilon)^{-i}>c_s$ for all $i<b(s)$.

We denote the operation (which increases $l(s)$ by at least one) by \textsf{Lift-Up}$(s)$.

\begin{algorithm}
	\caption{\textsf{Lift-Up$(s)$}}
	\label{alg:lift-up}
	\begin{algorithmic}[1]
		\If{$l(s) < b(s)$}
		\State $l(s) \leftarrow b(s)$.
		\Else
		\State $l(s) \leftarrow l(s)+1$.
		\EndIf
	\end{algorithmic}
\end{algorithm}

A lift-up of $s$ will be called only when $E_{l(s)}(s) = \emptyset$.
In addition, for potential analysis purpose (which will be clear later), we lift a set $s$ up only when $\phi(s) = 0$.
Thus before lifting the set $s$ up, we need to remove its dead weight.
On the other hand, since $w(s)+\phi(s)$ is decreased when we set $\phi(s)$ to $0$, we need to ensure that $s$ is tight after the insertion of the element, as we have $l(s) > 0$ after the lift-up.
More specifically, suppose we remove the dead weight of $s$ and its real weight before the insertion is $w(s)$.
Then our algorithm guarantees that the element is inserted at a level $k$ such that $w(s)+(1+\epsilon)^{-k} >\frac{c_s}{1+\epsilon}$, i.e., $s$ is tight after the insertion.

The detailed description can be found in Algorithm~\ref{alg:insertion}.
We lift sets up in a carefully chosen order, such that all sets are good before the element is inserted.

More specifically, let $k$ be the tentative level for the new element.
As the first step, we identify the collection of bad sets $B$ containing $e$.
Note that since our algorithm does not decrease the tentative level $k$, all sets containing $e$ that are not in $B$ will be good when $k$ settles at its final level.
We consider the sets in $B$ one by one in non-decreasing order of $c_s - w(s)$.
Intuitively, the set $s$ with the smallest $c_s-w(s)$ is most sensitive to the increase in weight due to the element insertion, and hence should be considered first.

Suppose a set $s\in B$ is considered.
We lift $s$ up until it becomes good, i.e., when either (1) $E_{l(s)}(s)\neq \emptyset$; or (2) the lift-up of $s$ increases the tentative level $k$ of $e$ such that $w(s)+(1+\epsilon)^{-k} \leq c_{s}$.
In the first case $s$ stops at level $l(s)$ and will be good no matter what the final value of $k$ will be.
In the later case all sets contained in $B$ are good, as $s$ is considered as the bad set with the minimum $c_{s}-w(s)$, and with the current tentative level $k$ it holds that $w(s) + (1+\epsilon)^{-k} \leq c_{s}$.
In other words, $k$ is so large that each set $s' \in B$ not considered yet has $w(s')+(1+\epsilon)^{-k} \leq c_{s'}$, i.e., $s'$ is good.
Hence no more lift-up will happen.
This is the reason why we consider sets in the order specified in line~3 of Algorithm~\ref{alg:insertion}.

\begin{algorithm}[htb]
	\caption{\textsf{Insert$(e)$}}
	\label{alg:insertion}
	\begin{algorithmic}[1]
		\State $k \leftarrow \max_{s: e\in s}\{l(s)\}$.
		\State Let $B = \{s_1,s_2,\ldots,s_b\}$ be the collection of bad sets containing $e$, where $b = |B|\leq f$.
		\State Sort the sets in $B$ and assume w.l.o.g. that $c_{s_1}-w(s_1)\leq c_{s_2}-w(s_2)\leq \ldots \leq c_{s_b}-w(s_b)$.
		\For{$i=1,2,\ldots,b$}
		\While{$s_i$ is bad}
		\State $\phi(s_i) \leftarrow 0$, \textsf{Lift-Up}$(s_i)$.
		\State $k \leftarrow \max\{k,l(s_i)\}$.
		\EndWhile
		\EndFor
		\State $l(e) \leftarrow k, w(e) \leftarrow (1+\epsilon)^{-k}$.
		\For{each $s$ containing $e$}
		\State $w(s) \leftarrow w(s) + w(e)$.
		\If{$w(s)+\phi(s) > c_s$}
		\State $\phi(s) \leftarrow \max\{0,c_s-w(s)\}$.
		\EndIf
		\EndFor
	\end{algorithmic}
\end{algorithm}

Note that it takes $O(f\log f)$ time to sort the sets in $B$ in line 3 of Algorithm~\ref{alg:insertion}.
Hence excluding the lift-up costs (which are charged to the sets), Algorithm~\ref{alg:insertion} runs in $O(f \log f)$ time.
Note that to ease the analysis we also consider line 7 to be part of the list-up cost.

Since the insertion does not create any dead weight or change the level of any existing element, the global $\phi$ invariant (Invariant~\ref{inv:global-phi}) remains valid.
The local $\phi$ invariant (Invariant~\ref{inv:local-phi}) is maintained because if $w(s)+\phi(s) > c_s$ after the insertion then we decrease $\phi(s)$ such that either $\phi(s)= 0$ or $w(s)+\phi(s) = c_s$.

Next we show that the tightness invariant (Invariant~\ref{inv:level-of-tight}) is maintained, and all sets are good before the element is inserted.
Recall that we remove the dead weight of each set that is lifted up, which might make a set slack.
We show in Lemma~\ref{lemma:before-insertion} that each set whose dead weight is removed will be tight after the insertion of the new element (which maintains the tightness invariant).

\begin{lemma}\label{lemma:before-insertion}
	Before the new element $e$ is inserted (in line 8 of Algorithm~\ref{alg:insertion}), all sets containing $e$ are good.
	Moreover, the tightness invariant is maintained after the insertion.
\end{lemma}
\begin{proof}
	As we have argued in the above discussion, since the tentative level $k$ does not decrease, if a set is good at some point, then it remains good afterwards.
	Hence our algorithm ensures that all sets containing $e$ are good before the element is inserted.
	
	Next we prove that the tightness invariant is maintained by showing that all sets lifted up are tight after the element insertion.
	More specifically, let $k^*$ be the final level of $e$, i.e., $w(e) = (1+\epsilon)^{-k^*}$.
	We show that if we lift set $s$ up, then $w(s)+w(e) > \frac{c_s}{1+\epsilon}$ at the end of the algorithm.
	For sets not lifted up, the tightness invariant trivially holds since if a set is at level at least $1$ then it is tight before the element is inserted, and the algorithm does not decrease its real weight or dead weight.
	
	Recall that in the while loop in which $s$ is chosen, our algorithm increases the level $l(s)$ of $s$ (which may also increase $k$) until either (1) $E_{l(s)}(s) \neq \emptyset$, or (2) $w(s) + (1+\epsilon)^{-k} \leq c_s$.
	
	If the algorithm never increases $k$, then the lemma easily follows since each set $s_i \in B$ is bad initially, which implies $w(s_i)+(1+\epsilon)^{-k^*} > c_{s_i}$.
	Otherwise we consider the last while loop in which $k$ is increased (to $k^*$) because some set $s_i$ is lifted up.
	
	If $s_i$ is lifted up from level $0$ to its base level, then we have
	\begin{equation}\label{equation:s_i}
		w(s_i) + (1+\epsilon)^{-k^*} = w(s_i) + (1+\epsilon)^{-b(s)} \geq c_{s_i} > \frac{c_{s_i}}{1+\epsilon}.
	\end{equation}
	
	Otherwise $s_i$ is lifted up by exactly one level, and we have
	\begin{equation*}
		w(s_i) + (1+\epsilon)^{-(k^*-1)} > c_{s_i},
	\end{equation*}
	which implies
	\begin{equation*}
		w(s_i) + (1+\epsilon)^{-k^*} \geq \frac{w(s_i)}{1+\epsilon} + (1+\epsilon)^{-k^*} > \frac{c_{s_i}}{1+\epsilon}.
	\end{equation*}
	
	In both case we have $w(s_i) + (1+\epsilon)^{-k^*} > \frac{c_{s_i}}{1+\epsilon}$.
	That is, $s_i$ is tight after the insertion.
	Next we show that every $s_j\neq s_i$ that is lifted up is also tight after the insertion. 
	
	Consider any $s_j$ with $j>i$.
	That is, $s_j$ is considered after $s_i$.
	Since $k=k^*$ when $s_j$ is considered in the while loop, if $s_j$ is lifted up, then $s_j$ is bad and we have $w(s_j)+(1+\epsilon)^{-k^*} > c_{s_j}$, as required.
	
	Now consider any $s_j$ with $j<i$, i.e., $s_j$ is considered before $s_i$.
	Then we have $c_{s_j} - w(s_j) \leq c_{s_i}-w(s_i)$, which implies
	\begin{equation*}
		w(s_j) + (1+\epsilon)^{-k^*} \geq \frac{w(s_j)}{1+\epsilon} + (1+\epsilon)^{-k^*}
		\geq \frac{c_{s_j} - c_{s_i}+w(s_i)}{1+\epsilon} + (1+\epsilon)^{-k^*} > \frac{c_{s_j}}{1+\epsilon},
	\end{equation*}
	where the last inequality comes from Inequality~\eqref{equation:s_i}.
\end{proof}

The bounded weight invariant (Invariant~\ref{inv:set-weight}) might be violated due to the increase of set weight, in which case we invoke the promotion subroutine as follows.

\subsubsection{Promotion}

Recall that if the bounded weight invariant does not hold for set $s$, then we have $w(s,l(s)+1)\geq c_s$ and $\phi(s) = 0$ (by the local $\phi$ invariant).
In this case we call \textsf{Promote}($s$), which increases the level of $s$, and decreases $w(e)$ for all $e\in E_{l(s)}(s)$.
In the meanwhile, the subroutine increases the dead weight $\phi(s')$ of each $s'\neq s$ containing $e$ accordingly to compensate for the decrease in $w(s')$, which maintains the tightness invariant (Invariant~\ref{inv:level-of-tight}).

If $E_{l(s)}(s) = \emptyset$ then we first lift $s$ up to a level $i$ such that $E_i(s)\neq \emptyset$. 
As before, we charge the lift-up cost to set $s$.

\begin{algorithm}[hbt]
	\caption{\textsf{Promote$(s)$}}
	\label{alg:promotion}
	\begin{algorithmic}[1]\While{$E_{l(s)}(s)=\emptyset$}
		\State \textsf{Lift-Up}$(s)$.
		\EndWhile
		\State $k \leftarrow l(s)$.
		\While{$w(s,k+1)\geq c_s$}
		\For{each $e\in E_{k}(s)$}
		\State $l(e) \leftarrow k+1$, $w(e) \leftarrow (1+\epsilon)^{-(k+1)}$.
		\State $w(s) \leftarrow w(s) - \epsilon(1+\epsilon)^{-(k+1)}$.
		\For{each $s'\neq s$ containing $e$}
		\State $w(s') \leftarrow w(s') - \epsilon(1+\epsilon)^{-(k+1)}$
		\If{$l(s')>0$}
		\State $\phi(s') \leftarrow \phi(s')+\epsilon(1+\epsilon)^{-(k+1)}$.
		\If{$w(s')+\phi(s') > c_{s'}$}
		\State $\phi(s') \leftarrow \max\{0,c_{s'}-w(s')\}$.
		\EndIf
		\EndIf
		\EndFor
		\EndFor
		\State $l(s)\leftarrow k+1$, $k\leftarrow k+1$.
		\EndWhile
	\end{algorithmic}
\end{algorithm}

If $E_{k}(s)\neq \emptyset$, the update time for promoting set $s$ from level $k$ to $k+1$ is $O(f\cdot |E_k(s)|)$.

As before, the local $\phi$ invariant (Invariant~\ref{inv:local-phi}) is maintained since if $w(s')+\phi(s') > c_{s'}$ then we decrease $\phi(s')$ such that either $\phi(s') = 0$ or $w(s')+\phi(s') = c_{s'}$ (in line 13 of Algorithm~\ref{alg:promotion}).
As the promotion does not decrease $w(x)+\phi(x)$ to below $c_x$ for $x=s$ or any set at level at least $1$ intersecting $s$, the tightness invariant (Invariant~\ref{inv:level-of-tight}) remains valid.
By repeatedly calling the promotion subroutine, we can also maintain the bounded weight invariant (Invariant~\ref{inv:set-weight}).

The global $\phi$ invariant (Invariant~\ref{inv:global-phi}) might be violated due to the increase of dead weight, in which case we invoke the rebuild subroutine as follows.

\subsubsection{Rebuild}\label{ssec:rebuild}

Our algorithm maintains $\Phi_i$, $T_i$, $E_i$ and $c(T_i)$ explicitly for all $i\in[L]$.
In addition, we maintain $\Phi = \Phi_{\leq L}$, $c(T) = c(T_{\leq L})$ and $w(\U) = \sum_{e\in \U}w(e)$.
Recall that we perform a rebuild only when the global $\phi$ invariant (Invariant~\ref{inv:global-phi}) does not hold, i.e., $\Phi > \epsilon\cdot (c(T)+f\cdot w(\U))$, which can be checked in $O(1)$ time.
When $\Phi > \epsilon\cdot (c(T)+f\cdot w(\U))$, we find the smallest $k$ for which $\Phi_{\leq k} > \epsilon\cdot (c(T_{\leq k})+f\cdot w(E_{\leq k}))$, and rebuild all levels $i\leq k$.
Trivially such a level can be found in $O(k)$ time (by scanning the levels in a bottom-up manner).
In the following, we show that such a level can be found in $O(|T_{\leq k}|)$ time (which is smaller than $O(k)$ when most of the $T_i$'s are empty).
Note that $O(|T_{\leq k}|) = O(f\cdot |E_{\leq k}|)$.

As we will show later, \textsf{Rebuild}$(k)$ takes time $O(\frac{f}{\epsilon^2}\cdot |E_{\leq k}| + \frac{1}{\epsilon}\log C)$, and we consider the time spent on finding the level $k$ to rebuild part of the rebuild cost.

\paragraph{Skipping Empty Levels.}
We keep the pointers to $T_0,T_1,\ldots,T_L$ as an array of size $L+1$.
Additionally we organize the pointers that point to non-empty sets in a doubly linked list.
Note that if there exists an element at level $i$ then there must also exists a set at level $i$.
Hence for all $i\geq 1$ if $E_i$ is non-empty then $T_i$ is also non-empty.
In other words, we only need to look at the levels $i$ for which $T_i$ are non-empty (following the linked list), and compute $\Phi_{\leq k}$, $c(T_{\leq k})$ and $w(E_{\leq k}) = \sum_{i=1}^k |E_i|\cdot (1+\epsilon)^{-i}$ in $O(|T_{\leq k}|)$ time.

\paragraph{Maintaining the Linked List.}
Note that deletions can change a set from tight to slack only if the set is at level $0$.
Hence it suffices to check whether $T_0$ is empty after each deletion, and update the linked list in $O(1)$ time.
After \textsf{Rebuild}$(k)$, we can reconstruct $T_i$ for all $i\in[k]$, and update the linked list in $O(f\cdot |E_{\leq k}|)$ time (which becomes part of the rebuild cost).
Similar to deletions, the insertion of an element can only change $T_0$, and the linked list can be updated in $O(1)$ time.
However, the insertion may trigger lift-ups or promotions of sets, which may change the collections of tight sets and the linked list.

The case when a set changes its level by one is easy.
Suppose a set $s$ is promoted or lifted up from level $k$ to $k+1$.
Then only $s$ changes its level, and only tight sets in level $0$ can change to slack.
Moreover, the collections of tight sets at each level $i\geq 1$ remain unchanged.
Thus it suffices to check $T_0$, $T_k$ and $T_{k+1}$, and update the linked list in $O(1)$ time.

Now suppose a set $s$ at level $0$ is lifted up to its base level $b(s)$, which happens when a new element $e$ contained in $s$ is inserted, and $w(s) + w(e) > 1$.
If $s$ was tight at level $0$ then we first check whether $T_0$ is empty after the lift-up, and update the linked list.
If $T_{b(s)}$ was non-empty before $s$ is lifted up then we can simply include $s$ in $T_{b(s)}$ without modifying the linked list further.
Otherwise we find the largest $i \leq b(s)$ for which $T_i$ is non-empty (which takes $O(b(s))$ time), and insert the pointer to $T_{b(s)}$ to the linked list in $O(1)$ time.
Recall that $O(b(s)) = O(\frac{1}{\epsilon}\log C)$, and we charge this update cost to the insertion.
Since each insertion causes $O(f)$ sets at level $0$ being lifted up, each insertion is charged at most $O(\frac{f}{\epsilon}\log C)$ due to the maintenance of the linked list. 

\paragraph{Properties after Rebuild.}
Let $S$ be the sets at level at most $k$ before \textsf{Rebuild}($k$) is triggered.
We show that after \textsf{Rebuild}($k$), we have the following properties, where $l(s)$ denotes the level of set $s$ after the rebuild.
\begin{itemize}
	\item For all $s\in  S$, $l(s)\leq k$ and $\phi(s) = 0$.
	
	\item For all $s\in  S$ and $l(s) = k$, $w(s) > \frac{c_s}{1+\epsilon}$ and $w(s,k+1)<c_s$.
	
	\item For all $s\in  S$ and $l(s) < k$, $w(s) < c_s$. Moreover, if $l(s)>0$ then $w(s) > \frac{c_s}{1+\epsilon}$.
\end{itemize}

Consequently, the bounded weight invariant, tightness invariant and local $\phi$ invariant are maintained.
Moreover, since the subroutine removes all dead weight at level at most $k$, by repeatedly calling the subroutine (at higher and higher levels), the global $\phi$ invariant can also be maintained. 

\medskip

Next, we describe the details of the rebuild subroutine, which is very similar to the rebuild subroutine of~\cite{focs/BhattacharyaHN19}.
On a high level, the subroutine has two steps.
In the first step, we put all elements and sets at level $k$, and remove all dead weight, which decrease the total weight of some sets.
Then in the second step we move the slack sets to lower levels (which increases the weight of elements and sets) until they become tight, or reach level $0$.
The second step is done by calling the subroutine \textsf{Fix-Level}, which we borrow from~\cite{focs/BhattacharyaHN19}.

\begin{lemma}\cite[Section B.1]{focs/BhattacharyaHN19}
	There exists a subroutine \textsf{Fix-Level} with the following properties.
	\textsf{Fix-Level}$(k,{S'},{E'})$ takes as input a collection of sets ${S'}$ and a collection of elements ${E'}$, both are at level $k$.
	Moreover, each set $s\in {S'}$ has $w(s) < c_s$ and $\phi(s) = 0$.
	The subroutine places each set $s\in {S'}$ at level $l(s)$ such that (1) $w(s) < c_s$ and (2) if $l(s)>0$ then $w(s) > \frac{c_s}{1+\epsilon}$.
	The subroutine runs in time $O(f\cdot |{E'}| + k)$.
\end{lemma}

We call an element \emph{slack} if all sets containing it are slack.

\begin{algorithm}[hbt]
	\caption{\textsf{Rebuild$(k)$}}
	\label{alg:rebuild}
	\begin{algorithmic}[1]
		\State let $\hat{S} = \{s: l(s)\leq k\}$ be the sets at levels at most $k$.
		\Comment{$|\hat{S}|=O(f\cdot |E_{\leq k}|)$}
		\For{each $s\in \hat{S}$}
		\State $\phi(s) \leftarrow 0$, $l(s) \leftarrow k$.
		\EndFor
		\For{each $e\in E_{\leq k}$}
		\State $l(e) \leftarrow k$, $w(e) \leftarrow (1+\epsilon)^{-k}$ (and decrease $w(s)$ by the same amount for all $s$ containing $e$).
		\EndFor
		\State let $S' \leftarrow \{ s\in \hat{S}: w(s) \leq \frac{c_s}{1+\epsilon} \}$ be the collection of slack sets.
		\State let $E'\subseteq E_{\leq k}$ be slack elements.
		\State let $k' = \min\{ k, \log_{1+\epsilon}\frac{2C\cdot |E'|}{\epsilon} \}$.
		\For{each $e\in S'$}
		\State $l(s) \leftarrow k'$.
		\EndFor
		\For{each $e\in E'$}
		\State $l(e) \leftarrow k'$, $w(e) \leftarrow (1+\epsilon)^{-k'}$ (and decrease $w(s)$ by the same amount for all $s$ containing $e$).
		\EndFor
		\State \textsf{Fix-Level}$(k',S',E')$. \Comment{$O(f\cdot |E'| + k')$ time}
	\end{algorithmic}
\end{algorithm}

A key difference between our algorithm and that of~\cite{focs/BhattacharyaHN19} is, before calling the subroutine \textsf{Fix-Level} on the collection of slack sets $S'$ and the slack elements $E'$, we first move all elements in $E'$ and sets in $S'$ to level $k' = \min\{ k, \log_{1+\epsilon}\frac{2C\cdot |E'|}{\epsilon} \}$.
As we will show in Lemma~\ref{lemma:before-fix-level}, the operation does not break any invariants, but helps improve the update time for \textsf{Fix-Level}.

Note that since each set $s\in S'$ has no dead weight, its total weight equals to its real weight.

\begin{lemma}\label{lemma:before-fix-level}
	Before calling \textsf{Fix-Level}$(k',S',E')$ in line 13, each $s\in S'$ has weight less than $c_s$.
\end{lemma}
\begin{proof}
	In line 6 of Algorithm~\ref{alg:rebuild}, we identify the slack sets, i.e., sets $s$ with $w(s)\leq \frac{c_s}{1+\epsilon}$.
	We then move all slack sets and elements containing only slack sets to level $k'$.
	If $k'= k$ then no movement is made and the lemma trivially holds.
	
	Suppose $k' = \log_{1+\epsilon}\frac{2C\cdot |E'|}{\epsilon} < k$.
	Note that each $e\in E'$ has weight $(1+\epsilon)^{-k'} = \frac{\epsilon}{2C\cdot |E'|}$ on level $k'$.
	Hence the movement increases the weight of each $s\in  S'$ by at most $\frac{\epsilon}{2C}\leq \frac{\epsilon\cdot c_s}{2}$, as there are at most $|E'|$ elements contained in $s$ that are moved to level $k'$.
	Since $w(s) \leq \frac{c_s}{1+\epsilon}$ before the movement, set $s$ has weight $\leq c_s \big( \frac{1}{1+\epsilon}+\frac{\epsilon}{2} \big) < c_s$ after the movement of elements, as claimed.
\end{proof}

By definition $k' \leq \log_{1+\epsilon}\frac{2C|E'|}{\epsilon} = O(\frac{|E'|}{\epsilon^2}+\frac{1}{\epsilon}\log C)$, and thus \textsf{Fix-Level}$(k',S',E')$ runs in time $O(f\cdot |E'| + k') = O(\frac{f}{\epsilon^2}\cdot |E'|+\frac{1}{\epsilon}\log C)$.
In summary, \textsf{Rebuild}$(k)$ runs in $O(\frac{f}{\epsilon^2}\cdot |E_{\leq k}|+\frac{1}{\epsilon}\log C)$ time, and cleans up all dead weight at levels at most $k$.
Moreover, the weight invariant and the tightness invariant are maintained.
By repeatedly invoking the subroutine, we also maintain the global $\phi$ invariant.

%% file: Potential_Amortized.tex
\subsection{Potential Analysis and Amortized Update Time} \label{ssec:potential-analysis}

We use a potential analysis to bound the total update time.
On a high level, we have two types of potentials.
Each set $s$ has potential if $w(s) > c_s$, which we call \emph{up potential}.
Each unit of dead weight also has some potential, which we call \emph{down potential}.
We show that
\begin{itemize}
	\item one unit of down potential will be used to pay for $O(1)$ update time (in Section~\ref{ssec:potential-rebuild});
	\item one unit of up potential will be transformed into at most one unit of down potential and in the meantime pay for $O(1)$ update time (in Section~\ref{ssec:potential-promote}).
\end{itemize}

Specifically, the up potential is released to pay for the update cost of promotions of sets and the increase in down potential due to the promotions.
The down potential is released to pay for update time of the rebuilds of levels and the lift-up costs.
In other words, the up and down potential are consumed during promotions and rebuilds.
Hence it suffices to upper bound the amount of potential we gain from element insertions and deletions.
We show in Section~\ref{ssec:potential-insert-delete} that
\begin{itemize}
	\item each insertion does not increase the total down potential, and increases the total up potential by $O(\frac{f^2}{\epsilon^3}+\frac{f}{\epsilon^2}\log C)$;
	\item each deletion does not increase the total up potential, and increases the total down potential by $O(\frac{f}{\epsilon^3}+\frac{1}{\epsilon^2}\log C)$.
\end{itemize}

Since each unit of potential is used to pay for $O(1)$ units of update time, the amortized update time is bounded by $O(\frac{f^2}{\epsilon^3}+\frac{f}{\epsilon^2}\log C)$.

\subsubsection{Potential Function}

We first define the up potential, i.e., potentials of sets $s$ with weight $w(s) > c_s$.
Recall that the local $\phi$ invariant guarantees $\phi(s) = 0$ for each set $s$ with $w(s) > c_s$.

If set $s$ has weight $w(s) \leq c_s$ then it has $0$ potential; otherwise its potential is given as follows.
Recall that $w(s,i)$ is the weight of $s$ if we set $l(s) = i$.
Hence $w(s,i)$ is a non-increasing function of $i$.
Suppose $w(s) = w(s,l(s)) > c_s$.
Let $\hat{l}(s)$ be the highest level such that $w(s,\hat{l}(s))\geq c_s$.
Let
\begin{equation*}
	\alpha_i = \left( \frac{f}{\epsilon^3}+\frac{1}{\epsilon^2}\log C\right) (1+\epsilon)^{i+1}
\end{equation*}
be the \emph{conversion ratio} from set weight to up potential at level $i$.
We define the potential of $s$ as
\begin{equation*}
	\sum_{i = l(s)}^{\hat{l}(s)-1} \big(w(s,i)-w(s,i+1)\big)\cdot \alpha_i + \big( w(s,\hat{l}(s))-c_s \big)\cdot \alpha_{\hat{l}(s)}.
\end{equation*}

As it will be clear from our later analysis, the up potential of set $s$ is defined such that the potential released by \textsf{Promote}$(s)$ is at least the update time and the increase in total down potential due to the promotion.

The down potential is defined as follows.
Let
\begin{equation*}
	\beta_i = \frac{1}{f}\cdot \alpha_i = \left( \frac{1}{\epsilon^3}+\frac{1}{\epsilon^2 f}\log C\right) (1+\epsilon)^{i+1}
\end{equation*}
be the \emph{conversion ratio} from dead weight to down potential at level $i$.
The dead weight $\phi(s)$ of set $s$ has potential $\phi(s)\cdot \beta_{l(s)}$.
In other words, each unit of dead weight at level $i$ has potential $\beta_i$.

\subsubsection{Rebuilds}\label{ssec:potential-rebuild}

Recall that \textsf{Rebuild}$(k)$ takes time $O(\frac{f}{\epsilon^2}\cdot |E_{\leq k}|+\frac{1}{\epsilon}\log C)$, and cleans up the dead weights at level at most $k$.
Also recall that a set $s$ is charged one unit of lift-up cost each time \textsf{Lift-Up}$(s)$ is called, which increases $l(s)$ by at least one.
Moreover, if $l(s)<b(s)$, then \textsf{Lift-Up}$(s)$ increases $l(s)$ to $b(s)$ immediately.
Hence from the last time set $s$ is involved in some rebuild, at most $\max\{l(s)-b(s)+1,0\}$ units of lift-up costs are charged to $s$.

We show that (1) \textsf{Rebuild}$(k)$ does not increase the up potential of any set; (2) the down potential released by dead weight $\Phi_{\leq k}$ is sufficient to pay for the rebuild cost and the lift-up cost charged to sets in $T_{\leq k}$.

\begin{lemma}\label{lemma:potential-from-rebuild}
	\textsf{Rebuild}$(k)$ does not increase the up potential of any set.
\end{lemma}
\begin{proof}
	Observe that \textsf{Rebuild}$(k)$ does not change the up potential of sets at levels higher than $k$.
	Moreover, after the rebuild, each set $s$ at level at most $k-1$ has weight at most $c_s$, and has $0$ up potential.
	Consider any set $s$ that is at level $k$ after the rebuild:
	\begin{itemize}
		\item if set $s$ was at level $i \leq k-1$ before the rebuild, then we show that we have $w(s)< c_s$ after the rebuild.
		Recall that in the first step of rebuild (line 2 of Algorithm~\ref{alg:rebuild}), we put all elements and sets at level $k$.
		By the bounded weight invariant, we have $w(s,k)\leq w(s,i+1) < c_s$.
		Since $s$ stays at level $k$ afterwards, its weight does not increase in the later steps of Rebuild($k$).
		Hence after the rebuild we have $w(s) < c_s$, which implies that $s$ has $0$ potential.
		
		\item if set $s$ was at level $k$ before the rebuild, then the first step of rebuild does not change $w(s)$ or the level of any element contained in $s$, since they were at level $k$ or higher.
		Hence $s$ have the same potential after the rebuild.
	\end{itemize}
	
	In summary, \textsf{Rebuild}$(k)$ does not increase the up potential of any set.
\end{proof}

\begin{lemma}\label{lemma:potential-released-dead}
	The down potential of $\Phi_{\leq k}$ is at least
	\begin{equation*}
		\frac{f}{\epsilon^2}\cdot |E_{\leq k}|+\frac{1}{\epsilon}\log C+\sum_{s\in T_{\leq k}}\max\{l(s)-b(s)+1,0\}.
	\end{equation*}
\end{lemma}
\begin{proof}
	Recall that the conversion ratio of dead weight to down potential is smaller at lower levels.
	Hence given $\Phi_{\leq k} > \epsilon\cdot (c(T_{\leq k})+f\cdot w(E_{\leq k}))$ and the fact that $k$ is the smallest level for which the inequality holds, the down potential of $\Phi_{\leq k}$ is minimized when $\Phi_i = \epsilon\cdot (c(T_i)+f\cdot w(E_i))$ for all $i\in[k-1]$.
	Besides, we have $\Phi_k > \epsilon\cdot (c(T_k)+f\cdot w(E_k))$.
	
	Let $t\in[k]$ be an arbitrary level for which $E_t \neq \emptyset$.
	The down potential of $\Phi_{\leq k}$ is at least (recall that $\beta_i = \left(\frac{1}{\epsilon^3} + \frac{1}{\epsilon^2 f}\log C\right) (1+\epsilon)^{i+1}$)
	\begin{align*}
		&\epsilon\cdot \sum_{i=0}^k \sum_{s\in T_i} c_s\cdot \beta_i
		+ \epsilon f\cdot \sum_{i=0}^k\sum_{e\in E_{i}} w(e)\cdot \beta_i \\
		\geq &\epsilon\cdot \sum_{i=0}^k \sum_{s\in T_i} c_s\cdot \frac{(1+\epsilon)^{i+1}}{\epsilon^3} 
		+ \epsilon f\cdot \sum_{i=0}^k\sum_{e\in E_{i}} w(e)\cdot \frac{(1+\epsilon)^{i+1}}{\epsilon^3}
		+ \epsilon f\cdot (1+\epsilon)^{-t}\cdot \frac{(1+\epsilon)^{t+1}}{\epsilon^2 f}\log C\\
		\geq & \sum_{i=0}^k \sum_{s\in T_i} \frac{(1+\epsilon)^{i-b(s)}}{\epsilon^2} 
		+ f\cdot \sum_{i=0}^k\sum_{e\in E_{i}} (1+\epsilon)^{-i}\cdot \frac{(1+\epsilon)^{i+1}}{\epsilon^2}
		+ \frac{1+\epsilon}{\epsilon}\log C \\
		\geq & \sum_{i=0}^k \sum_{s\in T_i} \cdot \frac{\epsilon\cdot \max\{i-b(s)+1,0\} }{\epsilon^2}
		+ \frac{f}{\epsilon^2} |E_{\leq k}| + \frac{1}{\epsilon}\log C \\
		\geq & \sum_{s\in T_{\leq k}} \max\{l(s)-b(s)+1,0\} + \frac{f}{\epsilon^2} |E_{\leq k}|+ \frac{1}{\epsilon}\log C,
	\end{align*}
	where the third inequality comes from $(1+\epsilon)^{x}\geq 1+\epsilon\cdot x\geq \epsilon\cdot(x+1)$ for all $x\geq 0$.
\end{proof}

In summary, \textsf{Rebuild}$(k)$ does not increase the up potential, and releases an amount of down potential that is sufficient to pay for the rebuild cost and the lift-up cost charged to sets in $T_{\leq k}$.

\subsubsection{Promotions}\label{ssec:potential-promote}

Suppose we promote $s$ from level $k$ to $k+1$.
Recall that the update time is $O(f\cdot |E_k(s)|)$.
Also recall that the promotion of $s$ from level $k$ to $k+1$ increases the total dead weight by at most $f\cdot |E_k(s)|\cdot \epsilon(1+\epsilon)^{-(k+1)}$.
Since all these increased dead weights are at level at most $k$, the total down potential is increased by at most
\begin{equation*}
	f\cdot |E_k(s)|\cdot \epsilon(1+\epsilon)^{-(k+1)}\cdot \left(\frac{1}{\epsilon^3}+\frac{1}{\epsilon^2 f}\log C\right) (1+\epsilon)^{k+1} = \left( \frac{f}{\epsilon^2}+\frac{1}{\epsilon}\log C \right)\cdot |E_k(s)|.
\end{equation*}

In the following, we show that promoting $s$ from level $k$ to $k+1$ does not increase the up potential of any set intersecting $s$, and releases at least $\left( \frac{f}{\epsilon^2}+\frac{1}{\epsilon}\log C \right) \cdot |E_k(s)|$ up potential of $s$.

\begin{lemma}\label{lemma:potential-released-promotion}
	Promoting $s$ from level $k$ to $k+1$ releases $\left( \frac{f}{\epsilon^2}+\frac{1}{\epsilon}\log C \right)\cdot |E_k(s)|$ up potential of $s$.
\end{lemma}
\begin{proof}
	Note that promoting $s$ from level $k$ to $k+1$ does not change $\hat{l}(s)$ or $w(s,i)$ for all $i\geq k+1$.
	Hence by definition, the up potential released from set $s$ is
	\begin{align*}
		& \big( w(s,k) - w(s,k+1) \big)\cdot \left( \frac{f}{\epsilon^3}+\frac{1}{\epsilon^2}\log C \right) (1+\epsilon)^{k+1} \\
		= & \epsilon(1+\epsilon)^{-(k+1)}\cdot|E_k(s)|\cdot \left( \frac{f}{\epsilon^3}+\frac{1}{\epsilon^2}\log C \right) (1+\epsilon)^{k+1} = \left( \frac{f}{\epsilon^2}+\frac{1}{\epsilon}\log C \right)\cdot|E_k(s)|,
	\end{align*}
	which concludes the lemma.	
\end{proof}

It remains to show that the up potentials of other sets are not increased due to the promotion.
Obviously the promotion does not change the up potential of any set that is disjoint from $s$, or any set at level higher than $k$.

\begin{lemma}\label{lemma:potential-of-neighbor-promotion}
	Promoting $s$ from level $k$ to $k+1$ does not increase the potential of set at levels at most $k$ intersecting $s$.
\end{lemma}
\begin{proof}
	Fix any set $s'$ at level $l(s') = i \leq k$ that intersects with $s$.
	Recall that promoting $s$ from level $k$ to $k+1$ decreases $w(s')$ and may possibly increase $\phi(s')$.
	If $\phi(s') > 0$ after the promotion then $s'$ has $0$ potential, and the lemma trivially holds.
	
	Otherwise the promotion moves some elements in $E_k(s')$ to $E_{k+1}(s')$, and does not change the level of any other element contained in $s'$.
	Hence $w(s',j)$ remains unchanged for all $j\geq k+1$, and $w(s',j)$ decreases by the same amount, for all $j\leq k$.
	Consequently, $w(s',j)-w(s',j+1)$ remains unchanged for all $j\neq k$, and $w(s',k)-w(s',k+1)$ decreases.
	Moreover, $\hat{l}(s')$ does not increase.
	Hence the up potential of $s'$ does not increase due to the promotion.
\end{proof}

\subsubsection{Insertions and Deletions}\label{ssec:potential-insert-delete}

Recall that \textsf{Delete}($e$) handles each deletion in $O(f)$ time and \textsf{Insert}($e$) handles each insertion in $O(f\log f)$ time.
Additionally, $O(\frac{f}{\epsilon}\log C)$ time is charged to each insertion for maintaining the linked list of pointers to non-empty collections of tight sets.

Next, we bound the total up and down potential increased due to element insertions and deletion.

\paragraph{Deletion of $e$.}
In Algorithm~\ref{alg:deletion}, we increase $\phi(s)$ by at most $w(e)$ for each $s$ containing $e$.
Hence the total down potential increases by at most
\begin{equation*}
	\sum_{s: e\in s} w(e) \left(\frac{1}{\epsilon^3} + \frac{1}{\epsilon^2 f}\log C\right) (1+\epsilon)^{l(s)+1} \leq
	\sum_{s: e\in s} w(e) \left(\frac{1}{\epsilon^3} + \frac{1}{\epsilon^2 f}\log C\right) (1+\epsilon)^{l(e)+1} = O(\frac{f}{\epsilon^3} +\frac{1}{\epsilon^2}\log C).
\end{equation*}

\paragraph{Insertion of $e$.}
Recall that in Algorithm~\ref{alg:insertion}, we lift sets containing $e$ up until all sets are good, and then insert element $e$.
For all $s$ containing $e$, we increase $w(s)$ by $w(e)\leq (1+\epsilon)^{-l(s)}$.

In the following, we show that $e$ increases the up potential of $s$ by $O(\frac{f}{\epsilon^3}+\frac{1}{\epsilon^2}\log C)$.
Since each element is contained in at most $f$ sets, the $O(\frac{f^2}{\epsilon^3}+\frac{f}{\epsilon^2}\log C)$ upper bound on the increase of total up potential follows.
A similar upper bound on the increase of potential appeared in~\cite{soda/BhattacharyaK19}.

\begin{lemma}\label{lemma:potential-from-insertion}
	An insertion of $e$ increases the potential of each $s$ containing $e$ by $O(\frac{f}{\epsilon^3}+\frac{1}{\epsilon^2}\log C)$.
\end{lemma}
\begin{proof}
	Let $w(s)$ be the weight of $s$ and $t=l(s)$, right before $e$ is inserted.
	If $w(s)+w(e)\leq c_s$ then $s$ has $0$ up potential after the insertion, and the lemma trivially holds.
	Otherwise (since the algorithm lifts $s$ up until it becomes good) we know that before the insertion $E_{t}(s) \neq \emptyset$.
	Moreover, by the bounded weight invariant, we have $w(s,t+1) < c_s$ before the insertion.
	
	Recall that the conversion ratio of set weight to up potential at level $i$ is given by $\alpha_i = \left( \frac{f}{\epsilon^3}+\frac{1}{\epsilon^2}\log C \right) (1+\epsilon)^{i+1}$, and the up potential of $s$ is given by
	\begin{equation*}
		\sum_{i = l(s)}^{\hat{l}(s)-1} \big(w(s,i)-w(s,i+1)\big)\cdot \alpha_i + \big( w(s,\hat{l}(s))-c_s \big)\cdot \alpha_{\hat{l}(s)}.
	\end{equation*}
	
	The insertion of $e$ increases $w(s)$ by $w(e)$, and can possibly increase $\hat{l}(s)$.
	Let $t'$ be the value of value of $\hat{l}(s)$ \emph{after} the insertion.
	By definition of $w(s,i)$, the insertion does not decrease $w(s,i)-w(s,i+1)$ for any $i$.
	Since the insertion increases $w(s) - c_s$ by $w(e)$, the up potential of $s$ is increased by at most $w(e)\cdot \alpha_{t'}$.
	(Recall that the conversion ratio is larger at higher level.)
	
	Since $w(e) \leq (1+\epsilon)^{-t}$, the increment in up potential of $s$ is at most
	\begin{equation*}
		(1+\epsilon)^{-t} \cdot \left( \frac{f}{\epsilon^3}+\frac{1}{\epsilon^2}\log C \right)(1+\epsilon)^{t'+1} =  \left( \frac{f}{\epsilon^3}+\frac{1}{\epsilon^2}\log C \right) (1+\epsilon)^{t'-t+1}.
	\end{equation*}
	
	It remains to show that $(1+\epsilon)^{t'-t} = O(1)$.
	
	Since $E_t(s)\neq \emptyset$ and $w(s,t+1)<c_s$ before the insertion, after the insertion, the weight of $s$ at level $t'$ is at most
	\begin{equation*}
		c_s- |E_t(s)|\big( (1+\epsilon)^{-(t+1)} - (1+\epsilon)^{-t'} \big) + (1+\epsilon)^{-t'} \leq c_s-(1+\epsilon)^{-(t+1)} + 2(1+\epsilon)^{-t'}.
	\end{equation*}
	
	On the other hand, the weight of $s$ at level $t'$ after the insertion is at least $c_s$. Hence we have $c_s-(1+\epsilon)^{-(t+1)} + 2(1+\epsilon)^{-t'} \geq c_s$, which is equivalent to $(1+\epsilon)^{t'-t} \leq 2(1+\epsilon)$, as required.
\end{proof}

\subsubsection{Total lift-up cost} \label{ssec:lift-up-cost}

Since \textsf{Rebuild}$(k)$ removes all lift-up costs charged to the sets in $T_{\leq k}$, each set $s$ at level $i$ is charged at most $\max\{i-b(s)+1,0\}$ units of lift-up cost.
Assume that the final instance contains no element.
Then all sets are slack and are at level $0$ in the final instance.
Hence at the end of update sequence there is no lift-up cost charged to any set.

Next we show that the assumption is without loss of generality.
Given any sequence of $t$ updates, we append to the end of the sequence a deletion for each element that exists after the $t$ updates.
Hence there are in total at most $\Gamma + 2t$ updates, where $\Gamma$ is the number of elements preprocessed.
The total update time of the new sequence of updates, given that the final graph is empty, is at most $O((\Gamma+t)\cdot \lambda)$, where $\lambda = O(\frac{f^2}{\epsilon^3}+\frac{f}{\epsilon^2}\log C)$ is the amortized update time.
Thus the total update time to process the $t$ updates is $O((\Gamma+t)\cdot \lambda)$.

%% file: Algorithm_Worst_Case.tex
\section{Our Algorithm for Worst Case Update Time} \label{sec:worst-case}


In this section we present the algorithm with a worst case $O(f\log^2 (Cn)/\epsilon^3)$ update time.
Specifically, we prove the following theorem.

\begin{theorem} \label{th:worst_case_time}
	There are deterministic dynamic algorithms for the minimum set cover problem with $(1+\epsilon)f$-approximation ratio and a worst-case update time of $O(f\log^2 (Cn)/\epsilon^3)$.
\end{theorem}

Our algorithm maintains the $(1+\epsilon)f$-approximate set cover \emph{value} after every update and can return it in constant time.
The algorithm can also output the solution in time linear to the solution size whenever the solution is asked for (similar to the dynamic matching algorithm in \cite{BernsteinFH-soda19}).

\paragraph{Notations.}
We use $[i,j]$ to denote $\{i,i+1,\ldots,j\}$ for any integers $j\geq i\geq 0$.
Our algorithm guarantees that every element $e$ is at level $l(e) = \max_{s\in \S:e\in s}\{l(s)\}$.
Similar to~\cite{focs/BhattacharyaHN19}, we classify the elements into three types, namely active, passive and dead.
Each active element $e$ has weight exactly $(1+\epsilon)^{-l(e)}$; each passive element $e$ has weight at most $(1+\epsilon)^{-l(e)}$.
An element becomes dead at the moment when it gets deleted.
The weight of a dead element is decided by its weight before the deletion.
As it will be clear from the description of the algorithm, each newly inserted element will be passive, which may turn into active only if it is involved in a rebuild of levels.

\paragraph{Worst Case Update Time.}
In standard worst case update time algorithms, we assume that a new update arrives when the previous one is properly handled.
The worst case update time is then measured by the maximum time between two consecutive updates.
In this paper we assume that updates and queries arrive in a fixed rate.
That is, the time between any two consecutive updates is fixed, say, by some parameter $\lambda$.
We show that our algorithm updates the data structure or answers the query before the next update/query arrives.
Our algorithm maintains $L$ parallel data structures, and does $O(L\cdot \lambda)$ units of work between two consecutive updates/queries.
Consequently we have a worst case update time of $O(L\cdot \lambda)$.
Our goal is the design algorithms that work with updates with a high arrival rate, i.e., small $\lambda$.

\medskip

We show that the main technical challenge is to do the rebuild efficiently while committing newly arrived updates to the data structure (hierarchy) that is being rebuild.
Specifically, to guarantee that the number of passive elements and dead elements are bounded, we need to assume that a bounded number of updates arrive when the rebuild is being executed in the background.

We first present in Section~\ref{ssec:rebuild-and-update} a rebuild algorithm that works with $\lambda = c\cdot \frac{1}{\epsilon}\cdot f\cdot L^2$, for some sufficiently large constant $c$, which implies an $O(\frac{f}{\epsilon^4}\log^3 (Cm))$ worst case update time.
We improve the algorithm and obtain an $O(\frac{f}{\epsilon^3}\log^2 (Cm))$ worst case update time in Section~\ref{sec:improved-algorithm}.

\paragraph{Data Structure.}
As before (in the amortized algorithm), for every set we maintain a partition of the elements it contains according to the levels of elements.
Additionally, every element maintains a pointer to every set containing it.

\subsection{Algorithm Framework} \label{ssec:worst-case-algorithm}

We first describe a data structure that consists of $L$ independent hierarchies as follows.

For each $k\in [1,L]$, we maintain a hierarchy $H_k$ that has levels $[0,k]$, which we refer to as the \emph{local view} of level $k$ (local view($k$)).
We use $A_k(i)$ to denote the collection of active elements at level $i$, from the local view$(k)$, and $A_k(\leq i)$ to denote the collection of active elements on or below level $i$, from local view($k$).
We define the sets $P_k(i)$, $P_k(\leq i)$ for passive elements, and $D_k(i)$, $D_k(\leq i)$ for dead elements similarly.
Let $E_k(i) = A_k(i)\cup P_k(i)\cup D_k(i)$, and $E_k(\leq i) = \bigcup_{j\leq i}E_k(j)$.
Let $S_k(i)$ (resp. $S_k(\leq i)$) denote the collection of sets at level $i$ (resp. at most $i$), from local view($k$).
We also maintain a counter for the size of each collection (of elements or sets) we maintain.
Unless otherwise specified, we use a doubly linked list for each of these collections.
Thus given the pointer to an element or a set, the insertion and deletion of the element or set can be done in $O(1)$ time.
Moreover, we can also merge two collections in $O(1)$ time.

Let $l_k(e)$ be the level of element $e$, and $l_k(s)$ be the level of set $s$, in local view($k$).

Additionally, every set in local view($k$) has an \emph{extra weight}, which represents the total weight the set receives from levels $[k+1,L]$ (which are not included in $H_k$).
Let $\delta_k(s)$ denote the extra weight of $s$ in local view($k$).
The weight of set $s$ in local view($k$) is then defined as
\begin{equation*}
w_k(s) = \delta_k(s) + \sum_{e\in s\cap E_k(\leq k)}w_k(e),
\end{equation*}
where $w_k(e)$ denotes the weight of element $e$ in local view($k$).

We call a set $s$ \emph{tight} (with respect to local view($k$)) if $w_k(s) > \frac{c_s}{1+\epsilon}$.
We maintain the invariant that each set at level $[1,k]$ of $H_k$ is tight.

\begin{invariant}[Local Tightness]\label{inv:local-tightness}
	For all $i\in[1,k]$, each $s\in S_k(i)$ has weight $w_k(s) > \frac{c_s}{1+\epsilon}$.
\end{invariant}

Every update (that inserts or deletes an element) becomes $L$ updates, one for each hierarchy.

\paragraph{Handling An Update in Local View($k$).}
Suppose an element $e$ is inserted or deleted.
We first check if all sets containing $e$ exist in $S_k(\leq k)$.
If no, then we ignore the update.
Otherwise
\begin{compactitem}
	\item if the element is deleted, we convert $e$ into a dead element in $H_k$;
	\item if the element is inserted, we insert $e$ as a passive element at level $l_k(e) = \max_{s: e\in s}\{l_k(s)\}$.
	If $l_k(e) > 0$ then we set $w_k(e) = 0$; otherwise $l(e) = 0$, and we set $w_k(e)$ to be $\min_{s: e\in s} \{ c_s - w_k(s) \}$, which is at most $1$.
	Note that by Invariant~\ref{inv:local-tightness} at least one set containing $e$ is tight.
\end{compactitem}
We also update the relevant counters for local view($k$) accordingly.

\medskip

We maintain the following local invariant regarding the total number of passive elements and dead elements in each local view.

\begin{invariant}[Local Element]\label{inv:local-element}
	For every $k\in[L]$, we have $|P_k(\leq k)\cup D_k(\leq k)| < 2\epsilon\cdot |A_k(\leq k)|$.
\end{invariant}

The local element invariant (Invariant~\ref{inv:local-element}) guarantees that the number of dead elements is bounded, which is crucial to bound the approximation ratio of the set cover.

Additionally, we define the \emph{rebuild-triggering event} for local view($k$) as $|P_k(\leq k)\cup D_k(\leq k)| \geq \epsilon\cdot |A_k(\leq k)|$.
Whenever the event is triggered, Scheduler($k$) gets into action and the hierarchy $H_k$ will be completely rebuilt.
We show (in Section~\ref{ssec:maintenance-of-invariant}) that when Scheduler($k$) is running (in the background), Invariant~\ref{inv:local-element} remains valid.
The triggering event defines a stronger version of Invariant~\ref{inv:local-element}.
It is crucial that we rebuild the data structure earlier (before Invariant~\ref{inv:local-element} is violated), since it might take a long time before the rebuild finishes, and we need to guarantee that the invariant is maintained before the rebuild finishes.

\paragraph{Scheduler($k$).}
We first describe the scheduler for level $k$, which rebuilds $H_k$ in the background, in a high level.
The rebuild (which we refer to as Rebuild($k$)) takes the elements in $E_k(\leq k)$ as input and outputs a new hierarchy $H^*_k$, which has levels $[0,k+1]$.
We show that we are able to commit the newly arrived updates to the hierarchy $H^*_k$ while it is being constructed.
Hence when Rebuild($k$) finishes, the output hierarchy $H^*_k$ is up-to-date.
Scheduler($k$) then replaces $H_i$ for all $i \leq k$ based on $H^*_k$.
In addition, Scheduler($k$) updates local view($k+1$) based on $H^*_k$.

\subsubsection{Consistency Invariants}

Initially, in the preprocessing we construct a hierarchy for the input graph, based on which we construct $L$ consistent hierarchies such that $E_k(i)$ is the same for all $k\geq i$.
In general (when there are updates), the local views can be inconsistent.
For example, when levels on or below $k$ are rebuilt in Scheduler($k$), local view($k-1$) will be completely replaced, while local view($k+2$) is oblivious to the rebuild. 
In this case $A_{k+2}(i)$ might be different from $A_{k-1}(i)$ for some $i\leq k$.

Nevertheless, our algorithm maintains the following invariants that guarantee some consistency between the parallel hierarchies.
We show that the invariants are maintained in Section~\ref{ssec:maintenance-of-invariant}.

\begin{invariant}[Lower View Consistency]\label{inv:lower-level-consistency}
	We have $S_{i+1}(\leq i) =  S_{i}(\leq i)$ for all $i\in[1,L-1]$.
\end{invariant}

\begin{invariant}[Upper View Consistency] \label{inv:upper-level-consistency}
	Suppose an element $e$ is passive (resp. dead) w.r.t. local view($i$). Then $e$ is also passive (resp. dead) w.r.t. local view($j$) for all $j > i$.
\end{invariant}

By invariant~\ref{inv:lower-level-consistency}, if a set exists in local view($k$) then it also exists in local view($k+1$).
Hence every update not ignored by local view($k$) will not be ignored by local view($i$) if $i > k$.
Consequently, every passive/dead element in $H_k$ also appears in $H_i$, as stated in Invariant~\ref{inv:upper-level-consistency}. 

\medskip

Next we present the details of Scheduler($k$).

\subsubsection{Rebuild and Synchronization}\label{ssec:synchronization}

Recall that Rebuild($k$) takes the elements in $E_k(\leq k)$ as input and outputs a new hierarchy $H^*_k$.
We will present two algorithms for the construction in Section~\ref{ssec:rebuild-and-update} and Section~\ref{sec:improved-algorithm}, respectively.

\begin{property}\label{property}
	We show that Rebuild($k$) constructs a hierarchy $H^*_k$ with the following properties.
	\begin{itemize}
		\item[(a)] The hierarchy $H^*_k$ has levels $[0,k+1]$, and contains every set in $S_k(\leq k)$.
		If an element $e$ in $A_k(\leq k) \cup P_k(\leq k)$ is not deleted by an update that arrives when Scheduler($k$) is running, then it will be contained in $H^*_k$.
		Moreover, if $e\in A_k(\leq k)$ then it is active in $H^*_k$; if $e\in P_k(\leq k)$ then it either becomes active, or will be passive at level $k+1$ of $H^*_k$. 
		\item[(b)] 
		Each set $s$ in $H^*_k$ has extra weight $\delta^*_k(s) = \delta_k(s)$.
		Each set at level $[1,k+1]$ of $H^*_k$ is tight.
		\item[(c)] When Rebuild($k$) finishes, the hierarchy $H^*_k$ is up-to-date, i.e., all updates that have arrived are committed to the hierarchy.
		\item[(d)] For each $i\leq k$, $|P^*_k(\leq i)\cup D^*_k(\leq i)| < \epsilon\cdot |A^*_k(\leq i)|$, where $P^*_k(\leq i)$, $D^*_k(\leq i)$, $A^*_k(\leq i)$ denote the collection of passive, dead, and active elements, respectively, at levels $[0,i]$ of $H^*_k$.
		\item[(e)] The construction of $H^*_k$ takes time at most $\frac{\epsilon}{2L}\cdot |E_k(\leq k)|\cdot \lambda$.
	\end{itemize}
\end{property}

\medskip

We run $k+1$ copies of Rebuild($k$) in Scheduler($k$), and construct $k+1$ identical hierarchies $H^*_k$.
We name the hierarchies as $H^*_{k\rightarrow 1}, H^*_{k\rightarrow 2},\ldots, H^*_{k\rightarrow k+1}$.
As the name indicates, the hierarchy $H^*_{k\rightarrow i}$ is constructed for the purpose of replacing/updating local view($i$) when the rebuild finishes.

For each $i \leq k$, we take the hierarchy $H^*_{k\rightarrow i}$, remove levels $[i+1,k+1]$ and update the extra weights as follows.
The extra weight $\delta^*_{k\rightarrow i}(s)$ of each set $s$ in $H^*_{k\rightarrow i}$ is defined to be its extra weight $\delta_k(s)$ in local view($k$) plus the total weight $s$ receives from (elements in) levels $[i+1,k+1]$ in $H^*_{k\rightarrow i}$.
Then we replace the old hierarchy $H_i$ with the new hierarchy $H^*_{k\rightarrow i}$, and free the space taken by the old hierarchy.
If Scheduler($i$) is currently running in the background, we terminate it immediately (in $O(1)$ time).

We also update local view($k+1$) as follows.
Take the hierarchy $H^*_{k\rightarrow k+1}$, and include all sets and elements at level $k+1$ of $H_{k+1}$ into level $k+1$ of $H^*_{k\rightarrow k+1}$.
By storing the elements in $A_{k+1}(k+1)$ (resp. $P_{k+1}(k+1)$, $D_{k+1}(k+1)$ and $S_{k+1}(k+1)$) as a linked list, the merging of data structure can be done in $O(1)$ time.
The extra weight $\delta^*_{k\rightarrow k+1}(s)$ of every set $s$ in $H^*_{k\rightarrow k+1}$ is set to be $\delta_{k+1}(s)$.

Let the resulting hierarchy be the new local view of $k+1$.
Consequently, in local view($k+1$), the extra weights of sets do not change after the update.

If Scheduler($k+1$) is currently running in the background, we do not terminate it.
Additionally, if Scheduler($k+1$) is currently reading/duplicating elements\footnote{As we will show later, the first step of Rebuild($k+1$) is making a copy for every set and element in $H_{k+1}$.} in the old hierarchy $H_{k+1}$, then we do not free the space taken by the old hierarchy $H_{k+1}$.
Instead, we keep committing newly arrived updates to the data structure until the duplicating finishes.
As soon as Scheduler($k+1$) finishes reading the old hierarchy $H_{k+1}$, we free the space taken by the old hierarchy.

\paragraph{Data Structure for the Hierarchies.}
We assume that each local view($i$), where $i\in[1,L]$, has a pointer that points to the hierarchy $H_i$.
The hierarchy $H_i$ maintains a pointer for each collection of the elements and sets at each level.
Hence when we need to replace the $H_i$ in local view($i$), it suffices to change the pointer from the old hierarchy to the new hierarchy in $O(1)$ time.

\subsubsection{Maintenance of Invariants}\label{ssec:maintenance-of-invariant}

In this section we show that Invariants~\ref{inv:local-tightness},~\ref{inv:local-element},~\ref{inv:lower-level-consistency} and~\ref{inv:upper-level-consistency} are maintained.

We first show that Invariants~\ref{inv:lower-level-consistency} and~\ref{inv:upper-level-consistency} are maintained.

\begin{lemma}
	The lower view consistency invariant (Invariant~\ref{inv:lower-level-consistency}) is maintained.
\end{lemma}
\begin{proof}
	Initially when all local views are consistent, the invariant trivially holds.
	Note that sets change their levels only when some hierarchy is rebuilt.
	In other words, insertions and deletions of elements do not change levels of sets.
	Consider the point in time when some Scheduler($k$) finishes, where $k\in [1,L]$.
	
	Recall that Scheduler($k$) replaces local view($i$) for every $i\leq k$, and update local view($k+1$) based on the identical hierarchies output by Rebuild($k$).
	Hence $S_i(\leq i)$ and $S_i(\leq i-1)$ does not change for every $i\in [k+2,L]$.
	Additionally, $S_{k+1}(\leq k+1)$ does not change.
	Thus Invariant~\ref{inv:lower-level-consistency} holds for every $i\in [k+1,L-1]$.
	Since after Scheduler($k$) finishes, $S_i(j) = S_{k+1}(j)$ for every $i\in [1,k]$ and $j\in [0,k]$, Invariant~\ref{inv:lower-level-consistency} is also maintained for every $i\in [1,k]$.
\end{proof}

\begin{lemma}
	The upper view consistency invariant (Invariant~\ref{inv:upper-level-consistency}) is maintained.
\end{lemma}
\begin{proof}
	Suppose $e$ is passive w.r.t. local view($i$) but not passive w.r.t. local view($j$) for some $j > i$.
	Note that when $e$ is inserted to $H_i$, it is also inserted to $H_k$ for every $k > i$ (by Invariant~\ref{inv:lower-level-consistency}), as passive elements.
	The element $e$ stays passive in $H_j$ until either (1) Scheduler($k$) converts $e$ into active, or moves it to level $k+1$, for some $k\geq j$; or (2) Scheduler($j-1$) converts $e$ into active.
	In both cases local view($i$) will be replaced, and $e$ will no longer be passive in $H_i$, which is a contradiction.
	The proof for dead elements is similar.
\end{proof}

Next we show that the local tightness invariant (Invariant~\ref{inv:local-tightness}) is maintained.

\begin{lemma}\label{lemma:local-tightness-maintained}
	The local tightness invariant (Invariant~\ref{inv:local-tightness}) is maintained.
\end{lemma}
\begin{proof}
	We show that after Scheduler($k$) replaces local view($i$) for all $i\leq k$ and updates local view($k+1$), all sets at levels $[1,i]$ of local view($i$) is tight, for all $i\leq k+1$.
	
	For local view($i$), where $i\leq k$, the invariant follows straightforwardly from the fact that every set at level $[1,k+1]$ of $H^*_k$ is tight, and the way we update the extra weight (in Section~\ref{ssec:synchronization}).
	Next we show that every set at level $[1,k+1]$ from the new local view($k+1$) is tight.
	
	Consider the point in time right before we update local view($k+1$).
	Fix any set $s$ at level $[1,k+1]$ of $H_{k+1}$.
	If $s\in S_{k+1}(k+1)$, then the weight $w_{k+1}(s)$ does not change when Scheduler($k$) updates local view($k+1$).
	If $s\in S_{k+1}(\leq k)$ then $s$ also appears in $S_k(\leq k)$ (by Invariant~\ref{inv:lower-level-consistency}) and $H^*_{k\rightarrow k+1}$.
	When Scheduler($k$) updates local view($k+1$), $w^*_{k\rightarrow k+1}(s)$ is updated as follows.
	\begin{itemize}
		\item It increases by
		$\sum_{e\in s\cap E_{k+1}(k+1)} w_{k+1}(e)$ when we include elements in $E_{k+1}(k+1)$ to $H^*_{k\rightarrow k+1}$.
		\item Then it decreases by $\delta_k(s) - \delta_{k+1}(s)$ when we update the extra weight of $s$ in $H^*_{k\rightarrow k+1}$.
	\end{itemize}
	
	We show that $\delta_k(s) - \delta_{k+1}(s) = \sum_{e\in s\cap E_{k+1}(k+1)} w_{k+1}(e)$, which implies that modifying $H^*_{k\rightarrow k+1}$ does not change $w^*_{k\rightarrow k+1}(s)$.
	Hence $s$ is tight in the new local view($k+1$) if $l^*_{k\rightarrow k+1}(s)\geq 1$.
	
	Consider the last point in time when $\delta_k(s)$ (the extra weight of $s$ in local view($k$)) is updated\footnote{Note that if $\delta_{k+1}(s)$ is updated then $\delta_k(s)$ will also be updated.}, which must be the time when Scheduler($j$) finishes, for some $j\geq k$.
	In this case the extra weights are updated (in Section~\ref{ssec:synchronization}) such that $\delta_k(s) - \delta_{k+1}(s) = \sum_{e\in s\cap E_{k+1}(k+1)} w_{k+1}(e)$.
	Since then, the LHS of the equality does not change.
	Furthermore, the RHS of the equality does not change either: if an element is inserted to $P_{k+1}(k+1)$, then it has weight $0$; if an element is deleted and moved to $D_{k+1}(k+1)$, then its contribution to $w_{k+1}(s)$ does not change.
	Hence we have $\delta_k(s) - \delta_{k+1}(s) = \sum_{e\in s\cap E_{k+1}(k+1)} w_{k+1}(e)$ right before we update local view($k+1$), and $s$ is tight if $l^*_{k\rightarrow k+1}(s)\geq 1$ after local view($k+1$) is updated.
\end{proof}

Finally, we show that the local element invariant (Invariant~\ref{inv:local-element}) is maintained.

\begin{lemma}\label{lemma:local-element-maintained}
	The local element invariant (Invariant~\ref{inv:local-element}) is maintained.
\end{lemma}
\begin{proof}
	By Property~\ref{property}(d), when Scheduler($k$) replaces the local view($i$) for every $i\leq k$, Invariant~\ref{inv:local-element} is maintained.
	Indeed, right after local view($i$) is replaced, we can guarantee that the rebuild-triggering event is not triggered, for all $i\leq k$.
	
	Recall that we also update local view($k+1$).
	However, by Invariant~\ref{inv:lower-level-consistency} and~\ref{inv:upper-level-consistency}, before local view($k+1$) is updated, every passive/dead element contained in $H^*_{k\rightarrow k+1}$ is also contained in $P_{k+1}(\leq k)\cup D_{k+1}(\leq k)$, since every update committed to local view($k$) and Scheduler($k$) is also committed to local view($k+1$).
	Hence updating local view($k+1$) does not increase $|P_{k+1}(\leq k+1)\cup D_{k+1}(\leq k+1)|$.
	Moreover, updating local view($k+1$) does not decrease $|A_{k+1}(\leq k+1)|$, because every active element in $A_{k+1}(\leq k)$ is also contained in $H^*_{k\rightarrow k+1}$.
	Hence Invariant~\ref{inv:local-element} will not be violated when Scheduler($k$) updates local view($k+1$).
	
	Finally, we show that Invariant~\ref{inv:local-element} is maintained in local view($k$) before Scheduler($k$) finishes.
	By Property~\ref{property}(d), Scheduler($k$) takes at most $\frac{\epsilon}{2}\cdot |E_k(\leq k)|\cdot \lambda$ time to construct the identical hierarchies, during which $x \leq \frac{\epsilon}{2}\cdot |E_k(\leq k)|$ updates arrive.
	
	Let $a$ be the size of $P_k(\leq k)\cup D_k(\leq k)$, and $b$ be the size of $A_k(\leq k)$ when Scheduler($k$) starts.
	By the definition of the rebuild triggering event, we have $a = \epsilon\cdot b$.
	Hence we have $x \leq \frac{\epsilon}{2}\cdot (a + b) =  \frac{\epsilon+\epsilon^2}{2}\cdot b$.
	Since each update decreases $|A_k(\leq k)|$ by at most one, before Scheduler($k$) finishes we have $|A_k(\leq k)| \geq b-x \geq (1-\frac{\epsilon+\epsilon^2}{2})\cdot b$.
	Since each update increases $|P_k(\leq k)\cup D_k(\leq k)|$ by at most one, before Scheduler($k$) finishes we have
	\begin{equation*}
	|P_k(\leq k)\cup D_k(\leq k)| \leq a+x \leq (\epsilon + \frac{\epsilon+\epsilon^2}{2})\cdot b
	\leq \frac{\epsilon + \frac{\epsilon+\epsilon^2}{2}}{1-\frac{\epsilon+\epsilon^2}{2}}\cdot |A_k(\leq k)|
	< 2\epsilon \cdot |A_k(\leq k)|.
	\end{equation*}	
	
	Note that while Scheduler($k$) is running in the background, we might update local view($k$) (when Scheduler($k-1$) finishes).
	However, as argued above, updating local view($k$) does not increase $|P_{k}(\leq k)\cup D_{k}(\leq k)|$ nor decrease $|A_{k}(\leq k)|$.
	Hence the above upper bound still holds.
\end{proof}

\subsubsection{Answering a Query}\label{ssec:answer-query}

We answer the query on the size of set cover as follows.
All slack sets in $S_1(0)$, i.e., sets $s$ at level $0$ of $H_1$ with $w_1(s) \leq \frac{c_s}{1+\epsilon}$, are not in the set cover. All other sets are in the set cover.
The correctness follows from the approximation ratio analysis in the next section.

\subsection{Consistent Hierarchy and Approximation Ratio}\label{sec:approximation}

We first show that while the local views are inconsistent, by Invariant~\ref{inv:lower-level-consistency}, there is a natural way of partitioning the sets into $L+1$ levels, which induces a consistent hierarchy.

\begin{lemma}\label{lemma:sets-partition}
	The collections $\{S_k(k)\}_{k\in [L]}$ form a partition of all sets into $L+1$ levels.
\end{lemma}
\begin{proof}
	We prove by induction on $k$ that $\{S_i(i)\}_{i\in [k]}$ is a partition of the sets $S_k(\leq k)$.
	
	The statement is trivially true for $k=1$.
	For the case $k+1$, we have $S_{k+1}(k+1) \cap S_{k+1}(\leq k) = \emptyset$.
	By Invariant~\ref{inv:lower-level-consistency}, we have $S_k(\leq k) = S_{k+1}(\leq k)$.
	By induction hypothesis $\{S_i(i)\}_{i\in [k]}$ is a partition of the sets in $S_{k+1}(\leq k)$.
	Hence $\{S_i(i)\}_{i\in [k+1]}$ is a partition of the sets in $S_{k+1}(\leq k+1)$.
\end{proof}

We show that at any point in time, there exists a hierarchy (with levels $[0,L]$) containing all elements and sets, which we refer to as the \emph{consistent hierarchy}, such that every set at level $[1,L]$ is tight, and the collection of sets at level $0$ is $S_1(0)$.
Moreover, set $s\in S(0) = S_1(0)$ is tight in the consistent hierarchy if and only if it is tight w.r.t. local view($1$).
Hence the collection of tight sets is a feasible set cover (which implies the correctness for answering a query in Section~\ref{ssec:answer-query})

We construct the consistent hierarchy as follows.
For all $i\in[1,L]$, let the sets and elements at level $i$ of the hierarchy be defined as follows:
\begin{equation*}
S(i) = S_i(i),\quad A(i) = A_i(i),\quad P(i) = P_i(i)\quad \text{and} \quad D(i) = D_i(i).
\end{equation*}

Let $S(0) = S_1(0)$, $A(0) = A_1(0)$, $P(0) = P_1(0)$ and $D(0) = D_1(0)$.
Let $E(i) = A(i)\cup P(i)\cup D(i)$.

We first show that the hierarchy is well defined, i.e., every set and element (that is not deleted) appears exactly once in the hierarchy.
By Lemma~\ref{lemma:sets-partition}, every set appears exactly once in the consistent hierarchy.
Let $l(s)$ (resp. $l(e)$) be the level of set $s$ (resp. element $e$) in the consistent hierarchy.
Note that for each $l(s) = i$, we have $l_i(s) = i$.
Consider any element $e$ that exists, e.g., inserted and not deleted.
Let $\max_{s: e\in s} \{l(s)\} = k$.
Hence $e$ exists in local view($k$) since all sets containing in $e$ appear in local view($k$).
Moreover, we have $e\in E_k(k) = E(k)$, which implies that $l(e) = \max_{s: e\in s} \{l(s)\}$. 
Moreover, (1) $e\notin E_i(i)$ for all $i>k$ since each element in $E_i(i)$ contains a set in $S_i(i)$, while the maximum level of sets containing $e$ is $k$; (2) since for all $i<k$, there exists some set containing $e$ that is not in $S_i(\leq i)$, $e$ does not appear in any of $\{E_i(i)\}_{i < k}$.
Hence $e$ appears exactly once (at level $k$) in the consistent hierarchy.

Since there is no extra weight attached to each set, the real weight $w(s)$ of a set $s$ in the consistent hierarchy is defined as the total weight of the elements it contains.
In the following, we call a set tight if it is tight w.r.t. the consistent hierarchy.

\begin{lemma}[Global Tightness]\label{lemma:consistent-hierarchy-tightness}
	Every set at level $[1,L]$ in the consistent hierarchy is tight.
	Set $s$ at level $0$ in the consistent hierarchy is tight if and only if it is tight w.r.t. local view($1$).
\end{lemma}
\begin{proof}
	Consider any $s \in S(i)$, where $i\geq 1$, in the consistent hierarchy.
	By the local tightness invariant (Invariant~\ref{inv:local-tightness}), from local view($i$), the weight $s$ receives from $E_i(i)$ plus $\delta_i(s)$ is more than $\frac{c_s}{1+\epsilon}$.
	Since we define $E(i) = E_i(i)$, it suffices to show that the extra weight $\delta_i(s)$ equals to the weight $s$ receives from levels $[i+1,L]$ in the consistent hierarchy.
	
	To prove that, we show that for every $k\in [i,L-1]$ (note that $\delta_L(s) = 0$),
	\begin{equation*}
	\delta_k(s) - \delta_{k+1}(s) = \sum_{e\in s\cap E_{k+1}(k+1)}w(e).
	\end{equation*}
	
	Consider the last point in time when $\delta_k(s)$ is changed, which must be the time when Scheduler($j$) finishes, for some $j\geq k$.
	When $\delta_k(s)$ changes, our algorithm guarantees that $\delta_k(s) - \delta_{k+1}(s) = \sum_{e\in s\cap E_{k+1}(k+1)}w_{k+1}(e)$.
	Since then no Rebuild($j$) finishes, for all $j\geq k$, and both the LHS and RHS of the equality does not change, which concludes the proof.
	
	Using the same argument, for every set $s\in S(0) = S_1(0)$, the weight $s$ receives from levels $[i+1,L]$ in the consistent hierarchy is $\delta_1(s)$.
	Since $E(0) = E_1(0)$, i.e., every element in level $0$ of local view($1$) is preserved in the consistent hierarchy, set $s\in S(0)$ is tight in the consistent hierarchy if and only if it is tight w.r.t. local view($1$). 
\end{proof}

Lemma~\ref{lemma:consistent-hierarchy-tightness} implies that every element $e$ in the consistent hierarchy contains at least one tight set: if $l(e) \geq 1$ then it contains a set at level $\geq 1$, which is tight; if $l(e) = 0$ then it contains a set $s\in S_1(0)$ with weight $c_s$, which is tight.  
Hence the collection of tight sets of the consistent hierarchy is a feasible set cover.
Next we analyze the approximation ratio.

\medskip

We show that in the consistent hierarchy, the number of dead elements is bounded.
We prove the following stronger statement, which upper bounds the total number of passive/dead elements.

\begin{lemma}\label{lemma:consistent-hierarchy-element-fraction}
	In the consistent hierarchy, for all $i\in[L]$ we have
	\begin{equation*}
	|P(\leq i)\cup D(\leq i)| \leq 2\epsilon\cdot |A(\leq i)|.
	\end{equation*}
\end{lemma}
\begin{proof}
	By the local element invariant (Invariant~\ref{inv:local-element}), we have $|P_i(\leq i)\cup D_i(\leq i)| \leq 2\epsilon\cdot |A_i(\leq i)|$.
	In the following we show that (1) $P(\leq i)\cup D(\leq i)\subseteq P_i(\leq i)\cup D_i(\leq i)$; (2) $A_i(\leq i)\subseteq A(\leq i)$.
		
	We first prove (1).
	Note that every $e\in P(\leq i)\cup D(\leq i)$ must appear as a passive or dead element in $P_j(j)\cup D_j(j)$ for some $j\leq i$.
	By Invariant~\ref{inv:upper-level-consistency}, $e$ is also contained in $P_i(\leq i)\cup D_i(\leq i)$.
	We remark that $P(\leq i)\cup D(\leq i)$ can be a proper subset of $P_i(\leq i)\cup D_i(\leq i)$.
	For example, when some dead element in $D_i(\leq i)$ is cleaned up by Scheduler($i-2$), it is still contained in $P_i(\leq i)\cup D_i(\leq i)$, but not in $P(\leq i)\cup D(\leq i)$.
	
	Statement (2) follows immediately from the fact that every active element appears exactly once (as an active element) in the consistent hierarchy.
	Observe that if $e\in A_i(\leq i)$, then all sets containing $e$ appear in $S(\leq i)$ in the consistent hierarchy, which implies $e\in A(\leq i)$.
	We remark that $A_i(\leq i)$ can be a proper subset of $A(\leq i)$.
	For example, when some passive element in $P_{i-2}(\leq i-2)$ is converted to active by Scheduler($i-2$), it will be contained in $A(\leq i)$, but not in $A_i(\leq i)$.
\end{proof}

Following the proof of~\cite[Lemma 4.8]{focs/BhattacharyaHN19}, Lemma~\ref{lemma:consistent-hierarchy-element-fraction} implies that the total weight of dead elements is at most $O(\epsilon)$ times the total weight of active elements.
This implies the approximation ratio.

%% file: Rebuild_Worst_Case.tex
\subsection{A Simple Rebuild Algorithm}\label{ssec:rebuild-and-update}

In this section we present a simple rebuild algorithm for constructing the hierarchy $H^*_k$ in Scheduler($k$) that works with $\lambda = c\cdot \frac{1}{\epsilon}\cdot f\cdot L^2$, for some sufficiently large constant $c$.

Recall the static rebuild algorithm from~\cite{focs/BhattacharyaHN19} as follows.
Roughly speaking, the algorithm puts all elements to level $k+1$ (line 1-7 in Algorithm~\ref{alg:rebuild_wc}), and then gradually moves the elements to lower levels until every set becomes tight or reaches level $0$  (line 8-17 in Algorithm~\ref{alg:rebuild_wc}).
Note that since the rebuild is done in the background, in the first phase (putting elements to level $k+1$), we need to make a copy for every active/passive element in $E_k(\leq k)$.
We also make a copy for each set in $S_k(\leq k)$.
Note that each set will be copied only once.
Moreover, when it is copied, its weight and extra weight are also copied.

In the following, we use $w^*_k(e)$ and $l^*_k(e)$ to denote the weight and level of element $e$ in $H^*_k$. The other notations, e.g., $w^*_k(s)$, $\delta^*_k(s)$, $l^*_k(s)$ and $A^*_k(i)$, are defined similarly.
We call a set $s$ tight if $w^*_k(s) > \frac{c_s}{1+\epsilon}$.
We call an element \emph{tight} if it is contained in at least one tight set; \emph{slack} otherwise.

Whenever we change the weight of an element we also update the weight of sets containing it.
Specifically, if we change $w^*_k(e)$ from $(1+\epsilon)^{-i}$ to $(1+\epsilon)^{-j}$, then we increase $w^*_k(s)$ by $(1+\epsilon)^{-j} - (1+\epsilon)^{-i}$ for each $s$ containing $e$.
For ease of presentation, we do not state it explicitly in the pseudocode.

Note that every element that is moved to level $k+1$ either has weight $(1+\epsilon)^{-(k+1)}$ (which becomes active), or is tight (and stays passive at level $k+1$).
In Round-$i$, where $i = k+1,k,\ldots,1$, we identify the sets that become tight, and move all slack sets and elements to level $i-1$.
Since each slack set has weight at most $\frac{c_s}{1+\epsilon}$ and slack elements increase their weights by a factor of $(1+\epsilon)$ when moved one level down, we can guarantee that each set has weight at most $c_s$, and all sets at levels $[1,k+1]$ are tight. 
Additionally, there is no dead element, and passive elements only appear at level $k+1$.

\begin{algorithm}[H]
	\caption{\textsf{Rebuild($k$)}}
	\label{alg:rebuild_wc}
	\begin{algorithmic}[1]
		\State Initialize $E^*_k \leftarrow \emptyset$ and $S^*_k \leftarrow \emptyset$.
		\For{each $e\in A_k(\leq k)$}
		\State $E^*_k \leftarrow E^*_k \cup \{e\}$, $S^*_k \leftarrow S^*_k\cup e$, $w_k^*(e) \leftarrow (1+\epsilon)^{-(k+1)}$.
		\EndFor
		\For{each $e\in P_k(\leq k)$}
		\State $E^*_k \leftarrow E^*_k \cup \{e\}$, $S^*_k \leftarrow S^*_k\cup e$, $w_k^*(e) \leftarrow \min\{ (1+\epsilon)^{-(k+1)}, \min_{s: e\in s}\{c_s - w_k^*(s) \} \}$.
		\EndFor
		\For{each $e\in E^*_k$ and $s\in S^*_k$}
		\State $l^*_k(e) \leftarrow \perp$, $l^*_k(s) \leftarrow \perp$.
		\Comment{their levels are undecided}
		\EndFor
		\For{$i$ from $k+1$ to $1$ (Round-$i$)}
		\For{each $s\in S^*_k$} \Comment{settle the tight sets at level $i$}
		\If{$w^*_k(s) > \frac{c_s}{1+\epsilon}$}
		\State $l^*_k(s) \leftarrow i$, $S^*_k \leftarrow S^*_k\setminus\{s\}$.
		\For{each $e \in s\cap E^*_k$}
		\State $l^*_k(e) \leftarrow i$, $E^*_k \leftarrow E^*_k\setminus\{e\}$.
		\EndFor
		\EndIf
		\EndFor
		\For{each $e\in E^*_k$}
		\Comment{move slack elements one level down}
		\State $w^*_k(e) \leftarrow (1+\epsilon)^{-(i-1)}$, and update $w^*_k(s)$ accordingly for all $s$ containing $e$.
		\EndFor
		\EndFor
		\For{each $s\in S^*_k$}
		\State $l^*_k(s) \leftarrow 0$.
		\EndFor
	\end{algorithmic}
\end{algorithm}

While $H^*_k$ are being rebuilt by Scheduler($k$) in the background, each update that arrives will be handled in both local view($k$) and Scheduler($k$).
Recall that we have already described the way an update is handled in local view($k$).
Next we describe how to handle an update in Scheduler($k$).

\paragraph{Handling An Update in Rebuild($k$).}
Consider any update that arrives in Round-$i$, for $i\in [k+1]$ (for ease of argument, we consider line 1-7 part of Round-$(k+1)$).

Suppose element $e$ is deleted:
\begin{itemize}
	\item If $l^*_k(e) = \perp$, i.e., we have not decided the final level of $e$, then $e$ will simply be removed, which also decreases the weight of the sets containing $e$.
	\item If $e$ has already settled its level (at some $j\in[i,k+1]$), then we convert it into a dead element (with the same weight) at level $j$.
\end{itemize}

Suppose element $e$ is inserted:
\begin{itemize}
	\item If $l^*_k(s) = \perp$ for all $s$ containing $e$, then we set the weight of $e$ as
	\begin{equation*}
	w^*_k(e) \leftarrow \min \left\{ (1+\epsilon)^{-i}, \min_{s: e\in s}\{ c_s-w^*_k(s) \} \right\}.
	\end{equation*}
	That is, imagine that we increase $w^*_k(e)$ gradually until either (1) some set $s$ containing $e$ has weight $c_s$, or (2) $w^*_k(e) = (1+\epsilon)^{-i}$.
	In case (1), $e$ will be passive and at least one set containing $e$ becomes tight in Round-$i$.
	Hence we set $l^*_k(e) \leftarrow i$ when we identify the tight sets containing $e$.
	In case (2), $e$ is active, and its level remains undecided.
	\item If there exists $l^*_k(s) \neq \perp$ for some $s$ containing $e$, then we set $l^*_k(e) \leftarrow \max_{s: e\in s \wedge l^*_k(s)\neq \perp} \{l^*_k(s)\}$ and $e$ is passive.
	If $l^*_k(e) \geq 1$ then we set $w^*_k(e) \leftarrow 0$.
	Otherwise $l^*_k(e) = 0$, which means that $e$ arrives in Round-$0$ and all sets containing $e$ are slack. Then we set $w^*_k(e) \leftarrow \min_{s: e\in s}\{ c_s-w^*_k(s) \}$, and identify the sets containing $e$ that become tight and mark them as tight.
\end{itemize}

Note that an update arriving when Scheduler($k$) is running does not necessarily introduce a passive or dead element.
The element $e$ that is inserted or deleted may become active, or be removed immediately. 
However, if $e$ arrives in Round-$i$ and becomes passive or dead, then $l^*_k(e) \geq i$.

\paragraph{Properties of $H^*_k$.}
It is easy to check that Property~\ref{property}(a) (b) and (c) are satisfied by the construction.
Moreover, all elements in $D^*_k(\leq k+1)$ and $P^*_k(\leq k)$ are from newly arrived deletions and insertions.
Next, we show that Property~\ref{property}(d) is satisfied, i.e., $|P^*_k(\leq i)\cup D^*_k(\leq i)| < \epsilon\cdot |A^*_k(\leq i)|$.
We remark that the statement also holds for $k+1$ but we do not need it.

\begin{lemma}\label{lemma:local-element-when-scheduler-finishes}
	When Rebuild($k$) finishes, we have $|P^*_k(\leq i)\cup D^*_k(\leq i)| < \epsilon\cdot |A^*_k(\leq i)|$ for all $i \leq k$.
\end{lemma}
\begin{proof}
	Let $E^*_k(i) = A^*_k(i)\cup P^*_k(i)\cup D^*_k(i)$ be the elements at level $i$ of $H^*_k$.
	We first show that for every $i\in [k]$, it takes $(i+1)\cdot (|S^*_k(\leq i)|+f\cdot |E^*_k(\leq i)|) = O(f\cdot L\cdot |E^*_k(\leq i)|)$ time to construct the levels $[0,i]$ (recall that every element is contained in at most $f$ sets).
	To construct each level $i$, we first scan through the collection of slack sets and elements (to identify the tight ones), and then move every slack element one level down (and update the set weights).
	Since each slack set/element that is scanned in Round-$i$ eventually settles at level at most $i$, it takes $O(|S^*_k(\leq i)|+f\cdot |E^*_k(\leq i)|) = O(f\cdot |E^*_k(\leq i)|)$ time to construct level $i$.
	Since there are $i+1$ levels on or below $i$, the claim follows.
	
	Recall that Scheduler($k$) runs $k+1 = O(L)$ copies of Rebuild($k$) simultaneously. Hence it takes $O(f\cdot L^2\cdot |E^*_k(\leq i)|)$ total time to construct the levels on or below $i$.
	
	Observe that every $e\in P^*_k(\leq i)\cup D^*_k(\leq i)$ comes from an update that arrives after Round-$(i+1)$ finishes: suppose otherwise, i.e., $e$ arrives in Round-$j$ for some $j\geq i+1$, then we have $l^*_k(e) \geq i+1$, which is a contradiction. 
	
	\begin{claim}\label{claim:number-of-arriving-updates}
		At most $\frac{\epsilon}{3}\cdot |E^*_k(\leq i)|$ updates arrive while Scheduler($k$) is constructing levels $[0,i]$.
	\end{claim}
	\begin{proof}
		Since it takes $O(f\cdot L^2 \cdot |E^*_k(\leq i)|)$ time to construct levels on or below $i$ and updates arrive every $c\cdot\frac{1}{\epsilon} f\cdot L^2$ time (for some sufficiently large constant $c$), we can assume that at most $\frac{\epsilon}{3}\cdot |E^*_k(\leq i)|$ updates arrive when levels $[0,i]$ are being constructed.
	\end{proof}
		
	Since each update increases $|P^*_k(\leq i)\cup D^*_k(\leq i)|$ by at most one. We have
	\begin{equation*}
	|P^*_k(\leq i)\cup D^*_k(\leq i)| \leq \frac{\epsilon}{3}\cdot |E^*_k(\leq i)| =
	\frac{\epsilon}{3}\cdot |A^*_k(\leq i)\cup P^*_k(\leq i) \cup D^*_k(\leq i)|.
	\end{equation*}
	
	Reordering the inequality, we have $|P^*_k(\leq i)\cup D^*_k(\leq i)| \leq \frac{\epsilon}{3-\epsilon}\cdot |A^*_k(\leq i)| < \epsilon \cdot |A^*_k(\leq i)|$.
\end{proof}

Note that if $x$ updates arrive while Scheduler($k$) is running in the background, then we have $|E^*_k(\leq k)| \leq |E_k(\leq k)| + x$, where $|E_k(\leq k)|$ denote the number of elements in $H_k$ when Scheduler($k$) starts.
Hence Claim~\ref{claim:number-of-arriving-updates} implies that at most
\begin{equation*}
\frac{\epsilon}{3}\cdot |E^*_k(\leq i)|\leq \frac{\epsilon}{3(1-\frac{\epsilon}{3})}\cdot |E_k(\leq k)|\leq \frac{\epsilon}{2}\cdot |E_k(\leq k)|
\end{equation*}
updates arrive while Scheduler($k$) is running in the background, which implies Property~\ref{property}(e).

Applying the rebuild algorithm to the general framework, we obtain an algorithm with worst case $O(\frac{f}{\epsilon^4}\log^3 (Cm))$ update time.

\subsection{Improving the Worst Case Update Time} \label{sec:improved-algorithm}

In this section we present an improved rebuild algorithm that works with $\lambda = c\cdot\frac{1}{\epsilon}\cdot f\cdot L$, for some sufficiently large constant $c$, which implies a worst case update time of $O(\frac{f}{\epsilon^3}\log^2 (Cm))$.
Specifically, we present an efficient Rebuild($k$) subroutine that runs in $O(f\cdot |E_k(\leq k)|+k)$ time.
We develop our algorithm based on the Fix-Levels($k$) algorithm from~\cite{focs/BhattacharyaHN19}.
Note that we can assume without loss of generality that $O(f\cdot |E_k(\leq k)|+k) = O(f\cdot |E_k(\leq k)|)$, since otherwise we can finish the rebuild in $O(L)$ time, before the next update arrives. 
The main challenge is to commit the updates to Scheduler($k$) when the new hierarchy is being rebuilt, while guaranteeing Property~\ref{property}(d).

\subsubsection{Static Efficient Rebuild($k$)}

Recall the rebuild algorithm from~\cite{focs/BhattacharyaHN19} as follows (refer to Algorithm~\ref{alg:efficient_rebuild}).
The first phase of the new algorithm is the same as Algorithm~\ref{alg:rebuild_wc}:
we make a copy of every active and passive element in $H_k$, and put them at level $k+1$ (with appropriate weight).
Then we identify the tight sets and elements, which stay at level $k+1$.
However, instead of moving the slack elements down level by level (in which case an element might get scanned $O(k)$ times), we scan the elements in a specific order such that every element that gets scanned will have its level decided immediately.
By doing so, the update time spent on each element is $O(f)$, which implies an $O(f\cdot |E_k(\leq k)|)$ total rebuild time.

Specifically, suppose $E^*_k$ contains the slack elements (each of which has weight $(1+\epsilon)^{-(k+1)}$), we define the target level $l_T(s)$ of every slack set $s$ as follows.
If $s\cap E^*_k = \emptyset$ then $l_T(s) = 0$; otherwise $l_T(s)$ is the maximum $i$ such that
\begin{equation*}
w_k^*(s) + \left((1+\epsilon)^{-i}-(1+\epsilon)^{-(k+1)}\right)\cdot |s\cap E^*_k| > \frac{c_s}{1+\epsilon}.
\end{equation*}

Note that we have $l_T(s)\leq x$ if and only if
\begin{equation*}
w_k^*(s) + \left((1+\epsilon)^{-(x+1)}-(1+\epsilon)^{-(k+1)}\right)\cdot |s\cap E^*_k| \leq \frac{c_s}{1+\epsilon}.
\end{equation*}

When Round-$k$ begins, we have $l_T(s)\leq k$ since each slack set $s$ has weight $\leq \frac{c_s}{1+\epsilon}$.

Roughly speaking, $l_T(s)$ is the highest level $i$ at which $s$ is tight if we move $s$ and the elements it contains to level $i$.
Note that the target level of $s$ may change when the collection of slack elements changes.
It is shown in~\cite{focs/BhattacharyaHN19} that when the number of slack elements decreases, the target level of each slack set does not increase.
This is sufficient as~\cite{focs/BhattacharyaHN19} shows only an amortized update time guarantee.
However, to guarantee a worst case update time we need to commit the updates to the hierarchy that is being rebuild.
In this case the above property is not guaranteed.
Instead, we show in our algorithm that all slack sets in Round-$i$ have target level at most $i$.

\begin{algorithm}[H]
	\caption{\textsf{EfficientRebuild($k$)}}
	\label{alg:efficient_rebuild}
	\begin{algorithmic}[1]
		\State Copy elements in $A_k(\leq k)\cup P_k(\leq k)$ to $E^*_k$, and sets in $S_k(\leq k)$ to $S^*_k$.
		\State Initialize their weights as in line 2-5 of Algorithm~\ref{alg:rebuild_wc}, and set their levels to $\perp$ (undecided).
		\State Identify the tight sets and elements, assign them level $k+1$ and exclude them from $S^*_k$ and $E^*_k$.
		\For{$i$ from $k$ to $0$ (Round-$i$)}
		\Comment{settle the level of sets with target level $i$}
		\While{exists $s\in S^*_k$ with $l_T(s)=i$}
		\State $l^*_k(s) \leftarrow i$, $S^*_k \leftarrow S^*_k\setminus\{s\}$.
		\For{each $e \in s\cap E^*_k$}
		\State $w^*_k(e) \leftarrow (1+\epsilon)^{-i}$, $l^*_k(e) \leftarrow i$.
		\State $E^*_k \leftarrow E^*_k\setminus\{e\}$, and update $l_T(s')$ for all $s'$ containing $e$.
		\EndFor
		\EndWhile
		\EndFor
	\end{algorithmic}
\end{algorithm}

We maintain an array $\Gamma[0,\ldots,k]$ such that $\Gamma[i]$ points to a linked list of sets with target level $i$.
Hence in Round-$i$, we can identify a set with target level $i$ (line 5 of Algorithm~\ref{alg:efficient_rebuild}) in $O(1)$ time.

\subsubsection{Handling a Newly Arrived Update}

Next we describe the algorithm that handles an update arriving in Round-$i$.
For convenience we refer to line 1-3 of Algorithm~\ref{alg:efficient_rebuild} as Round-$(k+1)$.
Suppose element $e$ is deleted:
\begin{itemize}
	\item If $l^*_k(e) = \perp$, i.e., we have not decided the final level of element $e$, then $e$ will simply be removed.
	Note that for each $s$ containing $e$, the removal decreases the weight $w^*_k(s)$ of $s$ and possibly the target level $l_T(s)$.
	\item If $e$ has already settled its level (at some $j\in[i,k+1]$), then we convert it into a dead element (with the same weight) at level $j$.
\end{itemize}

Suppose element $e$ is inserted:
\begin{itemize}
	\item If $l^*_k(s) = \perp$ for all $s$ containing $e$, then we first try to include $e$ to $E^*_k$ with weight $(1+\epsilon)^{-(k+1)}$.
	Note that $e$ increases $w^*_k(s)$ by $(1+\epsilon)^{-(k+1)}$ and $|s\cap E^*_k|$ by one, for every $s$ containing $e$.
	If after inserting $e$ to $E^*_k$ we have $l_T(s)\leq i$ for all $s$ containing $e$ then we include $e$ in $E^*_k$ as stated above, and mark $e$ active.
	Recall that for each set $s$ we have $l_T(s) \leq i$ if
	\begin{equation*}
	w^*_k(s) + \left( (1+\epsilon)^{-(i+1)}-(1+\epsilon)^{-(k+1)} \right)\cdot |s\cap E^*_k| < \frac{c_s}{1+\epsilon}.
	\end{equation*}
	
	Otherwise we insert $e$ as a passive element with appropriate weight such that after the insertion all sets containing $e$ have target level at most $i$, and at least one of them has target level $i$.
	In this case, $e$ is not included in $E^*$, and will be passive at level $i$.
	
	Specifically, we maintain the invariant that all sets in $S^*_k$ have target level at most $i$ in Round-$i$.
	If there exists $s$ containing $e$ with $l_T(s) = i$ before the insertion then we insert $e$ as a passive element at level $i$ with $w^*_k(e) \leftarrow 0$.
	Otherwise we set $w^*_k(e)$ to be slightly larger than\footnote{Here ``slightly larger than'' means larger than but arbitrarily close to}
	\begin{equation*}
	\min_{s: e\in s} \left\{ \frac{c_s}{1+\epsilon} - w^*_k(s) - \left((1+\epsilon)^{-i}-(1+\epsilon)^{-(k+1)}\right)\cdot |s\cap E^*_k| \right\},
	\end{equation*}
	so that the maximum target level of sets containing $e$ becomes $i$.
	Set $e$ as a passive element at level $i$.
	Note that we can guarantee $w^*_k(e) > 0$ since $w^*_k(s) - \left((1+\epsilon)^{-i}-(1+\epsilon)^{-(k+1)}\right)\cdot |s\cap E^*_k| \leq \frac{c_s}{1+\epsilon}$ for all $s$ containing $e$ before the insertion.
	We can guarantee $w^*_k(e) \leq (1+\epsilon)^{-i}$ since otherwise for all $s$ containing $e$:
	\begin{align*}
	& \left(w^*_k(s) + (1+\epsilon)^{-(k+1)}\right) + \left((1+\epsilon)^{-(i+1)}-(1+\epsilon)^{-(k+1)}\right)\cdot (|s\cap E^*_k|+1) \\
	< & w^*_k(s) + \left((1+\epsilon)^{-i}-(1+\epsilon)^{-(k+1)}\right)\cdot |s\cap E^*_k| + (1+\epsilon)^{-i} < \frac{c_s}{1+\epsilon},
	\end{align*}
	which implies that $e$ should have been included in $E^*_k$ and become active.
	
	\item If there exists $l^*_k(s) \neq \perp$ for some $s$ containing $e$, then we set $l^*_k(e) \leftarrow \max_{s: e\in s \wedge l^*_k(s)\neq \perp} \{l^*_k(s)\}$ and $e$ is passive.
	If $l^*_k(e) \geq 1$ then we set $w^*_k(e) \leftarrow 0$.
	Otherwise $l^*_k(e) = 0$, and we set $w^*_k(e) \leftarrow \min_{s: e\in s}\{c_s-w^*_k(s)\} \leq 1$, and identify the sets containing $e$ that become tight and mark them as tight.
\end{itemize}

As before, if $e$ arrives in Round-$i$ and becomes passive or dead, then we have $l^*_k(e) \geq i$.

\paragraph{Correctness.}
We show that every set (and element) will be assigned a level in $[k+1]$ when Rebuild($k$) finishes.
Specifically, we show that in Round-$i$, all slack sets have level at most $i$.
Since we identify the sets with target level $i$ in Round-$i$ and settle their levels, the following lemma implies that all slack sets have target level $0$ in Round-$0$, and will have their level decided after Round-$0$.

\begin{lemma}\label{lemma:bounded-target-level}
	In Round-$i$, every set in $S^*_k$ has target level at most $i$.
\end{lemma}
\begin{proof}
	We prove that for every $i\in[0,k]$, if when Round-$i$ begins all sets in $S^*_k$ have target level at most $i$, then during Round-$i$ all sets in $S^*_k$ have target level at most $i$.
	Note that Round-$i$ ends only if there is no set with target level $i$, we can guarantee that all sets in $S^*_k$ have target level at most $i-1$ when Round-$i$ ends.
	As we have shown, when Round-$k$ begins we have $l_T(s)\leq k$ for all $s\in S^*_k$.
	Hence the above statement implies the lemma.
	
	Suppose that every set in $S^*_k$ has target level at most $i$ when Round-$i$ begins.
	We show that there does not exist any set with target level larger than $i$ during Round-$i$.
	
	Consider any set $s\in S^*_k$ in Round-$i$. 
	The target level $l_T(s)$ of $s$ changes either (1) when some element in $s$ is moved to level $i$ in Round-$i$, or (2) when an element in $s$ is deleted and removed in Round-$i$; or (3) when a new element is inserted and contained in $s$ in Round-$i$.
	\begin{itemize}
		\item In case (1), $w^*_k(s)$ increases by $(1+\epsilon)^{-i} - (1+\epsilon)^{-(k+1)}$, and $|s\cap E^*_k|$ decreases by one.
		By definition of target level, $l_T(s)$ does not increase, given that originally $l_T(s)\leq i$.
		\item In case (2), both $w^*_k(s)$ and $|s\cap E^*_k|$ decrease, and thus $l_T(s)$ does not increase.
		\item In case (3), suppose $e$ is inserted.
		If $e$ is inserted as an active element in $E^*_k$ then our algorithm guarantees that $l_T(s) \leq i$ for all $s$ containing $e$ after the insertion.
		If $e$ is inserted as a passive element with weight $0$ then $l_T(s)$ does not change.
		Otherwise $e$ is inserted as a passive element with weight larger than (but arbitrarily close to)
		\begin{equation*}
		\min_{s': e\in s'} \left\{ \frac{c_{s'}}{1+\epsilon} - w^*_k(s') - \left((1+\epsilon)^{-i}-(1+\epsilon)^{-(k+1)}\right)\cdot |s'\cap E^*_k| \right\}.
		\end{equation*}
		
		Consider any $s'$ (including $s$) containing $e$ after the insertion.
		If $|s'\cap E^*_k| = 0$, when by definition we have $l_T(s') = 0\leq i$ (before and after the insertion of $e$).
		Otherwise $|s'\cap E^*_k| \geq 1$ and we have
		\begin{align*}
		& w^*_k(e) + w^*_k(s') + ((1+\epsilon)^{-(i+1)}-(1+\epsilon)^{-(k+1)})\cdot |s'\cap E^*_k| \\
		\leq & w^*_k(e) + w^*_k(s') + ((1+\epsilon)^{-i}-(1+\epsilon)^{-(k+1)})\cdot |s'\cap E^*_k| - ((1+\epsilon)^{-i}-(1+\epsilon)^{-(i+1)}) < \frac{c_{s'}}{1+\epsilon},
		\end{align*}
		which implies $l_T(s') \leq i$.
	\end{itemize}

	Hence set $s$ have target level at most $i$ during Round-$i$, which implies the lemma.
\end{proof}

\subsubsection{Properties of the Hierarchy}

It is easy to check that Property~\ref{property}(a) and (c) are satisfied by the construction.
Next we show that Property~\ref{property}(b) is satisfied.
Recall that we set $\delta^*_k(s) = \delta_k(s)$ for all $s$ in $H^*_k$.

\begin{lemma}
	When EfficientRebuild($k$) finishes, all sets at levels $[1,k+1]$ are tight.
\end{lemma}
\begin{proof}
	Note that each set $s$ settles its level at $i$ if and only if it has target level $i$.
	By definition of target level, when we set $l^*_k(s) = i > 0$, the weight of $s$ is more than $\frac{c_s}{1+\epsilon}$.
	Moreover, by definition of target level, if we set $l^*_k(s) = i+1$ then $w^*_k(s) \leq \frac{c_s}{1+\epsilon}$.
	Hence we have $w^*_k(s) \in (\frac{c_s}{1+\epsilon}, 1]$ when we set $l^*_k(s) = i$.
	Since the weight of $s$ does not change afterwards, $s$ remains tight when EfficientRebuild($k$) finishes.
\end{proof}

Next we show that Property~\ref{property} (e) is satisfied.
By Lemma~\ref{lemma:bounded-target-level}, every element will be scanned exactly once in EfficientRebuild($k$).
Moreover, the algorithm spends $O(f)$ time to settle an element and update the data structure.
Hence immediately we have the following.

\begin{claim}\label{claim:efficient-rebuild-time}
	EfficientRebuild($k$) spends $O(f\cdot |E^*_k(\leq i)|)$ time to construct levels $[0,i]$.
\end{claim}

Specifically, construction of $H^*_k$ takes total time $O(f\cdot |E^*_k(\leq k+1)|) = O(f\cdot |E_k(\leq k)|) \leq \frac{\epsilon}{2L}\cdot |E_k(\leq k)|\cdot \lambda$, if we assume $\lambda = c\cdot\frac{1}{\epsilon}\cdot f\cdot L$, for some sufficiently large constant $c$.

\medskip

Finally, we show that Property~\ref{property}(d) is satisfied.
Recall that Scheduler($k$) runs $L$ parallel copies of EfficientRebuild($k$).

\begin{lemma}
	When EfficientRebuild($k$) finishes, for all $i \leq k$, $|P^*_k(\leq i)\cup D^*_k(\leq i)| < \epsilon\cdot |A^*_k(\leq i)|$.
\end{lemma}
\begin{proof}
	We first show that if an update that arrives in Round-$j$ creates a passive or dead element $e$, then we have $l^*_k(e) \geq j$. 
	Claim~\ref{claim:efficient-rebuild-time} implies that Scheduler($k$) takes $O(f\cdot L\cdot |E^*_k(\leq i)|)$ time to construct levels $[0,i]$.
	Since updates arrive every $\lambda = O(\frac{f}{\epsilon}\cdot L)$ time, the lemma follows immediately.
	
	Consider any update that arrives in Round-$j$ and inserts or deletes element $e$.
	
	If $e$ is deleted, then either $e$ will be removed immediately, or it will be assign to the level of some set $s$ containing $e$ whose level is already decided. In the later case we have $l^*_k(e) \geq j$.
	If $e$ is inserted, then either (1) $e$ becomes active; or (2) $e$ is assign to the level of some set $s$ containing $e$ whose level is already decided, which has level at least $j$; or (3) $e$ is given an appropriate weight such that at least one set $s$ containing $e$ has target level $i$.
	In the last case $e$ will be assigned level $j$ in Round-$j$, which completes the case analysis.	
\end{proof}